\documentclass[nonacm,acmsmall]{acmart}

\AtBeginDocument{%
  }

\setcopyright{none}

\usepackage{amsmath}
\usepackage{algorithm}
\usepackage{algpseudocode}
\usepackage{mathtools}
\usepackage{xspace}
\usepackage{hyperref}
\usepackage[capitalize]{cleveref}
\usepackage{braket}
\usepackage{caption}
\usepackage{booktabs}
\usepackage{multirow}
\usepackage{pifont}

\usepackage{mathpartir}
\usepackage{wrapfig}
\usepackage{changepage}
\usepackage{listings}
\usepackage{tablefootnote}
\usepackage{tikz}
\usetikzlibrary{quantikz2, arrows.meta, positioning}
\usepackage{enumitem}
\usepackage{subcaption}
\usepackage{cancel}
\usepackage{adjustbox}

\newcommand{\fctrl}[1]{%
  \ctrl[style={draw=gray!40, text=gray!40, solid, opacity=0.40}]{#1}%
}
\newcommand{\bluectrl}[1]{%
  \ctrl[style={draw=blue, text=blue, solid}]{#1}%
}

\newcommand{\ftarg}{%
  \targ[
    style={draw=gray!40, text=gray!40, solid, opacity=0.40}
  ]{}%
}
\newcommand{\bluetarg}{%
  \targ[
    style={draw=blue, text=blue, solid}
  ]{}%
}

\newcommand{\fgate}[1]{%
  \gate[
    style={draw=gray!40, text=gray!40, solid, opacity=0.40}
  ]{\textcolor{gray!40}{#1}}%
}

\newcommand{\bAwDep}{\textbf{Aw-Dep}\xspace}

\newcommand{\ctlarrow}{\bullet\!\!\!\to
}
\newcommand{\func}{\mathit{func}}

\newcommand{\Uncomp}{\mathit{Uncomp}}

\newcommand{\RwUn}{\textsf{RwUn}\xspace}

\newcommand{\Rwsynth}{\textsf{NSyn}\xspace}
\newcommand{\NWsynth}{\textsf{NWsynth}\xspace}
\newcommand{\Tpsynth}{\textsf{TpSynth}\xspace}

\newcommand{\Unqompx}{\textsf{Unqomp}\xspace}
\newcommand{\Reqompx}{\textsf{Reqomp}\xspace}
\newcommand{\RwUnx}{\textsf{RwUn}\xspace}
\newcommand{\TpUnx}{\textsf{TpUn}\xspace}

\newcommand{\templateborrow}{\text{template-based borrow}\xspace}
\newcommand{\rewriteborrow}{\text{rewrite-based borrow}\xspace}

\newcommand{\cH}{\mathcal{H}}

\newcommand{\cC}{\mathcal{C}}

\newcommand{\cG}{\mathcal{G}}

\newcommand{\tX}{\mathtt{X}}

\newcommand{\tZ}{\mathtt{Z}}
\newcommand{\tH}{\mathtt{H}}
\newcommand{\tS}{\mathtt{S}}
\newcommand{\tT}{\mathtt{T}}

\newcommand{\tCNOT}{\mathtt{CNOT}}
\newcommand{\tCZ}{\mathtt{CZ}}
\newcommand{\tCS}{\mathtt{CS}}
\newcommand{\tCT}{\mathtt{CT}}
\newcommand{\tToffoli}{\mathtt{CCNOT}}
\newcommand{\tMCX}[1]{\mathtt{C}^{#1}\mathtt{NOT}}
\newcommand{\tMCU}[2]{\mathtt{C}^{#1}(#2)}

\renewcommand{\ket}[1]{|{#1}\rangle}
\renewcommand{\bra}[1]{\langle{#1}|}
\renewcommand{\braket}[2]{\<{#1}|{#2}\>}
\newcommand{\ketbra}[2]{\left|#1\right\rangle\!\!\left\langle #2\right|}

\newcommand{\red}[1]{{\color{purple}#1}}
\newcommand{\blue}[1]{{\color{blue}#1}}
\newcommand{\green}[1]{{\color{teal}#1}}

\newcommand{\olp}{\overline{p}}
\newcommand{\olq}{\overline{q}}
\newcommand{\ola}{\overline{a}}
\newcommand{\old}{\overline{d}}
\newcommand{\olt}{\overline{t}}

\renewcommand{\>}{\rangle}
\newcommand{\<}{\langle}
\newcommand{\true}{\textbf{true}}
\newcommand{\false}{\textbf{false}}

\newcommand{\sem}[1]{\left \llbracket #1 \right \rrbracket}

\newcommand{\pare}[1]{\left ( #1 \right )}

\newcommand{\bluep}[1]{{\blue{[}{#1}\blue{]}}}

\newcommand{\QCa}{\mathbf{QC^a}}
\newcommand{\bskip}{\mathbf{skip}}
\newcommand{\bqif}{\mathbf{qif}}
\newcommand{\bthen}{\mathbf{then}}
\newcommand{\belse}{\mathbf{else}}
\newcommand{\bborrow}{\mathbf{borrow}}
\newcommand{\bstore}{\mathbf{store}}
\newcommand{\buse}{\mathbf{use}}
\newcommand{\qif}[3]{\bqif\ {#1}\ \bthen\ {#2}\ \belse\ {#3}}
\newcommand{\qifT}[2]{\bqif\ {#1}\ \bthen\ {#2}}
\newcommand{\borrowc}[3]{\bborrow\ {#1} := \ket{0}\ \bstore\ {#2}\ \buse\ {#3}}
\newcommand{\borrowd}[3]{\bborrow\ {#1}\ \bstore\ {#2}\ \buse\ {#3}}
\newcommand{\borrowcd}[3]{\bborrow\ {#1}\ \bluep{:=\ket{0}}\ \bstore\ {#2}\ \buse\ {#3}}
\newcommand{\borrowcc}[2]{\bborrow\ {#1} := \ket{0}\ \buse\ {#2}}
\newcommand{\borrowdd}[2]{\bborrow\ {#1}\ \buse\ {#2}}
\newcommand{\borrowccdd}[2]{\bborrow\ {#1}\ \bluep{:=\ket{0}}\ \buse\ {#2}}

\newcommand{\qfree}{{\textsf{\textbf{qfree}}}\xspace}
\newcommand{\const}{{\textsf{\textbf{const}}}\xspace}
\newcommand{\qbf}{{\textsf{\textbf{QBF}}}\xspace}
\newcommand{\bcqbf}{{\textsf{\textbf{CQBF}}}\xspace}

\newcommand{\spform}[2]{\text{{\textsf{\textbf{const}}}}[{#1}]({#2})}
\newcommand{\cqbf}[3]{\bcqbf[{#1}]\{{#2}\}({#3})}

\newtheorem{theorem}{Theorem}[section]
\newtheorem{corollary}{Corollary}[section]
\newtheorem{lemma}{Lemma}[section]
\newtheorem{remark}{Remark}[section]
\newtheorem{definition}{Definition}[section]
\newtheorem{proposition}{Proposition}[section]

\lstdefinelanguage{QCa}{
  morekeywords={skip,borrow,store,use,qif,then,else}, 
  sensitive=true,
  morecomment=[l]{//}, 
  morecomment=[s]{/*}{*/}, 
  morestring=[b]" 
}

\begin{document}

\title{Quantum Uncomputation of Clean and Dirty Ancilla Qubits}

\author{Chenke Liu}
    \orcid{0009-0009-5703-6614}
    \affiliation{
        \department{Key Laboratory of System Software}
        \institution{Institute of Software, Chinese Academy of Sciences}
		\city{Beijing}
		\country{China}
	}
    \affiliation{
        \institution{University of Chinese Academy of Sciences}
		\city{Beijing}
		\country{China}
	}
	\email{liuchenke24@mails.ucas.ac.cn}

\author{Li Zhou}
\correspondingauthor
  \orcid{0000-0002-9868-8477}
    \affiliation{
        \department{Key Laboratory of System Software}
        \institution{Institute of Software, Chinese Academy of Sciences}
		\city{Beijing}
		\country{China}
	}
\email{zhouli@iscas.ac.cn}
\email{zhou31416@gmail.com}

\author{Boning Meng}
    \orcid{0009-0006-0088-1639}
    \affiliation{
        \department{Key Laboratory of System Software}
        \institution{Institute of Software, Chinese Academy of Sciences}
		\city{Beijing}
		\country{China}
	}
    \affiliation{
    \institution{University of Chinese Academy of Sciences}
		\city{Beijing}
		\country{China}
	}
	\email{mengbn@ios.ac.cn}

\begin{abstract}

Automatic uncomputation aims to provide programming-language-level support to facilitate the correct and safe use of ancilla qubits in quantum computing, but efforts have only been made for clean ancillas, leaving dirty ancillas unexplored.  
We present a unified formalization of the uncomputation of both clean and dirty ancillas.  
For the first time, we prove that checking the existence of uncomputation is coNP-hard. 
We introduce two complementary synthesis-oriented existence checking methods: a syntax-directed static reasoning system and a rewrite-based normalization procedure (\RwUnx), together forming a top-down pipeline.
We implement \RwUnx in Qiskit. Compared with the state-of-the-art \Reqompx, \RwUnx handles practical benchmarks with complex dependencies where \Reqompx does not, and achieves higher success rates in random-circuit experiments.
\end{abstract}


\maketitle

\section{Introduction}\label{sec:introduction}

\textit{Ancillas} are auxiliary qubits introduced as temporary storage and are fundamental to quantum programming languages.
Mirroring the role of temporaries in classical programming, ancillas improve code clarity, modularity, and reusability by allowing programmers to store intermediate results and thereby decompose complex quantum operations into smaller, more readable components.  
In addition, ancillas facilitate quantum resource optimization at the compiler level: they can substantially reduce circuit size and depth~\cite{cuccaro2004new, draper2000addition, takahashi2010, Shende, sun, shoralgorithm, conditionalcleanqubits, riseofconditional, zindorf2025efficientimplementationmulticontrolledquantum, Low2024tradingtgatesdirty, quantumimplementationmininalT}.

However, unlike classical temporaries, incautious disposal of ancillas may destroy information and compromise correctness~\cite[\S4.4]{Nielsen_Chuang_2010}, owing to the fundamental principles of quantum mechanics.
\textit{Uncomputation} is the standard technique for the \emph{safe} use of ancillas~\cite{architectureuncomputation}.  
It applies additional operations to restore ancillas to their initial states, thereby enabling safe disposal without sacrificing quantum information.
While essential, uncomputation is tedious and error-prone when performed manually~\cite{gidney2015}: it can require as many quantum gates as the original computation, and naive, intuition-based attempts often fail.
This motivates the need for \emph{automatic} uncomputation, with two fundamental problems to be resolved:
\begin{enumerate}
    \item \emph{Existence}: deciding whether a given program admits a valid uncomputation, since the no-deleting theorem implies that uncomputation does not always exist;
    \item \emph{Synthesis}: constructing an (efficient) explicit uncomputation when one exists.
\end{enumerate}

\textit{Synthesis-oriented} existence checking additionally targets the gap between these two aspects---\emph{acceptable} programs (uncomputation exists) and \emph{executable} ones (with uncomputation synthesized).
As we will show, checking the existence of uncomputation is coNP-hard (coNP\footnote{coNP-hardness and coNP-completeness are defined analogously to NP-hardness and NP-completeness.} is the class of decision problems whose complements are in NP), and synthesizing an uncomputation can be computationally even more challenging. Accepting all programs that admit uncomputation would impose an unrealistic burden on the compiler.
Addressing this gap as a front-end acceptance criterion cleanly separates concerns and is more practical, while leaving the back-end to focus on code generation and optimization.

Recent work has extensively explored these problems for \emph{clean ancillas}---ancillas initialized in $|0\>$.
Silq~\cite{silq} is a language-level approach that guarantees the existence of uncomputation via type checking, and Qurts~\cite{qurts} utilizes lifetime mechanism inspired by Rust's type system.  
\Unqompx~\cite{unqomp} provides the first automatic synthesis of uncomputation for quantum circuits, while \Reqompx~\cite{reqomp} improves on \Unqompx by extending applicability and supporting recycling and reusing ancillas.  
\cite{modularuncomp} proposes an intermediate representation and modular synthesis algorithm for Silq programs.
These works can be viewed as \emph{synthesis-enabling}: Silq~\cite{silq} and Qurts~\cite{qurts} accept a subclass of circuits whose synthesis is relatively straightforward~\cite{unqomp,reqomp,modularuncomp}.
However, these approaches restrict ancilla operations to classical ones and largely rely on dependency tracking, which may limit their applicability when dependencies are complex and difficult to analyze.

Another key missing piece in existing work is support for \emph{dirty ancillas}---ancilla qubits that may start in an arbitrary, unknown state and can therefore be borrowed from any temporarily idle qubit.
This borrowing flexibility makes dirty ancillas widely used in circuit design and optimization: it can reduce the total qubit count and can also mitigate depth overhead of ancilla recycling~\cite{xu2026bonaautomaticmanagementdirty}.
However, dirty ancillas are more challenging to use correctly than clean ones---safe uncomputation must restore the ancilla to the \emph{same} unknown input state, not merely to \(|0\rangle\)---so in practice their use is often confined to a small set of expert-crafted template idioms (e.g., toggling-based encodings).
This motivates the study of \emph{automatic} uncomputation for dirty ancillas: automation can make correct use routine, lower the expertise barrier, and enable dirty-ancilla optimizations to be adopted more widely and deliver their full benefits.

\paragraph{This Paper} The overall objective of our work is to propose a synthesis-oriented existence-checking framework
supporting automatic uncomputation for both clean and dirty ancillas in quantum circuits. 
To this end, we first present a unified formalization for clean and dirty ancillas.

\paragraph{Definition of Uncomputation}
We extend the definition of uncomputation for clean ancillas~\cite{unqomp} to dirty ancillas, as informally described below. Given a quantum circuit, it \textit{computes}\footnote{
Any pure state $\ket{\psi}$ on multiple qubits can be uniquely decomposed w.r.t.\ a qubit $a$ (in the computational basis) as
\(\ket{\psi}=\ket{0}_a\ket{\psi_0}+\ket{1}_a\ket{\psi_1}\), 
where $\ket{\psi_0},\ket{\psi_1}$ (possibly unnormalized) are the states of the remaining qubits in the $\ket{0}_a$ and $\ket{1}_a$ quantum branches.
Here, we consider an arbitrary unitary circuit, so it maps arbitrary input $|x\>_a|\varphi\>$ to another pure state with decomposition $|0\>_a|\varphi'_{0x}\> + |1\>_a|\varphi'_{1x}\>$; in particular, we write $|0\>_a|\varphi'_{00}\> + |1\>_a|\varphi'_{10}\>$ for the decomposition of the output when the input is $|0\>_a|\varphi\>$ (i.e., $a$ is a clean ancilla).
}
arbitrary $|\varphi\>$ to $|0\>_a|\varphi'_{0x}\> + |1\>_a|\varphi'_{1x}\>$ given the ancilla initialized in $\ket{x}_a$ (with $|x\>=|0\>$ corresponding to the clean ancilla case). \textit{Uncompute} refers to disentangling the ancilla from the working qubits, i.e., restoring it to the original state, while preserving the states of the remaining qubits as much as possible:
\[
{\ket{x}_a\color{purple}\underline{{\color{black}\ket{\varphi}}}} 
\xmapsto{\textit{compute}}
\ket{0}_a\ket{\varphi'_{0x}} + \ket{1}_a\ket{\varphi'_{1x}} 
\xmapsto{\textit{uncompute}}
{\ket{x}_a\color{purple}\underline{{\color{black}\left( \ket{\varphi'_{00}} + \ket{\varphi'_{10}} \right)}}}.
\]
\begin{tikzpicture}[overlay,remember picture]
\color{purple}
\draw[semithick] (0.74+2.335-0.075, 0.5-0.2) to (0.74+2.335+0.075, 0.5-0.2);
\draw[semithick] (0.74+2.335, 0.5-0.2) to (0.74+2.335, 0.5-0.35);
\draw[semithick] (0.523+10.015+0.007-0.07,0.5-0.18-0.07) to (0.523+10.015+0.007, 0.5-0.18);
\draw[semithick] (0.523+10.015-0.007+0.07,0.5-0.18-0.07) to (0.523+10.015-0.007, 0.5-0.18);
\draw[semithick] (0.523+10.015,0.5-0.35) to (0.523+10.015, 0.5-0.2);
\draw[semithick] (0.74+2.335-0.01, 0.5-0.35) to (0.523+10.015+0.01, 0.5-0.35);
\node at (0.48+6.855-0.5,0.5-0.58) {{\footnotesize\textit{functionality}}};
\end{tikzpicture}
~\\[0.5em]
We briefly explain the intuition behind this extension:
\begin{itemize}
    \item \emph{Safety}: As emphasized in prior work (\cite[Thm.~5.3]{su2025borrowingdirtyqubitsquantum}), a dirty ancilla is safely used if and only if it is restored to its original state $|x\>$ after uncomputation;
    \item \emph{Functionality}: As a consequence of safety (i.e., no information is leaked to the working qubits), the functionality, identified as the map \(\ket{\varphi}\mapsto \ket{\varphi'_{0x}} + \ket{\varphi'_{1x}} \), is independent of $|x\>$, i.e., \(\ket{\varphi'_{0x}} + \ket{\varphi'_{1x}}\) should be replaced by \(\ket{\varphi'_{0y}} + \ket{\varphi'_{1y}}\) for some constant $|y\>$;
    \item \emph{Consistency}—equivalence to the clean case: It is therefore natural to choose $|0\>$ as the constant $|y\>$, so programmers can follow the same functional abstraction as in the clean ancilla setting.
\end{itemize}
Building on this extension, we formalize the unified existence problem:
\begin{changemargin}{1.5cm}{1.5cm} 
    \emph{Given a quantum circuit, decide whether uncomputation exists, or equivalently, \\ decide whether the functionality $|\varphi\>\mapsto\ket{\varphi'_{00}} + \ket{\varphi'_{10}}$ is a unitary operator}\footnote{If the functionality is a unitary $V$ and the \textit{compute} step is $U$, then the \textit{uncompute} step can be implemented simply by applying $U^\dagger$ followed by $V$, i.e., $VU^\dagger$.}.
\end{changemargin}

\paragraph{Hardness of the Existence Problem}
We show that the existence problem for \textit{reversible Boolean circuits}, a subclass of quantum circuits, is equivalent to another problem: deciding whether a reversible Boolean function remains reversible when fixing some of its inputs and discarding the corresponding outputs. We prove the latter is coNP-complete via a reduction from UNSAT~\cite{cooktheorem}, which implies that deciding the existence of uncomputation is coNP-hard.
As a byproduct, we prove that the existence problem is at least as hard as deciding the safety of uncomputation.

For reversible Boolean circuits, we also prove that any uncomputable circuit has a semantically equivalent \emph{normal form}, where operations on working qubits are performed before any modifications to ancillas, yielding a \emph{prefix--suffix} structure and enabling uncomputation by directly inverting the ancilla modifications.

\paragraph{Language and a Lightweight Static Reasoning System}

We introduce $\QCa$, a simple quantum circuit language featuring an explicit \emph{borrow-ancilla} construct. Its denotational semantics is a unitary when uncomputation exists for all borrowed ancillas, and is invalid 
otherwise. Inspired by the type-based disciplines for clean ancillas in Silq~\cite{silq} and Qurts~\cite{qurts}, we develop a syntax-directed static reasoning system that checks validity of a program in $\QCa$ with both ancillas, and also guarantees a regular synthesis if checked.
This is achieved by introducing intrinsic properties that support reasoning about dirty-ancilla scopes. In particular, the system captures the toggle-detection pattern--encoding explicit bit information via the occurrence of an ancilla flip--and therefore covers many practical usage patterns. This lightweight method trades applicability for efficiency and serves as a first-pass filter. When it fails, the normal form suggests a more promising path.

\paragraph{Rewrite-Based Normalization}
Guided by the observation of normal-form structure, we develop a practical and terminating rewrite-based normalization procedure, partly inspired by SPARE~\cite{spare}.
We introduce rewrite rules that, whenever a pair of gates violates the normal-form order, push ancilla-modifying operations rightward to restore the normal-form order locally.
Iterating these local rewrites eventually concentrates all ancilla-modifying gates at the circuit tail, yielding the desired normal-form structure. We further extend the normal form for general quantum circuits in a way that is sufficient--but not necessary--to guarantee existence and enable regular synthesis. This extension comes with additional rewrite rules for $\tH, \tZ, \tS,$ and $\tT$ gates; in particular, non-classical phase gates (single-qubit $Z$-rotations) on ancillas can be handled directly.

For reversible Boolean circuits, the applicability of our method is limited by the presence of \emph{aw-cycles} (i.e., directed cycles between ancilla and working qubits, indicating mutual dependency). For general quantum circuits, our method is further limited by the presence of $\tH$ gates: they are disallowed on ancillas and on working qubits that ancillas depend on.

\paragraph{Implementation and Evaluation}
We implemented our rewrite-based normalization (\RwUnx) as a Qiskit~\cite{qiskit} plugin and compared it with \Reqompx~\cite{reqomp}---the state-of-the-art uncomputation synthesis strategy for circuits.  
We examine 21 practical benchmarks with varying dependency complexity: \RwUnx succeeds on 19/21, whereas \Reqompx succeeds on 14/21. The seven benchmarks handled by \RwUnx but not \Reqompx tend to exhibit higher dependency complexity. The two benchmarks handled by \Reqompx but not \RwUnx are rejected because their controlled-rotation gates fall outside the gate set covered by our rewrite rules.

We also test on random circuits: for random reversible Boolean circuits generated by the universal set $\{\tX, \tCNOT, \tToffoli\}$, we solve roughly twice as many instances as \Reqompx; for random quantum circuits generated by the universal gate set $\{\tX, \tCNOT, \tToffoli, \tZ, \tH, \tS, \tT \}$ but forbidding $\tH$ gates on ancillas---beyond the scope of \Reqompx and other existing methods---our approach still succeeds in about 60\% of cases.

\paragraph{Structure and Contributions}
Section~\ref{sec:preliminary} introduces the basic concepts of quantum computing and the role of ancillas. The subsequent sections present our main contributions:
\begin{itemize}
    \item Section~\ref{sec:hardness} formalizes the unified framework for both clean and dirty ancillas and proves the hardness of deciding the existence of uncomputation.
    \item Section~\ref{sec:language} presents our extended quantum circuit language $\QCa$ with a borrow statement, and Section~\ref{sec:template} develops a syntax-directed static reasoning system for $\QCa$.
    \item Section~\ref{sec:rwun} defines the normal form and develops the rewrite-based normalization.
    \item Section~\ref{sec:evaluation} describes our implementation and experimental evaluation.
\end{itemize}
Finally, Section~\ref{relatedwork} discusses related work.

\section{Ancillas in Quantum Computing: Background and Motivation}
\label{sec:preliminary}

We review the background of quantum computing, focusing on the use of ancillas at the circuit level.
The distinction between clean and dirty ancillas and their impact on resource optimization and uncomputation provide the motivation for the techniques developed in subsequent sections.

\subsection{Quantum Computing}

We provide a brief overview of the essential concepts in quantum computing; see \cite{Nielsen_Chuang_2010} for details.

\paragraph{Quantum States.}
A qubit serves as the quantum analogue of a classical bit. A single-qubit state is represented by a unit vector in the two-dimensional Hilbert space $\mathbb{C}^2$, typically expressed in ket notation. The computational basis of $\mathbb{C}^2$ is given by
$
\ket{0} = \bigl(\!\begin{smallmatrix}1\\0\end{smallmatrix}\!\bigr)
$
and
$
\ket{1} = \bigl(\!\begin{smallmatrix}0\\1\end{smallmatrix}\!\bigr)
$.
An arbitrary qubit state can be written as a linear combination of the computational basis:
$
\ket{\psi} = \alpha\ket{0} + \beta\ket{1} = \bigl(\!\begin{smallmatrix} \alpha \\ \beta \end{smallmatrix}\!\bigr),
$
with complex amplitudes $\alpha, \beta \in \mathbb{C}$ satisfying the normalization condition $|\alpha|^2 + |\beta|^2 = 1$.

An $n$-qubit state lives in the tensor-product space $\mathbb{C}^{2^{n}}$, where $n$ is the number of qubits. For example, the Hilbert space of a two-qubit system is $\mathbb{C}^4$, with computational basis
\[
\ket{00} = \ket{0} \otimes \ket{0}, \quad 
\ket{01} = \ket{0} \otimes \ket{1}, \quad 
\ket{10} = \ket{1} \otimes \ket{0}, \quad 
\ket{11} = \ket{1} \otimes \ket{1},
\]
where $\otimes$ denotes the Kronecker product.
It is common to use subscripts to indicate which qubit(s) a state corresponds to; for example, $\ket{\phi}_q\ket{\psi}_{\overline{a}}$ denotes that qubit $q$ is in state $\ket{\phi}$ and qubits $\overline{a}$ (we use an overline to denote a list of distinct qubits) are in state $\ket{\psi}$.

Not all states in the tensor-product space can be expressed as simple tensor products of individual qubit states. Such non-separable states are referred to as \emph{entangled states}. A canonical example is the Bell state
$
\ket{\Phi^+} = \frac{1}{\sqrt{2}}(\ket{00} + \ket{11}).
$

\paragraph{Unitary Transformation.}
The evolution of a closed quantum system is described by a \emph{unitary transformation}---a linear operator $U$ satisfying $U^\dagger U = UU^\dagger = I$, where $U^\dagger$ denotes the conjugate transpose.
A linear operator is unitary if and only if it preserves the norm of quantum states.
When a unitary transformation acts on one or more qubits, it is commonly referred to as a \emph{gate}. The following are some common single-qubit and two-qubit gates:
\[
\tX = \begin{pmatrix} 0 & 1 \\ 1 & 0 \end{pmatrix}, \quad 
\tH = \frac{1}{\sqrt{2}}\begin{pmatrix} 1 & 1 \\ 1 & -1 \end{pmatrix}, \quad
\tT = \begin{pmatrix} 1 & 0 \\ 0 & e^{i\pi/4} \end{pmatrix}, \quad
\tCNOT = \left(\begin{smallmatrix} 1 & 0 & 0 & 0 \\ 0 & 1 & 0 & 0 \\ 0 & 0 & 0 & 1 \\ 0 & 0 & 1 & 0 \end{smallmatrix}\right).
\]

\emph{Quantum circuits} formalize the compositional structure of quantum gate applications, providing an intuitive graphical representation of unitary transformations.  
For example, the following circuit models a computation on a two-qubit system initialized in $\ket{00}$: two Hadamard ($\tH$) gates are first applied to the qubits, followed by a two-qubit $\tCNOT$ gate.
\begin{center}
\small
\begin{quantikz}[row sep = 0.2cm]
    \lstick{$\ket{0}$} & \gate{\tH} & \ctrl{1} & \qw \\
    \lstick{$\ket{0}$} & \gate{\tH} & \targ{}  & \qw
\end{quantikz}
\end{center}

\paragraph{Multi-controlled Gates.}
An important class of multi-qubit operators is \emph{controlled gates}. 
A multi-controlled gate $\tMCU{k}{U}$, 
where the unitary transformation $U$ on the \emph{target} qubit(s) is conditioned on the states of $k$ \emph{control} qubits, is defined as
\[
\tMCU{k}{U}(\ket{c_1,\dots,c_k}\ket{t}) = 
\begin{cases}
    \ket{c_1,\dots,c_k}(U\ket{t}) & \text{if } c_1 \wedge \dots \wedge c_k = 1,\\
    \ket{c_1,\dots,c_k}\ket{t} & \text{otherwise.}
\end{cases}
\]
When $U$ is fixed as the NOT gate $\tX$, i.e., $\tMCU{k}{\tX}$, the gate is commonly referred to as a \emph{multi-controlled NOT (MCX)} and denoted by $\tMCX{k}$.
It acts on $k$ control qubits $c_1,\dots,c_k$ and a single target qubit $t$:
\[
\tMCX{k} \ket{c_1,\dots,c_k, t} = 
\ket{c_1,\dots,c_k,\, t \oplus (c_1 \wedge \cdots \wedge c_k)},
\]
where $c_1,\dots,c_k,t\in \{0, 1\}$ and $\oplus$ denotes XOR.
In other words, the target qubit flips if and only if all controls are $\ket{1}$. 
The simplest example is the controlled-NOT gate $\tMCX{1}$, usually denoted by $\tCNOT$, whose matrix form is given above, or equivalently, $\tCNOT\ket{c,t} = \ket{c, c\oplus t}$.
The \emph{Toffoli gate} $\tMCX{2}$, commonly denoted by $\tToffoli$, with mapping
\(
\tToffoli \ket{c_1,c_2,t} = \ket{c_1, c_2, t \oplus (c_1 \wedge c_2)},
\)
forms a \emph{universal} gate for reversible Boolean circuits.

\subsection{Safe Use of Ancillas}\label{subsec:safeusecase}

A clean ancilla is an ancilla qubit with a known initial state (by convention, $\ket{0}$), whereas a dirty ancilla is an ancilla qubit with an arbitrary, unknown initial state.  
As discussed earlier, ancillas require additional checks to ensure their correct and safe use compared to classical auxiliary variables.  
Formally, given a circuit that involves ancillas $\overline{a}$ and represents a unitary operator $U$, we say that $U$ \emph{safely uses clean ancillas} $\overline{a}$ if 
\begin{equation}
\label{eq:safe-clean}
  \forall\,\ket{\phi},\ U(\ket{0}_{\overline{a}}\ket{\phi}) = \ket{0}_{\overline{a}}\ket{\psi} \text{ for some } \ket{\psi},
\end{equation}
and that $U$ \emph{safely uses dirty ancillas} $\overline{a}$ if 
\begin{equation}
\label{eq:safe-dirty}
\begin{aligned}
  &\forall\,\ket{x},\ket{\phi}, \ 
  U(\ket{x}_{\overline{a}}\ket{\phi}) = \ket{x}_{\overline{a}}\ket{\psi}
  \text{ for some } \ket{\psi}, \\
  &\text{or equivalently, } 
  U = I_{\ola} \otimes V \text{ for some unitary } V.
\end{aligned}
\end{equation}
For clean ancillas initialized in $\ket{0}_{\ola}$, since only unentangled qubits can be safely discarded, mapping $\ket{0}_{\ola}$ back to $\ket{0}_{\ola}$ ensures safe disposal~\cite[\S4.4]{Nielsen_Chuang_2010}.
For dirty ancillas initialized in an arbitrary state $|x\>_{\ola}$, there are several equivalent formulations of safe use~\cite{su2025borrowingdirtyqubitsquantum}.
Let us explain the idea behind these definitions with a case of implementation of $\tMCX{3}$ gate on $(q_1, q_2, q_3, t)$ using only $\tToffoli$ gates.

\paragraph{Design `Compute'---Transformation except Ancillas}
Since $\tMCX{3}$ maps $\ket{q_1,q_2,q_3}\ket{t}$ (marked in blue) to $\ket{q_1,q_2,q_3}\ket{t \oplus (q_1 \land q_2) \land q_3}$ (marked in red), it is not difficult to implement it by first using an $\tToffoli[q_1,q_2,a]$ to store $q_1\wedge q_2$ in the clean ancilla $a$ (marked in green), and then using another $\tToffoli[a,q_3,t]$ to transform $t$ to $t \oplus (q_1 \land q_2) \land q_3$.
\[
\small
\begin{array}{c@{\qquad}l}
\vcenter{\hbox{
\begin{quantikz}[row sep = 0.2cm]
    \lstick{$q_1$}          & \ctrl{2}   & \qw    & \ghost{}\qw     \\
    \lstick{$q_2$}          & \ctrl{1}   & \qw      & \qw     \\
    \lstick{$a := \ket{0}$} & \targ{}    & \ctrl{2} & \qw     \\
    \lstick{$q_3$}          & \qw        & \ctrl{1} & \qw     \\
    \lstick{$t$}            & \qw        & \targ{}  & \qw
\end{quantikz}
}}
&
\begin{aligned}
&\blue{\ket{q_1,q_2}}\ket{0}_a\blue{\ket{q_3}\ket{t}} \\
\xmapsto{\tToffoli[q_1, q_2, a] }\;&
  \ket{q_1,q_2}\green{\ket{0 \oplus q_1 \land q_2}_a}\ket{q_3}\ket{t} \\
\xmapsto{\tToffoli[a, q_3, t]\ \,}\;&
  \red{\ket{q_1,q_2}}\ket{q_1 \land q_2}_a\red{\ket{q_3}\ket{t \oplus (q_1 \land q_2) \land q_3}}.
\end{aligned}
\end{array}
\]

However, such a step is not safe or correct if the input is in a superposition such as $\frac{1}{\sqrt{2}}(|000\>|0\> + |111\>|1\>)$, which after the compute is in state
\(
\frac{1}{\sqrt{2}}(|000\>|0\>_a|0\> + |111\>|1\>_a|0\>),
\)
and due to implicit measurement~\cite[\S4.4]{Nielsen_Chuang_2010}, discarding $a$ yields state $|000\>|0\>$ with probability $0.5$ and $|111\>|0\>$ with probability $0.5$, a probabilistic combination of $|000\>|0\>$ and $|111\>|0\>$, rather than the superposition state $\tMCX{3}\left(\frac{1}{\sqrt{2}}(|000\>|0\> + |111\>|1\>)\right) = \frac{1}{\sqrt{2}}(|000\>|0\> + |111\>|0\>)$.

\paragraph{Uncompute for Clean Ancillas}
To make the result correct for superposition, an additional uncomputation $\tToffoli[q_1,q_2,a]$ is required, which restores the ancilla $a$ to state $|0\>$ (marked in green):
\[
\small
\begin{array}{c@{\qquad}l}
\vcenter{\hbox{
\begin{quantikz}[row sep = 0.3cm]
    \lstick{$q_1$}          & \ctrl{2}\gategroup[wires=5,steps=2,style={dashed,rounded corners,draw=blue, fill=none, inner xsep=2pt},label style={label position=above,anchor=north,xshift=-0.2cm,yshift=0.4cm}]{{\rm \blue{compute}}}   & \qw      & \ctrl{2}\gategroup[wires=5,steps=1,style={dashed,rounded corners,draw=red, fill=none, inner xsep=2pt},label style={label position=above,anchor=north,xshift=+0.2cm,yshift=0.4cm}]{{\rm \red{uncompute}}} & \ghost{}\qw     \\
    \lstick{$q_2$}          & \ctrl{1}   & \qw      & \ctrl{1} & \qw     \\
    \lstick{$a := \ket{0}$} & \targ{}    & \ctrl{2} & \targ{}  & \qw     \\
    \lstick{$q_3$}          & \qw        & \ctrl{1} & \qw      & \qw     \\
    \lstick{$t$}            & \qw        & \targ{}  & \qw      & \qw
\end{quantikz}
}}
&
\begin{aligned}
&\blue{\ket{q_1,q_2}}\ket{0}_a\blue{\ket{q_3}\ket{t}} \\
\xmapsto{\blue{\tToffoli[q_1, q_2, a]} }\;&
  \ket{q_1,q_2}\ket{0 \oplus q_1 \land q_2}_a\ket{q_3}\ket{t} \\
\xmapsto{\blue{\tToffoli[a, q_3, t]}\ \,}\;&
  \ket{q_1,q_2}\ket{q_1 \land q_2}_a\ket{q_3}\ket{t \oplus (q_1 \land q_2) \land q_3} \\
\xmapsto{\red{\tToffoli[q_1, q_2, a]} }\;&
  \red{\ket{q_1,q_2}}\green{\ket{0}_a}\red{\ket{q_3}\ket{t \oplus (q_1 \land q_2) \land q_3}}.
\end{aligned}
\end{array}
\]
Combining the "compute" and "uncompute" parts finally yields a safe use of clean ancilla $a$.

However, the above circuit is still not satisfied for dirty ancillas since (1) it yields a wrong state when the ancilla is not initialized in $|0\>$, and (2) ancilla is entangled with other qubits if it is initialized in a superposition state. For example, if the input is $|111\>|1\>$ with dirty ancilla $a$ in $|1\>$, the final state is 
$|111\>|1\>_a|1\>$, or $|111\>|1\>$ after discarding ancilla, which is different from
the desired state $\tMCX{3}(|111\>|1\>) = |111\>|0\>$.

\paragraph{Additional Compute for Dirty Ancillas}
To fix this, an additional $\tToffoli[a,q_3,t]$ (marked in green) is required at the beginning of the circuit. The two compute parts together form a \textit{toggle detection} pattern, the standard technique for dirty-ancilla usage, which yields a correct state of working qubits and restores the ancilla to the original state without further entanglement:
\[
\small
\begin{array}{c@{\qquad}l}
\vcenter{\hbox{
\begin{quantikz}[row sep = 0.2cm]
    \lstick{$q_1$}          & \qw  & \qw\gategroup[wires=5,steps=1,style={dashed,rounded corners,draw=teal, fill=none, inner xsep=2pt},label style={label position=above,anchor=north,xshift=-0.2cm,yshift=0.4cm}]{{\rm \green{compute}}}        & \ctrl{2}\gategroup[wires=5,steps=2,style={dashed,rounded corners,draw=blue, fill=none, inner xsep=2pt},label style={label position=above,anchor=north,yshift=0.4cm}]{{\rm \blue{compute}}}   & \qw     & \ctrl{2}\gategroup[wires=5,steps=1,style={dashed,rounded corners,draw=red, fill=none, inner xsep=2pt},label style={label position=above,anchor=north,xshift=0.3cm,yshift=0.4cm}]{{\rm \red{uncompute}}} & \ghost{}\qw     \\
    \lstick{$q_2$}          & \qw  & \qw        & \ctrl{1}   & \qw      & \ctrl{1} & \qw     \\
    \lstick{$a$}            & \qw  & \ctrl{2}   & \targ{}    & \ctrl{2} & \targ{}  & \qw     \\
    \lstick{$q_3$}          & \qw  & \ctrl{1}   & \qw        & \ctrl{1} & \qw      & \qw     \\
    \lstick{$t$}            & \qw  & \targ{}    & \qw        & \targ{}  & \qw     & \qw  
\end{quantikz}
}}
&
\hspace{-0.5cm}
\begin{aligned}
&\blue{\ket{q_1,q_2}}\green{\ket{a}_a}\blue{\ket{q_3}\ket{t}}\\
\xmapsto{\green{\tToffoli[a, q_3, t]}\ \, }\;&
  \ket{q_1,q_2}\ket{a}_a\ket{q_3}\ket{t \oplus a \land q_3} \\
\xmapsto{\blue{\tToffoli[q_1, q_2, a]}}\;&
  \ket{q_1,q_2}\ket{a \oplus q_1 \land q_2}_a\ket{q_3}\ket{t \oplus a \land q_3} \\
\xmapsto{\blue{\tToffoli[a, q_3, t]}\ \, }\;&
  \ket{q_1,q_2}\ket{a \oplus q_1 \land q_2}_a\ket{q_3}\ket{t \oplus q_1 \land q_2 \land q_3} \\
\xmapsto{\red{\tToffoli[q_1, q_2, a]} }\;&
  \red{\ket{q_1,q_2}}\green{\ket{a}_a}\red{\ket{q_3}\ket{t \oplus q_1 \land q_2 \land q_3}}\\
\end{aligned}
\end{array}
\]
Therefore, combining the two "compute" and "uncompute" parts gives a safe use of dirty ancilla $a$.

\subsection{Why Dirty Ancillas in Programming Languages}

While dirty ancillas have been widely explored in circuit design and optimization, relatively little attention has been paid from a programming language perspective. To this end, the subsection is devoted to explaining why it is appealing to introduce dirty ancillas in a programming language.

\paragraph{Unique Potential of Dirty Ancillas}
Figure~\ref{fig:dirty-utility} summarizes the costs in circuit design using clean ancillas or dirty ancillas. 
While both types of ancillas can help reduce cost, dirty ancillas offer several advantages over clean ones:
(1) they provide unique flexibility in management while keeping the circuit depth asymptotically at the same order as in the clean case;
(2) they achieve better performance in circuit size and depth when the circuit width is fixed~\cite{quantumimplementationmininalT} or preferentially constrained; and
(3) they enable reductions in both $T$-gate depth and count in unitary synthesis.

\begin{figure}[tbp]
    \centering
    \resizebox{\textwidth}{!}{
        \begin{tabular}{c|c|cccc}
            \cmidrule{1-6}\morecmidrules\cmidrule{1-6}
            \multirow{6}{*}{
            \begin{tabular}{@{}c@{}}
            \textbf{Multi-} \\ \textbf{controlled} \\ \textbf{NOT}\cite{table-n-CNOT}
            \end{tabular}}
                                          &\ \ Ref.\ \            & \# Clean Ancilla      & \# Dirty Ancilla                                                & Circuit Depth                                    &                                                  \\
            \cmidrule{2-5}
                                          & \cite{gidney2015}      & 0                  & 0                                                            & $494n-1413$                                               &                                                  \\
                                          & \cite{elementgates} & 0                  & $n-2$                                                          & $24n-43$                                           &                                                  \\
                                          & \cite{elementgates} & $n-2$                & 0                                                            & $12n-12$                                           &                                                  \\
                                          & \cite{table-n-CNOT}     & 0                  & 1                                                            & $43\log(n)^3-1287$               &                                             \\
                                          & \cite{table-n-CNOT}     & 1                  & 0                                                            & $27\log(n)^3-808$                  &                                                  \\
            \cmidrule{1-6}\morecmidrules\cmidrule{1-6}
            \multirow{6}{*}{
            \begin{tabular}{@{}c@{}}
            \textbf{Constant} \\ \textbf{Adder}\cite{factoring2n+2}
            \end{tabular}}
                                          &\ \ Ref.\ \             & \# Clean Ancilla      & \# Dirty Ancilla                                                & Circuit Depth                                    & Circuit Size                                     \\
            \cmidrule{2-6}
                                          & \cite{draper2000addition}      & 0                  & 0                                                            & $\varTheta(n)$                                                & $\varTheta(n^2)$                            \\
                                          & \cite{cuccaro2004new}     & $n+1$                & 0                                                            & $\varTheta(n)$                                             & $\varTheta(n)$                                                \\
                                          & \cite{takahashi2010}   & $n$                  & 0                                                            & $\varTheta(n)$                                                & $\varTheta(n)$                                               \\
                                          & \cite{factoring2n+2}       & 0                  & 1                                                            & $\varTheta(n)$                                                & $\varTheta(n\log n)$                                        \\
                                          & \cite{Remaud2025}          & 0                  & 1                                                            & $O(\log^3 n)$                                                & $O(n\log^2 n)$                                        \\
            \cmidrule{1-6}\morecmidrules\cmidrule{1-6}
            \multirow{4}{*}{
            \begin{tabular}{@{}c@{}}
            \textbf{Factoring}\cite{factoringn+2}
            \end{tabular}}
                                          &\ \ Ref.\ \             & \# Clean Qubits & \# Dirty Ancilla                                                & Circuit Depth                                    & Circuit Size                                     \\
            \cmidrule{2-6}
                                          & \cite{shoralgorithm}        & $\varTheta(n)$               & 0                                                            & $\varTheta(nM(n))$                                         & $\varTheta(nM(n))$                                         \\
                                          & \cite{zalka}       & $1.5n+O(1)$          & 0                                                            & $\varTheta(n^3\log \frac{1}{\epsilon})$               & $\varTheta(n^3\log \frac{1}{\epsilon}\log \frac{n}{\epsilon})$    \\
                                          & \cite{factoringn+2}      & $n+2$                & $n-1$                                                          & $\varTheta(n^3)$                          & $\varTheta(n^3\log n)$                      \\
            \cmidrule{1-6}\morecmidrules\cmidrule{1-6}           
            \multirow{4}{*}{
            \begin{tabular}{@{}c@{}}
            \textbf{Unitary} \\ \textbf{Synthesis}\cite{Low2024tradingtgatesdirty}
            \end{tabular}} 
                                          &\ \ Ref.\ \             & \# Total Qubits & \# Dirty Ancilla                                                & $T$ Gate Depth                                     & $T$ Gate Count                                     \\
            \cmidrule{2-6}                             
                                          & \cite{Shende}      & $\log n$               & 0                                                            & $n^2\log(\frac{n}{\epsilon})$ & $n^2\log(\frac{n}{\epsilon})$ \\
                                          & \cite{sun}         & $\log n +\lambda$               & $\lambda=O(\frac{n}{\log n\log\log n})\cup\Omega(n)$ & $\frac{n^2}{\log n+\lambda}\log\pare{\frac{n}{\epsilon}}+n\log n\log(\frac{n}{\epsilon})$ & $n^2\log(\frac{n}{\epsilon})$ \\
                                          & \cite{Low2024tradingtgatesdirty}            & $\log n+\lambda$                   &    $\lambda=[\log\frac{K}{\epsilon'},n\log\frac{K}{\epsilon'}]$                                                          &     $\frac{Kn}{\lambda}\log\frac{n}{\epsilon'}+K\log n\log\frac{K\lambda}{\epsilon'}$                                             &        $\lambda K\log\frac{n}{\epsilon}+\frac{Kn}{\lambda}\log\frac{n}{\epsilon'}$                                          \\
            \cmidrule{1-6}\morecmidrules\cmidrule{1-6}
            \end{tabular}}
    \caption{
    Circuit costs with and without dirty ancillas.
    The depths in \cite{table-n-CNOT} are numerically fitted for $10^2$--$10^7$ control qubits, and $M(n)$ in \cite{factoringn+2} is known to be asymptotically at most $n \cdot (\log n)\cdot 2^{O(\log^* n)}$ \cite{faster-integer}.}
    \Description{A table comparing clean-qubit, dirty-ancilla, depth, and size costs of representative constructions for multi-controlled NOT, constant addition,
    factoring, and unitary synthesis. The comparisons show the resource tradeoffs
    enabled by dirty ancillas.
    }
    \label{fig:dirty-utility}
\end{figure}

\paragraph{Leveraging Flexibility in Compilation}
Flexibility—borrowed from any idle qubits in the system—is often regarded as the greatest advantage of dirty ancillas.
Recent work~\cite{xu2026bonaautomaticmanagementdirty} has revealed a unique benefit for optimizing highly parallel high-level circuits (including multi-qubit gates from non-elementary gate sets): when parallelism increases, dirty-ancilla designs can maintain constant depth without requiring additional ancillas, whereas clean-ancilla designs either demand a linear number of ancillas or a linear depth.
Take \cref{fig:why_dirty} as an example of four-way parallelism.
Supporting dirty ancillas at the programming-language level enables automated management of borrowing, which alleviates the programmer's burden, ensures the correctness of qubit reuse, and facilitates efficient scheduling for large-scale circuit synthesis.

\begin{figure}[t]
  \centering
  \vspace{-0.5\baselineskip} 
  \includegraphics[width=0.95\textwidth]{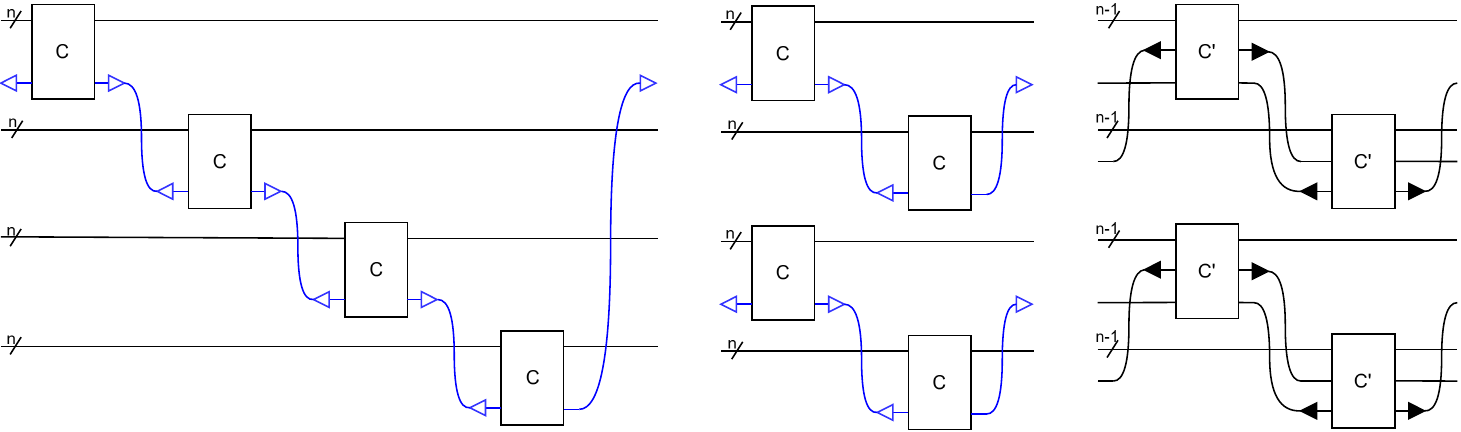}
  \caption{For applying four identical functionalities in parallel on four disjoint $n$-qubit blocks (where $C$ and $C'$ are different implementations using a clean and a dirty ancilla, respectively) (left to right): (i) linear depth with a single clean ancilla (allocation $\triangleleft$, free $\triangleright$), (ii) a linear number of clean ancillas ($2=\frac{1}{2}\times 4$) with constant doubled depth, and (iii) no additional ancilla needed with constant doubled depth (allocation $\blacktriangleleft$, free $\blacktriangleright$). A free followed by an allocation indicates reuse: clean ancillas can only reuse clean ancillas, whereas dirty ancillas may reuse working qubits. }
  \Description{
    A three-panel diagram comparing ways to apply the same functionality to four
    disjoint n-qubit blocks. The left panel reuses a single clean ancilla across
    the four blocks, resulting in linear depth. The middle panel uses two clean
    ancillas and achieves twice the depth of the original implementation. The
    right panel borrows working qubits as dirty ancillas, requiring no additional
    ancilla qubits while also incurring twice the original depth.
    }
  \label{fig:why_dirty}
  \vspace{-0.5\baselineskip} 
\end{figure}

\paragraph{Automated Checking of Safe Use of Dirty Ancillas}
Fully unleashing the advantages of dirty ancillas requires ensuring their safe use---a significant challenge and burden for designers, as previously shown.
Providing language-level support for dirty ancillas, together with synthesis-oriented automated checking, can substantially lower the barrier to their use.
As demonstrated in this work, for most cases, programmers only need to focus on designing the compute part, which makes the use of dirty ancillas as easy as working with temporaries.

\section{Uncomputation Problem and Hardness Results}\label{sec:hardness}

In this section, we will formalize the uncomputation problem, i.e., the existence of uncomputation, for clean and dirty ancillas, and then show that the uncomputation problem is coNP-hard.

\subsection{Existence of Uncomputation}\label{sec:existence-of-uncomputation}

We start by extracting the functionality of a given computation, as illustrated in~\cref{subsec:safeusecase}. 
Let $G$ be a unitary operator acting on $m$ ancilla qubits $\ola$ and $n$ working qubits $\olq$, with the mapping:
\begin{equation}
    \ket{0}_{\ola} \ket{\varphi}_{\olq}
    \xmapsto{G}
    \sum_{i=0}^{2^m-1} \ket{i}_{\ola}
    \ket{\varphi_{i}}_{\olq}.
\end{equation}
We define $\func(G,\ola)$, the functionality of $G$, as the mapping:
\[
    \func(G,\ola) : \ket{\varphi}_{\olq }\mapsto \sum_{i=0}^{2^m-1} \ket{\varphi_{i}}_{\olq}.
\]
Due to no-deleting theorem, the functionality is not always valid, i.e., the mapping $\func(G,\ola)$ might not be a unitary operator. For example, while $\frac{1}{\sqrt{2}} (\ket{0}_a \ket{0}_q + \ket{1}_a \ket{0}_q )$ is a valid quantum state, the state $\ket{0}_a (\tfrac{1}{\sqrt{2}} \ket{0}_q + \frac{1}{\sqrt{2}} \ket{0}_q)$ has norm $\sqrt{2}$ and therefore $\func(G,\ola)$ is not norm-preserving, i.e., not a unitary operator. Formally, we ask:
\begin{definition}[Validity of functionality]
    Decide whether $\func(G,\ola)$ is a unitary operator.
\end{definition}

Deciding the existence of uncomputation is closely related to deciding the validity of functionality, as we require uncomputation exactly achieves this functionality while keeping ancillas unchanged. Formally, we follow the definition first given in \Unqompx \cite{unqomp}:
\begin{definition}[Uncomputation problem \cite{unqomp}]\label{def:cleanuncomp}
    The uncomputation problem asks whether there exists a unitary operator $\cG_{\ola}$ that restores $\ola$ while preserving the state of $\olq$:
    \begin{equation}
        \forall\,|\varphi\>,\ \ket{0}_{\ola} \ket{\varphi}_{\olq }
        \xmapsto{\cG_{\ola}}
        \ket{0}_{\ola} \left(\sum_{i=0}^{2^m-1} \ket{\varphi_i}_{\olq}\right).
    \end{equation}
    In the remainder of the paper, we simply write $\cG$ instead of $\cG_{\ola}$ whenever $\ola$ is clear from the context.
\end{definition}

Note that $\cG_{\ola}$ may not exist also due to no-deleting theorem, and is \emph{not unique} if it exists. 
We say that $G$ is \emph{uncomputable}---that is, the \emph{uncomputability} of $G$---if $\Uncomp(G, \overline{a})$, i.e., the set of all $\cG_{\ola}$, is non-empty.
It is straightforward to observe that any $\cG_{\ola} \in \Uncomp(G, \ola)$ safely uses clean ancillas $\ola$ (recall~\cref{eq:safe-clean}), and correctly implements the functionality $\func(G,\ola)$ on $\olq$. 

As explained in~\Cref{sec:introduction}, the functionality of $G$ with dirty ancillas $\ola$ should be independent of the initial state of $\ola$, so it is natural to set it as the same as $\func(G,\ola)$. Now, we extend Definition~\ref{def:cleanuncomp} to the dirty case:
\begin{definition}[Uncomputation problem (dirty ancillas)]\label{def:dirtyuncomp}
    The uncomputation problem for dirty ancillas asks whether there exists a unitary operator $\cG_{\ola}^*$, such that for any input $\ket{x}_{\ola}$ the register $\ola$ is restored while the state of $\olq$ is preserved:
    \begin{equation}
        \forall\,\ket{x},\ket{\varphi},\ \ket{x}_{\ola} \ket{\varphi}_{\olq }
        \xmapsto{\cG_{\ola}^*}
        \ket{x}_{\ola} \left(\sum_{i=0}^{2^m-1} \ket{\varphi_i}_{\olq}\right).
    \end{equation}
    We simply write $\cG^*$ instead of $\cG_{\ola}^*$ whenever $\ola$ is clear from the context.
\end{definition}
Note that $\cG_{\ola}^*$ is \emph{unique} if exists. We denote the set of $\cG_{\ola}^*$ by $\Uncomp^*(G, \ola)$ (either empty or a singleton), a subset of $\Uncomp(G, \ola)$.
Again, one may realize that any $\cG_{\ola}^* \in \Uncomp^*(G, \ola)$ safely uses dirty ancillas $\ola$ (recall~\cref{eq:safe-dirty}), and correctly implements the functionality $\func(G,\ola)$ on $\olq$. 

If $\func(G,\ola)$ is unitary, then $I_{\ola}\otimes \func(G,\ola)$ is a unitary witness of \cref{def:cleanuncomp};
conversely, a unitary witness $\cG$ restricts to a unitary on $\ket{0}_{\ola}\otimes\mathcal H_{\olq}$, and this restriction corresponds to $\func(G,\ola)$ by definition.
Therefore, the validity of functionality, uncomputation problem for clean and dirty ancillas, are equivalent, which leads to a unified formalization of existence of uncomputation:
\begin{proposition}\label{prop:eq-valid-uncomputation}
    $\func(G,\ola)$ is unitary iff
    $\Uncomp(G, \ola) \ne \emptyset$ iff
    $\Uncomp^*(G, \ola) \ne \emptyset$.
\end{proposition}

\subsection{Deciding Existence of Uncomputation Is coNP-Hard}

This subsection presents our main result on the inherent computational complexity of the uncomputation problem for circuits of polynomial gate size.

We first restrict our attention to \emph{reversible Boolean circuits}, a canonical and well-understood subclass of quantum circuits. 
Reversible Boolean circuits realize \emph{reversible Boolean functions} (RBFs), i.e., bijections $f:\{0,1\}^n \to \{0,1\}^n$, using reversible gates such as the universal Toffoli and Fredkin gates~\cite{reversiblecomputing, conservativelogic}.
Intuitively, checking the validity of functionality for a reversible Boolean circuit $G$ is equivalent to asking: if we fix certain input positions of the corresponding RBF to $0$ and then ignore the outputs on those same positions, does the resulting function remain reversible? 
Formally, we define the $k$-fixed RBF as follows.

\begin{definition}[$k$-fixed RBF]\label{def:KRBF}
    Let $f : \{0,1\}^n \to \{0,1\}^n$ be an RBF, with 
    $f = (f_1, \dots, f_n)$. 
    For $k \in \{1,\dots, n-1\}$, the \emph{$k$-fixed version} of $f$, denoted by \(f|_k : \{0,1\}^{n-k} \to \{0,1\}^{n-k}\), is the mapping
    \[
        f|_k(b_1, \dots, b_{n-k}) \triangleq (f_{k+1}(0^k, b), \dots, f_n(0^k, b)),
    \]
    where $0^k$ denotes a prefix of $k$ zeros and $b = (b_1,\dots,b_{n-k})$. 
\end{definition}

This leads to the following decision problem.

\begin{definition}[Reversibility of $k$-fixed RBF (R-kRBF)]\label{def:RKF}
    Given a reversible Boolean circuit over $n$ bits of polynomial gate size, constructed from $\tToffoli$ and $\tX$ gates, implementing the RBF $f : \{0, 1\}^n \to \{0, 1\}^n$, and $k \in \{1,\dots, n-1\}$, decide whether the $k$-fixed RBF $f|_k$ is reversible.
\end{definition}

The validity of functionality and R-kRBF are equivalent.

\begin{proposition}\label{prop:juequalrkf}
Let $G$ be a reversible Boolean circuit constructed from $\tToffoli$ and $\tX$ gates, acting on $k$ ancilla qubits $\ola$ and $n-k$ working qubits $\olq$, and let $f$ denote the RBF implemented by $G$. Then $\Uncomp(G, \ola) \ne \emptyset$ if and only if $f|_k$ is reversible.
\end{proposition}

The following technical lemma gives a polynomial construction from a conjunctive normal form (CNF) formula to a reversible Boolean circuit, 
serving as a key step to bridge the well-known NP-complete problem \emph{Boolean satisfiability} (SAT)~\cite{cooktheorem} to R-kRBF.

\begin{lemma}\label{lem:jordancorollary}
    Given any polynomial-size CNF formula $f$ with $n$ input bits, a polynomial-size circuit over $\{\tToffoli,\tX\}$ acting on $n+5$ bits can be constructed in polynomial time.The circuit preserves the $n$ input bits and applies a fixed 5-cycle to the remaining five bits if $f$ is satisfied, and the identity permutation otherwise.
\end{lemma}

This lemma follows from an adaptation of Jordan's proof in~\cite{strongequivalence}, which in turn relies on Barrington's Theorem~\cite{boundedwidthpolysizebranching} to establish a construction acting on $n+6$ bits, with one additional control bit initialized to $\ket{1}$.
We adapt this construction by using $\tX$ gates to implement negatively controlled Fredkin gates, thereby eliminating the need for the extra control bit and obtaining a circuit over $\{\tToffoli,\tX\}$ on $n+5$ bits.

\begin{theorem}\label{thm:RKFcoNPC}
    R-kRBF is coNP-complete.
\end{theorem}

\begin{proof}
We prove this in two parts: membership in coNP and coNP-hardness.

\paragraph{Step 1: $\textrm{R-kRBF} \in \textrm{coNP}$.}
Given a reversible Boolean circuit of polynomial gate size implementing 
$f : \{0,1\}^n \to \{0,1\}^n$, 
and an integer $k \in \{1,\dots, n-1\}$, 
a \emph{non-reversibility certificate} consists of two distinct inputs 
$\overline{b} \neq \overline{c}$ such that 
$f|_k(\overline{b}) = f|_k(\overline{c})$, 
which is verifiable in polynomial time. 
Hence, $\textrm{R-kRBF} \in \textrm{coNP}$.

\paragraph{Step 2: $\textrm{coNP}$-Hardness via Reduction from UNSAT}
Let $f$ be a CNF formula over $n$ variables.
By Lemma~\ref{lem:jordancorollary}, we can construct, in polynomial time,
a reversible Boolean circuit $G$ acting on $n+5$ bits
(using $\tToffoli$ and $\tX$ gates)
that implements a fixed 5-cycle on $a,b,c,d,e$
conditioned on $f(\overline{q})$:
\[
g(a,b,\overline{q},c,d,e)=
\begin{cases}
(b,c,\overline{q},d,e,a) & \text{if } f(\overline{q})=1,\\
(a,b,\overline{q},c,d,e) & \text{if } f(\overline{q})=0,
\end{cases}
\]
where $g$ denotes the RBF implemented by $G$.
In particular, $g$ satisfies:
\begin{align*}
&g(0,0,\overline{q},0,0,0)
  =(0,0,\overline{q},0,0,0)
  &&\text{if } f(\overline{q})=1,\\
&g(0,1,\overline{q},0,0,0)
  =(1,0,\overline{q},0,0,0)
  &&\text{if } f(\overline{q})=1,\\
&g(a,b,\overline{q},c,d,e)
  =(a,b,\overline{q},c,d,e)
  &&\text{if } f(\overline{q})=0.
\end{align*}

For the $1$-fixed version $g|_a$:
\begin{itemize}
  \item If $f$ is satisfiable, there exist inputs
  $\overline{q}$ such that
  \[
  g|_a(0,\overline{q},0,0,0)
  =g|_a(1,\overline{q},0,0,0)
  =(0,\overline{q},0,0,0),
  \]
  implying that $g|_a$ is not reversible.

  \item If $f$ is unsatisfiable, then $g$ is the identity
  function, and therefore $g|_a$ is also the identity
  function, which is reversible.
\end{itemize}

Hence, $f$ is unsatisfiable if and only if $g|_a$
is reversible, which establishes a polynomial-time
reduction from UNSAT to R-kRBF.
Therefore, R-kRBF is coNP-hard.
\end{proof}

Combining~\cref{prop:eq-valid-uncomputation} and~\cref{prop:juequalrkf}, we obtain the following corollary.

\begin{corollary}
    For reversible Boolean circuits of polynomial gate size composed of $\tToffoli$ and $\tX$ gates, the uncomputation problem is coNP-complete.
\end{corollary}

We now turn to general quantum circuits, for which deciding uncomputation is more challenging due to the presence of superposition. 

\begin{theorem}
For any quantum gate set in which both $\tToffoli$ and $\tX$ gates can be constructed in polynomial size, 
deciding whether a given quantum circuit of polynomial gate size with ancillas is uncomputable is coNP-hard.
\end{theorem}

In addition, the problem is harder than deciding the safe use of clean ancillas (recall~\cref{eq:safe-clean}):

\begin{proposition}\label{prop:hardercleansafety}
For general quantum circuits of polynomial gate size, deciding the existence of an uncomputation is at least as hard as deciding the safe use of clean ancillas.    
\end{proposition}

\section{Quantum Circuit Language}\label{sec:language}

This section introduces the quantum circuit language $\QCa$, which extends the simple quantum circuit language $\mathbf{QC}$~\cite{ying2023quantumrecursiveprogrammingquantum} with an explicit \emph{borrow-ancilla} construct. 
It provides a semantics with automatic uncomputation of different usage of clean and dirty ancillas.

\subsection{Syntax}
\label{sec:language-syntax}
We start by giving the syntax of quantum circuit language $\QCa$:
\begin{definition}[]
    A quantum circuit $C \in \QCa$ is inductively generated by the syntax:
    \begin{align*}
        C\ ::=\ & \bskip \mid U[\olq] \mid C_1; C_2 \mid \qif{q}{C_1}{C_0} \mid \borrowccdd{a}{C}
    \end{align*}
    where items in blue square brackets $\blue{[]}$ are optional.
\end{definition}
The constructs, except the $\bborrow$ statement, are rather standard:
\begin{itemize}
    \item Statement $\bskip$ does nothing and terminates immediately;
    \item Statement $U[\olq]$ applies unitary transformation $U$ on a list of distinct qubits $\olq$;
    \item Sequential composition $C_1; C_2$ executes $C_1$ followed by $C_2$;
    \item Quantum if-statement $\qif{q}{C_1}{C_0}$ models quantum control flow: if $q$ is in $|1\>$, then it executes $C_1$, and if in $|0\>$, then executes $C_0$.
    It provides a way to write quantum multiplexor.
\end{itemize}

\emph{$\bborrow$ statement} models scoped ancilla usage.
It allocates an ancilla $a$---clean with the optional initializer $:=\ket{0}$, and dirty otherwise---and the body executes as $C$.
At the end of the scope, $a$ is disposed. We assume that each $\bborrow$ statement introduces a fresh ancilla name.

\begin{definition}[Quantum variables]
We define the set of quantum variables as follows:
\begin{mathpar}
    qv(\bskip) \triangleq \emptyset

    qv(U[\olq]) \triangleq \{\, q \mid q \in \olq \,\}

    qv(C_1; C_2) \triangleq qv(C_1) \cup qv(C_2)

    qv(\qif{q}{C_1}{C_0}) \triangleq
    \{q\} \cup qv(C_1) \cup qv(C_0)
    \quad\text{if } q \notin qv(C_1) \cup qv(C_0)

    qv(\borrowccdd{a}{G}) \triangleq qv(G) \setminus \{a\}
\end{mathpar}
Here $qv(C)$ collects the \emph{working} qubits of $C$; borrowed ancillas are excluded.
\end{definition}

For simplicity, we introduce some syntax sugars of $\bqif$:
$\qifT{q}{C}\triangleq \qif{q}{C}{\bskip}$, 
$\qifT{\mathsf{not}\ q}{C}\triangleq \qif{q}{\bskip}{C}$, and 
for $i \in \mathbb{N}$, we inductively define $\tMCX{0}[t] \triangleq \tX[t]$ and $\tMCX{i+1}[q_1, \dots, q_{i+1}, t] \triangleq \bqif\ {q_1}\ \bthen$ ${ \tMCX{i}[q_2, \dots, q_{i+1}, t] }$.

\subsection{Semantics}

A quantum circuit on a fixed set of qubits denotes a unitary operator.
In $\QCa$, however, ancillas are intended to be borrowed and eventually disposed.
Therefore, the semantics of a program should be the induced transformation on the working qubits,
which we formalize as $\func(\cdot, \cdot)$ in \cref{sec:existence-of-uncomputation}.
To remain consistent with circuit semantics, this induced transformation must itself be unitary.

Under our definition, $\func(\cdot, \cdot)$ may fail to be unitary for some programs $C$ with ancillas $\ola$.
We regard such programs as \emph{invalid} and represent their denotation by a distinguished element $\bot$.
We stipulate that $\bot$ propagates through any composition, e.g., sequential composition $\bot U = U\bot = \bot$ and tensor product $U \otimes \bot = \bot \otimes U = \bot$ for all unitary operators $U$.
Hence, any invalid subprogram renders the whole program invalid.

\begin{definition}\label{def:semantics}
    The denotational semantics of a program $C \in \QCa$, denoted $\sem{C}$, is defined inductively: it is either a unitary operator on qubits $qv(C)$, or $\bot$ if $C$ is invalid, where $\bot$ is absorbing for addition, sequential composition, and tensor product.

    \begin{enumerate}
        \item $\sem{\bskip} \triangleq I$;
        
        \item $\sem{U[\olq]} \triangleq U_{\olq}$;
        
        \item $\sem{C_1; C_2} \triangleq 
        \bigl( \sem{C_2} \otimes I_{qv(C_1) \setminus qv(C_2)} \bigr)
        \bigl( \sem{C_1} \otimes I_{qv(C_2) \setminus qv(C_1)} \bigr)$;
        
        \item $\sem{\qif{q}{C_1}{C_0}} \triangleq 
        \ketbra{1}{1}_q \otimes \bigl(\sem{C_1} \otimes I_{qv(C_0) \setminus qv(C_1)}\bigr) +
        \ketbra{0}{0}_q \otimes \bigl(\sem{C_0} \otimes I_{qv(C_1) \setminus qv(C_0)}\bigr)$;

        \item $\sem{\borrowccdd{a}{C}} \triangleq
            \begin{cases}
            \func(\sem{C},a), & \text{if } \sem{C}\neq\bot \text{ and } \Uncomp(\sem{C},a)\neq\emptyset,\\
            \bot, & \text{otherwise.}
            \end{cases}$
    \end{enumerate}
\end{definition}

The denotational semantics, except for the $\bborrow$ statement, is the same as~\cite{ying2023quantumrecursiveprogrammingquantum}, which is induced from operational semantics.
For $\borrowccdd{a}{C}$, we define its semantics to be $\func(\sem{C},a)$ when an uncomputation exists, and $\bot$ otherwise.

We further define the semantic equivalence of programs:
\begin{definition}[Semantic equivalence~\cite{lawsofquantumprog}]
    Let $C_1, C_2 \in \QCa$ be two programs. They are called equivalent, written $C_1 \equiv C_2$, if $\sem{C_1} \otimes I_{qv(C_2) \setminus qv(C_1)} = \sem{C_2} \otimes I_{qv(C_1) \setminus qv(C_2)}$. Here equality is lifted to include $\bot$, with the convention that $\bot=\bot$ and $\bot\neq U$ for any unitary operator $U$.
\end{definition}

\subsection{Automatic Uncomputation---An Overview}

Under the semantics above, checking validity coincides with checking existence of uncomputation.
Beyond existence, we also emphasize \emph{synthesis}. First, existence checking can be computationally expensive; it is therefore desirable to make this effort constructive, returning not only a yes/no answer but also a witness that can be turned into code.
Second, even when an uncomputation exists, synthesizing an explicit uncomputation circuit may be more difficult or costly, preventing the program from actually running.
Accordingly, our goal is \emph{synthesis-oriented existence checking}, which simultaneously checks existence of uncomputation and produces actionable evidence that directly guides the construction of uncomputation.

In this paper, we instantiate this idea at two levels: a high-level, syntax-directed \emph{static reasoning system}
in~\cref{sec:template} and a low-level \emph{rewrite-based normalization} procedure in~\cref{sec:rwun}.
The static system operates on program syntax; it is coarse-grained and lightweight, but applies only when programs follow
the required discipline.
When it fails to check a borrow scope, we fall back to the circuit-level rewrite procedure on the remaining program,
which is finer-grained and more broadly applicable, though not always efficient.
Together, these two components form a top-down pipeline.

\section{A Lightweight Check}\label{sec:template}

Although existence checking is computationally hard in general, a lightweight approach remains attractive in practice. It serves as a sufficient filter: for well-structured programs, it can check existence and enable regular synthesis, while flagging potential misuse early (e.g.,~\cite{silq,qurts}).

\subsection{A Syntax-Directed Static Reasoning System}

\paragraph{Type Discipline for Clean Ancillas}
For clean ancillas, we adapt the two key qualifiers (or types) used in Silq~\cite{silq}, which can be understood as capturing two semantic constraints of a program $C$:
(i) $C : \qfree$ asserts that the unitary $\sem{C}$ is a permutation---i.e., a reversible Boolean transformation on the computational basis. 
It is typically used to ensure that the ancilla-writing (information \emph{store}) operation can be inverted regularly (e.g., by appending its inverse circuit);
(ii) $C : \spform{\olq}{\olt}$ with $qv(C)\subseteq\olq\cup\olt$ asserts that the unitary $\sem{C}$ remains unchanged on $\olq$, i.e., $\sem{C}(\sum_i|i\>_{\olq}|\phi_i\>) = \sum_i|i\>_{\olq}|\phi'_i\>$ where $|i\>$ range over the computational basis of $\olq$. 
It expresses dependency invariance: qubits on which the ancilla depends must remain unchanged over the region where their values are later needed to reverse the ancilla.
Together, these qualifiers give rise to our first core rule CLEAN (shown in \cref{fig:selected-rule}), which guarantees the existence of uncomputation for the clean ancilla $a$ and enables a regular synthesis: $\sem{C_s}^\dagger \sem{C_u} \sem{C_s} \in \Uncomp(\sem{C_s;C_u}, a)$. That is, the ``uncompute'' step can be realized by appending the inverse of the $C_s$ part.

\begin{figure}[t]
    \centering
    \small
    \begin{mathpar}
    \inferrule*[right=CLEAN]{ 
        \Gamma \vdash C_s : \spform{\olq}{a} \\ \Gamma \vdash C_u : \spform{\olq, a}{\olt} \\ \Gamma \vdash C_s : \qfree 
    }{ 
        \Gamma \vdash \borrowcc{a}{C_s;C_u} : \spform{\olq}{\olt}
    }\\
    \inferrule*[right=DIRTY]
    { \Gamma \vdash C_s : \spform{\olq}{a} \\ \Gamma \vdash C_u : \cqbf{\olq}{a}{\olt} \\ \Gamma \vdash C_s : \qfree }
    { \Gamma \vdash \borrowdd{a}{C_s;C_u} : \spform{\olq}{\olt} }\\
    \inferrule*[right=DIRTY$^*$]
    { \Gamma \vdash C_s : \spform{\olq}{a} \\ \Gamma \vdash C_u : \cqbf{\olq}{a}{\olt} \\ \Gamma \vdash C_s : \qfree }
    { \Gamma \vdash \borrowd{a}{C_s}{C_u} : \spform{\olq}{\olt} }
\end{mathpar}
    \caption{Selected rules}
    \Description{
    Three selected rules of the syntax-directed static reasoning system, labeled CLEAN, DIRTY, and DIRTY-star.
    }
    \label{fig:selected-rule}
\end{figure}

\paragraph{Challenge: Synthesis for Dirty Ancillas}

While existence of uncomputation can be certified by a common static discipline, \emph{constructing} correct dirty-ancilla circuits is generally more challenging. Existing synthesis methods for clean ancillas do not directly apply: to the best of our knowledge, there is no systematic way to transform an \emph{arbitrary} clean-ancilla usage into a dirty-ancilla one.
\Cref{subsec:safeusecase} introduces a potential regular construction pattern---the template
$\sem{C_s}^\dagger \sem{C_u} \sem{C_s} \sem{C_u}$, preceded by an additional leading “compute” step $\sem{C_u}$. However, this pattern does not always belong to $\Uncomp^*(\sem{C_s;C_u}, a)$; for example, it can fail when implementing $\tMCU{3}{\tT}$ in the same manner---by replacing $\tToffoli[a,q_3,t]$ with $\tMCU{2}{\tT}[a,q_3,t]$.
This observation suggests that \qfree and \const alone are insufficient to capture dirty-ancilla usages that admit regular synthesis.

\paragraph{Solution: Controlled Quantum Boolean Functions (\bcqbf)}
Inspired by the explanation in~\cite{conditionalcleanqubits}, we introduce \bcqbf. We write $C : \cqbf{\olp}{\olq}{\olt}$ for a program $C$ if $\sem{C}$ remains unchanged on $\olp$, and acts nontrivially---as a quantum Boolean function~\cite{quantumbooleanfunctions} $U$ (i.e., $U^2 = I$)---on the targets $\olt$ only if the controls $\olq$ are in state $\ket{\overline{1}}$.
By additionally requiring $C_u$ to satisfy \bcqbf, as shown in rule DIRTY in \cref{fig:selected-rule}, we \emph{enable} the above regular construction
$\sem{C_s}^\dagger \sem{C_u} \sem{C_s} \sem{C_u} \in \Uncomp^*(\sem{C_s;C_u}, a)$.
Concretely, the sequence $\sem{C_u}\sem{C_s}\sem{C_u}$ forms a canonical toggling pattern: whether the operation $U$ (the corresponding quantum Boolean function $U$ of $\sem{C_u}$ on the target $\olt$) is executed once or twice is determined by the occurrence of a flip on $a$; executing it once applies the update $U$ on $\olt$, while executing it twice cancels to the identity (since $U^2 = I$).
The $\tMCX{3}$ example in \cref{subsec:safeusecase} with dirty ancillas illustrates this pattern: the use operation $\tToffoli[a,q_3,t]$ conforms to $\cqbf{q_3}{a}{t}$.

By leveraging \bcqbf, our system is able to jointly guarantee sound checking and regular synthesis, leading to the soundness result stated below.
Full details of the system and the synthesis procedure are provided in appendix.

\begin{theorem}[Soundness of checking and synthesis]
Let $G \in \QCa$ be a program with clean ancillas $\ola$, dirty ancillas $\overline{d}$, and working qubits $\olq$ (i.e., $\olq = qv(G)$).
If $G$ is derived to be $\spform{\cdot}{\cdot}$ (the exact parameter is not essential and thus omitted), then $\sem{G} \ne \bot$, and the synthesis procedure outputs a circuit $\cG$ containing no $\bborrow$ statement such that, for all $\ket{x}$ and $\ket{\varphi}$,
\[
    \sem{\cG}\ket{0}_{\ola}\ket{x}_{\overline{d}}\ket{\varphi}_{\olq}
    = \ket{0}_{\ola}\ket{x}_{\overline{d}}(\sem{G}\ket{\varphi}_{\olq}).
\]
\end{theorem}

\paragraph{Efficiency Problem: Annotating Checkpoints via Syntactic Sugar}
Notably, in contrast to \qfree and \const, reasoning about \bcqbf is non-compositional: whether a program satisfies \bcqbf cannot be determined from its subprograms alone, but depends on global interactions (e.g., whether sequential components commute). 
As a consequence, a naive (or exhaustive) search over all possible checkpoints---i.e., splitting the borrow body into $C_s$ and $C_u$ and then checking $C_u:\bcqbf$---would be computationally expensive in practice. To address this, we introduce syntactic sugar $\borrowd{a}{C_s}{C_u} \triangleq \borrowdd{a}{C_s; C_u}$ with the rule DIRTY$^*$ shown in \cref{fig:selected-rule}, which allows programmers to explicitly annotate the checkpoint, thereby facilitating the application of the general \textsc{Dirty} rule. 
With the help of this syntactic sugar, our system is sufficient to validate many dirty-ancilla usages that follow the standard toggling pattern, including both the clean and dirty versions of Deutsch--Jozsa~\cite{deutsch}, MultiControlled-$X$, MultiControlled-$Y$ Rotation, IntegerComparator, Incrementer, and HighestBitConstAdder; for Adder, only the clean version is supported.
In the next subsection, we illustrate how to better identify checkpoints through a concrete example.

\subsection{Case Study: Using Dirty Ancillas as Easily as Clean Ones}

We take \emph{HighestBitConstAdder} as an example, showing how our system simplifies the use of dirty ancillas. Similar to the clean ancillas, the programmer focuses solely on the functionality, and the necessary checks and uncomputation synthesis are delegated to automated reasoning.

\cref{fig:hbadder-circ} (left) presents an adapted instance of HighestBitConstAdder~\cite{factoring2n+2} with dirty ancillas. The program computes the most significant bit of the integer $(x_3 x_2 x_1 x_0)_2$ after adding the constant $(0101)_2$ and stored it in $x_3$, i.e., transforms $x_3$ from $|x_3\>$ to $|x_3 \oplus x_2 \oplus x_0 x_1 \oplus x_0 x_1 x_2\>$ while leaving other qubits unchanged. At first glance, the program is difficult to interpret: its functionality, design principles, and correctness are far from obvious. Rather than reverse-engineering it, we construct a functionally equivalent program in the right panel of \cref{fig:hbadder-circ}, written in a clean-ancilla style but using dirty ancillas.
The compute follows by boolean arguments with two (clean) ancillas $a_0$ and $a_1$:
\begin{align*}
    & a_0 := x_0; & \mbox{line 3: }& \mbox{\textit{store} $x_0$ in $a_0$}\\
    & a_1 := x_1 \land a_0; & \mbox{line 5: }& \mbox{\textit{use} $a_0$ as control; \textit{store} temporal value in $a_1$}\\
    & x_3 := ((\neg x_2) \wedge a_1) \oplus x_3; & \mbox{line 10: }& \mbox{\textit{use} $a_1$ as control}\\
    & x_3 := x_2 \oplus x_3; & \mbox{line 14: }& \mbox{no ancilla involved}
\end{align*}
yields $x_3$ in $x_2 \oplus (((\neg x_2) \wedge (x_1 \land x_0)) \oplus x_3) = x_3 \oplus x_2 \oplus x_0 x_1 \oplus x_0 x_1 x_2$. Here, \textit{store} and \textit{use} indicate the partition of commands into $C_s$ and $C_u$, respectively. Principally, a dirty ancilla is preferably used in two steps: a \emph{store} step that modifies only the ancilla (possibly controlled by other qubits), followed by a \emph{use} step that treats the ancilla purely as a control while applying a quantum Boolean function to other qubits, without modifying any qubit that the store phase depends on.

\begin{figure}[tbp]
\centering
\hfill
\begin{minipage}{0.42\linewidth}
\centering
\begin{subfigure}{0.98\linewidth}
    \centering
    \begin{lstlisting}[language=QCa, xleftmargin=0.2\linewidth, framexrightmargin=0.2\linewidth]
        CNOT[a2, x3];
        CNOT[x2, a2];
        X[x2];
        CCNOT[a1,x2,a2];
        CCNOT[a0,x1,a1];
        CNOT[x0,a0];
        CCNOT[a0,x1,a1];
        CCNOT[a1,x2,a2];
        CNOT[a2, x3];
        CCNOT[a1,x2,a2];
        CCNOT[a0,x1,a1];
        CNOT[x0,a0];
        CCNOT[a0,x1,a1];
        CCNOT[a1,x2,a2];
        X[x2];
        CNOT[x2, a2]
    \end{lstlisting}
    \caption{Manually designed and uncomputed.}
    \label{fig:manual-uncomp}
\end{subfigure}
\end{minipage}
\hfill
\begin{minipage}{0.54\linewidth}
\centering
\begin{subfigure}{0.98\linewidth}
    \centering
    \begin{lstlisting}[language=QCa, xleftmargin=0.2\linewidth, framexrightmargin=0.2\linewidth]
    borrow a1 store {
        borrow a0 store {
            CNOT[x0, a0]
        } use {
            CCNOT[a0, x1, a1]
        }
    } use {
        qif a1 then {
            qif not x2 then {
                X[x3]
            }
        }
    };
    CNOT[x2, x3]
    \end{lstlisting}
    \caption{Automatically checked.}
    \label{fig:auto-uncomp}
\end{subfigure}
\end{minipage}
\caption{Comparison of implementations for computing $x_3$, the most significant bit of $(x_3 x_2 x_1 x_0)_2  + (0101)_2$, with manual vs.\ automatic uncomputation. 
}
\Description{
Two code listings comparing manual and automatic uncomputation for computing
the most significant bit of adding the constant 0101 to a four-bit input.
The left listing shows a manually designed and uncomputed circuit, while the
right listing uses nested borrow constructs that are automatically checked.
}
\label{fig:hbadder-circ}
\end{figure}

\subsection{Limitations and Complements}

A lightweight, syntax-directed analysis that recognizes disciplined ancilla usage in well-structured programs can be practically effective---and may cover a broad fragment---but it is inherently far from complete. This limitation is especially pronounced for dirty ancillas: even under disciplined patterns, correctness typically depends on semantic side conditions,
most notably commutativity. Such conditions are non-compositional and often costly to establish, limiting what a lightweight system can achieve. Accordingly, we treat the system developed in this section as an efficient first-pass filter; when it fails, we fall back to the lower-level rewrite-based method in the next section, which targets more flexible usage patterns, including those outside known templates.

\section{Rewrite-Based Normalization}\label{sec:rwun}

In this section, we introduce our second synthesis-oriented existence checking method, \emph{rewrite-based normalization}. We present the method in a plain circuit model (without $\bborrow$), rather than restricting it to $\QCa$, to keep the development focused and \emph{method-centric}. It applies to $\QCa$ directly: we eliminate $\bborrow$ statements from the inside out by normalizing each scope body while treating the borrowed qubit as an ancilla, replacing the scope with the synthesized uncomputed circuit upon success, and repeating.

\subsection{Intuition}

Consider a circuit $G$ that is uncomputable but has not yet been uncomputed. Our key insight is that its behavior may be \emph{semantically separated}:
\begin{quote}
The circuit behaves as if it first applies the intended functionality $\cG_a$ to the working qubits (without touching the ancilla),
and only then updates the ancilla (possibly based on the working qubits).
\end{quote}

\cref{fig:normal-form-intuition} illustrates this separation idea.
From left to right: (i) the original circuit $G$, (ii) a semantically equivalent circuit with a structural separation that first applies $\cG_a$ and then updates the ancilla $a$, and (iii) an uncomputation of $a$ obtained by appending to $G$ the inverse of the update suggested by (ii). Therefore, existence checking reduces to finding an equivalent circuit with this separated structure, and synthesis is then straightforward.

\begin{figure}[h]
\centering
\begin{minipage}[c]{0.24\linewidth}
\centering
\begin{quantikz}[row sep=0.2cm]
    \lstick{$q$}          & \gate[wires=2]{G} & \qw \\
    \lstick{$a$} &                   & \qw 
\end{quantikz}
\end{minipage}
\hspace{-0.6em}
$\equiv$
\hspace{-0.6em}
\begin{minipage}[c]{0.34\linewidth}
\centering
\begin{quantikz}[row sep=0.2cm]
    \lstick{$q$} & \gate{\blue{\cG_a}} & \ctrl{1} & \qw \\
    \lstick{$a$} & \qw                      & \gate{U} & \qw
\end{quantikz}
\end{minipage}
\hspace{-0.6em}
$\Rightarrow$
\hspace{0.2em}
\begin{minipage}[c]{0.34\linewidth}
\centering
\begin{quantikz}[row sep=0.2cm]
    \lstick{$q$} & \gate[wires=2]{G} & \ctrl{1} & \qw \\
    \lstick{$a$} &                   & \gate{\red{U^\dagger}} & \qw 
\end{quantikz}
\end{minipage}
\caption{Separating functionality from the ancilla update enables uncomputation by inverting the update.}
\Description{
Three circuit diagrams illustrating the separation idea for uncomputing ancilla
a: the original circuit G, an equivalent circuit that separates the computation
on q from the update of a, and the resulting uncomputation obtained by appending
the inverse update to G.
}
\label{fig:normal-form-intuition}
\end{figure}

\subsection{Normal Form}

We now characterize the targeted separated structure.
We consider \qfree circuits generated by the gate set
$S_c \triangleq \{ \tMCX{i} \}_{i \in \mathbb{N}}$ (with $\tMCX{0} \triangleq \tX$),
and quantum circuits generated by
$S_q \triangleq S_c \cup \{\tZ, \tH, \tS, \tT\}$.

\paragraph{Normal Form of \qfree Circuits}
The following proposition characterizes exactly when uncomputation exists for \qfree circuits.

\begin{proposition}\label{prop:complete-qfreenormal}
Let $G$ be a \qfree circuit constructed within $S_c$, acting on ancillas $\ola$ and working qubits $\olq$. Then $G$ is uncomputable if and only if there exists another \qfree circuit $G'$ acting on the same qubits such that the following two conditions hold:
\begin{enumerate}
\item $\sem{G} = \sem{G'}$.
\item For any $\tMCX{i}$ gate in $G'$ whose target is an ancilla in $\ola$ and uses no ancilla as a control qubit, every gate appearing to its right is also a $\tMCX{i}$ gate targeting (possibly different) ancillas in $\ola$.
\end{enumerate}
\end{proposition}

Condition~(2) induces a \emph{prefix--suffix} separation of $G'$.
Let $g$ be the first $\tMCX{i}$ gate that targets an ancilla in $\ola$ and uses no ancilla as a control qubit; we use $g$ as the separation point.
\begin{enumerate}
  \item \emph{Prefix (before $g$).}
  Every gate in the prefix either (i) targets a working qubit, or (ii) targets an ancilla but is controlled by at least one ancilla. A gate of type (ii) is dormant until all of its controlling ancillas (initialized in $\ket{0}$) are flipped to $\ket{1}$, which never happens in a prefix that contains only these two types of gates.
  Hence, the prefix leaves all ancillas unchanged (remaining in $\ket{0}$) and acts only as a unitary transformation on the working qubits---namely, the intended unitary behavior.

  \item \emph{Suffix (from $g$ onward).}
  By Condition~(2), every gate to the right of $g$ is also a $\tMCX{i}$ gate targeting an ancilla in $\ola$.
  Therefore, the suffix (including $g$) never targets any working qubit and can be viewed as an ancilla-update tail.
\end{enumerate}

\noindent
Overall, Condition~(2) formalizes the separated structure in which \emph{ancilla-targeting} operations are confined to a rightmost tail, up to ancilla-controlled ancilla writes that are dormant in the prefix.

The next proposition gives a variant of \cref{prop:complete-qfreenormal} that weakens the equivalence requirement while strengthening the structural separation requirement.

\begin{proposition}\label{prop:qfreenormal}
$G$ is uncomputable if and only if there exists another \qfree circuit $G'$ acting on the same qubits such that the following two conditions hold:
\begin{enumerate}
\item $\sem{G} \ket{0}_{\ola}\ket{e_i}_{\olq} = \sem{G'} \ket{0}_{\ola}\ket{e_i}_{\olq}$ for all computational basis states $\ket{e_i}_{\olq}$ of $\olq$.
\item For any $\tMCX{i}$ gate targeting an ancilla in $\ola$, 
all gates to its right must also be $\tMCX{i}$ gates targeting (possibly different) ancillas in $\ola$. 
\end{enumerate}
\end{proposition}

Compared to Condition~(2) of \cref{prop:complete-qfreenormal}, Condition~(2) above enforces a cleaner separation:
no gate in the prefix targets an ancilla, so the circuit is separated into a working-target-only prefix followed by an ancilla-target-only suffix.

The next proposition clarifies why \cref{prop:qfreenormal} can weaken the equivalence requirement. By \cref{def:cleanuncomp} and \cref{def:dirtyuncomp}, we directly have:
\begin{proposition}\label{prop:uncompeqwithnormal}
    For any $G'$ satisfying the condition~(1) in \cref{prop:qfreenormal}, $\Uncomp(\sem{G}, \ola) = \Uncomp(\sem{G'}, \ola)$ and $\Uncomp^*(\sem{G}, \ola) = \Uncomp^*(\sem{G'}, \ola)$.
\end{proposition}

\paragraph{Normal Form of Quantum Circuits}

Unlike the \qfree case, we do not have a structural separation of the above kind that is both necessary and sufficient for the existence of uncomputation for general quantum circuits.
Nevertheless, we can still define a \emph{sufficient} one that guarantees uncomputability.

\begin{proposition}\label{prop:quantumnormal}
    Let $G$ be a quantum circuit constructed within $S_q$, acting on ancillas $\ola$ and working qubits $\olq$. Then $G$ is uncomputable if there exists another quantum circuit $G'$ acting on the same qubits such that the following two conditions hold:

    \begin{enumerate}
    \item $\sem{G}\ket{0}_{\ola}\ket{e_i}_{\olq} = \sem{G'}\ket{0}_{\ola}\ket{e_i}_{\olq}$ for all computational basis states $\ket{e_i}_{\olq}$ of $\olq$.
    \item No $\tH$ gate targets any ancilla. For any gate targeting an ancilla in $\ola$, all gates to its right must also be gates targeting (possibly different) ancillas in $\ola$.
    \end{enumerate}
\end{proposition}

\subsection{Rewrite-Based Normalization}

The normal forms above turn synthesis-oriented existence checking into a reachability problem.
For a given circuit $G$, a semantically equivalent circuit $G'$ in normal form witnesses the existence of uncomputation. Meanwhile, $G'$ suggests a regular synthesis scheme---either canceling the ancilla-update tail in $G'$ or appending its inverse to $G$.

We take \cref{prop:quantumnormal} as the reference notion of \emph{normal form}, which specializes to \cref{prop:qfreenormal} for \qfree circuits (in both Conditions~(1) and~(2)).
We say that a circuit is \emph{normalizable} if it admits a semantically equivalent circuit in this normal form.
Normalizability implies uncomputability; for \qfree circuits, the converse also holds (recall \cref{prop:qfreenormal}).

Guided by the separated structure captured by Condition~(2), we design a normalization procedure that transforms a circuit into a semantically equivalent one satisfying this condition.
Equivalently, we aim to eliminate all \emph{violating pairs}---an ancilla-targeting gate immediately followed by a working-targeting gate---until none remains, at which point the circuit is in normal form.
A naive swap of such a pair would restore the desired order but may change the semantics; therefore, we insert an additional gate between the reordered (now normal-form) pair when necessary to maintain semantic equivalence.
This insertion also preserves the normal-form order.
Each such commute (and possible insertion) step is equivalence-preserving and is captured by rewrite rules.
Applying these local rewrites repeatedly yields a normalization procedure for reaching the normal form.

\paragraph{Rewriting Rules}

Let $\olq = q_1 \dots q_m$ and $\olp = p_1 \dots p_n$, with $0 \le m$ and $0 \le n$. Let $a$ denote an ancilla and $t$ denote a working qubit. We first list the rewrite rules for \qfree gates:
{
\small
\begin{align}
& \tMCX{m}[\olq, a]; \tMCX{n}[\olp, t]
  \equiv
  \tMCX{n}[\olp, t]; \tMCX{m}[\olq, a]
& \text{if } t \notin \olq \;\land\; a \notin \olp 
\tag{R-1}\label{rule:R-1}\\[2pt]
& \tMCX{m}[\olq, a]; \tMCX{n}[\olp, t]
  \equiv
  \tMCX{n}[\olp, t];
  \mathtt{C^h NOT}[(\olq \cup \olp)\setminus\{a\},\, t];
  \tMCX{m}[\olq, a]
& \text{if } a \in \olp \;\land\; t \notin \olq \tag{R-2}\label{rule:R-2}\\[2pt]
& \tMCX{m}[\olq, a]; \tMCX{n}[\olp, t]
  \equiv
  \tMCX{n}[\olp, t];
  \mathtt{C^h NOT}[(\olq \cup \olp)\setminus\{t\},\, a];
  \tMCX{m}[\olq, a]
& \text{if } a \notin \olp \;\land\; t \in \olq \tag{R-3}\label{rule:R-3}
\end{align}
}

For general quantum gates, we additionally use the following commutation rules,
mainly for phase-type gates $\tZ$, $\tS$, and $\tT$.
\begin{align}
U[a]; V[t]
  \ &\equiv
   V[t]; U[a]
  && \text{if }
    U \in \{\tX,\tZ,\tS,\tT\}
    \;\land\;
    V \in \{\tX,\tZ,\tH,\tS,\tT\}
  \tag{R-4}\label{rule:R-4}\\
U[a]; \tMCX{m}[\olp, t]
  \ &\equiv
   \tMCX{m}[\olp, t]; U[a]
  && \text{if }
    U \in \{\tZ,\tS,\tT\}
  \tag{R-5}\label{rule:R-5}\\
\tMCX{n}[\olp, a]; U[t]
  \ &\equiv
   U[t]; \tMCX{n}[\olp, a]
  && \text{if }
    U \in \{\tZ,\tS,\tT\}
    \;\lor\;
    (U = \tH \land t \notin \olp)
  \tag{R-6}\label{rule:R-6}
\end{align}

\begin{wrapfigure}{r}{7.0cm}
    \centering
    \begin{minipage}[t]{0.48\textwidth}
        \begin{algorithm}[H]
        \caption{Rewrite-based normalization.}
        \label{alg:rewrite-normalization}
        \begin{algorithmic}[1]
        \algtext*{EndIf}
        \algtext*{EndWhile}
        \Require $G, \ola, \olq$.
        \State \textbf{assert}(No $\tH$ gate acts on any ancilla)
        \State $G' \gets G.\text{copy}()$
        \While{\textbf{true}}
            \State Find the first violating pair.
            \If{No pair is found}
                \State \Return $G'$
            \Else
                \If{A rewriting rule applies}
                    \State Rewrite in $G'$.
                \Else
                    \State Report failure.
                \EndIf
            \EndIf
        \EndWhile
        \end{algorithmic}
        \end{algorithm}
    \end{minipage}
\end{wrapfigure} 

All rules above eliminate a local violating pair; two of the \qfree rules additionally insert a single compensating gate to preserve semantic equivalence while keeping the reordered pair in normal-form order.
Any violating pair not covered by the rules above is not normalizable.

\paragraph{Normalization Algorithm}
The procedure in \cref{alg:rewrite-normalization} scans $G'$ for the first violating pair; if a rewrite rule applies, it rewrites the pair and continues.
Otherwise, it reports failure.
When no violating pair remains, the resulting circuit is in normal form and is returned.

The success stopping condition is designed to coincide with
Condition~(2) of \cref{prop:quantumnormal}. Together with the fact that each rewrite step is
equivalence-preserving, this yields the following result.

\begin{theorem}[Termination and Soundness]\label{thm:normalize-soundness}
Given a quantum circuit $G$ acting on ancillas $\ola$ and working qubits $\olq$, \cref{alg:rewrite-normalization}  always terminates.
If it succeeds and returns a circuit $G'$,
then $\sem{G'}=\sem{G}$ and $G'$ satisfies the conditions in~\cref{prop:quantumnormal}.
Consequently, $\Uncomp(\sem{G}, \ola) \ne \emptyset$.
\end{theorem}

The normal form suggests a regular synthesis procedure.
For clean ancillas, synthesis from the normal form is uniform for both \qfree and general quantum circuits.
For dirty ancillas, a further but still regular operation is needed; for general quantum circuits, an additional step is required before this operation. We use \RwUnx to denote the overall procedure that first normalizes a circuit and then outputs a synthesized uncomputation; full details are provided in appendix.

\begin{theorem}[Synthesis]
\label{thm:rwun-soundness}
Given a quantum circuit $G$ acting on clean ancillas $\ola$, dirty ancillas $\overline{d}$ and working qubits $\olq$.
If $\RwUnx(G)$ terminates with success, then it outputs a circuit containing no $\bborrow$ statement, and for all $\ket{x}$ and $\ket{\varphi}$,
\[
    \sem{\RwUnx(G)}\ket{0}_{\ola}\ket{x}_{\overline{d}}\ket{\varphi}_{\olq}
    = \ket{0}_{\ola}\ket{x}_{\overline{d}}(\sem{G}\ket{\varphi}_{\olq}).
\]

\end{theorem}

\subsection{Case Study}

\paragraph{Normalization}
\Cref{fig:rwun-case-normalize} visualizes the execution of \cref{alg:rewrite-normalization} on the running example.
Each subfigure shows the \emph{entire circuit} after completing \emph{one} rewrite step.
We use the following visual conventions throughout:
(i) the dotted box highlights the fragment produced by rewriting the \emph{previous} step's violating pair;
(ii) the dashed box marks the leftmost violating pair \emph{to be rewritten next};
(iii) the solid box indicates the currently established prefix/suffix, which grows monotonically and will never participate in future violating pairs.
At every step, all violating pairs lie in the middle region between the blue prefix and suffix.

\begin{figure}[h]
\centering
\setlength{\tabcolsep}{2pt}
\renewcommand{\arraystretch}{1.0}
\begin{tabular}{cc}
\begin{subfigure}{0.48\linewidth}
\centering
\begin{quantikz}[row sep=0.15cm]
    \lstick{$q_1$} & \qw       & \ctrl{2}\gategroup[wires=5,steps=2,style={dashed,rounded corners,draw=red, fill=none, inner xsep=2pt},label style={label position=above,anchor=north,xshift=-0.1cm,yshift=0.4cm}]{{\rm \red{violating}}}  & \qw       & \ghost{}\qw     \\
    \lstick{$q_2$} & \qw       & \ctrl{1}   & \qw      & \qw     \\
    \lstick{$a$}   & \gate{Z}  & \targ{}    & \ctrl{2} & \qw     \\
    \lstick{$q_3$} & \qw       & \qw        & \ctrl{1} & \qw     \\
    \lstick{$t$}   & \qw       & \qw        & \targ{}  & \qw
\end{quantikz}
\caption{Initialization: the leftmost violating pair to be rewritten (dashed box).\\ }
\end{subfigure}
\ &\ 
\begin{subfigure}{0.48\linewidth}
\centering
\begin{quantikz}[row sep=0.15cm]
    \lstick{$q_1$} & \qw\gategroup[wires=5,steps=2,style={dashed,rounded corners,draw=red, fill=none, inner xsep=2pt},label style={label position=above,anchor=north,xshift=-0.1cm,yshift=0.4cm}]{{\rm \red{violating}}}  & \qw\gategroup[wires=5,steps=3,style={dotted,rounded corners, fill=none, inner sep=1pt},label style={label position=above,anchor=north,xshift=-0.1cm,yshift=0.4cm}]{}      & \ctrl{4} & \ctrl{2}\gategroup[wires=5,steps=1,style={solid,rounded corners,draw=blue, fill=none, inner xsep=2pt},label style={label position=above,anchor=north,xshift=-0.1cm,yshift=0.4cm}]{{\rm \blue{suffix}}}    & \ghost{}\qw     \\
    \lstick{$q_2$} & \qw      & \qw      & \ctrl{3} & \ctrl{1}    & \qw     \\
    \lstick{$a$}   & \gate{Z} & \ctrl{2} & \qw      & \targ{}     & \qw     \\
    \lstick{$q_3$} & \qw      & \ctrl{1} & \ctrl{1} & \qw         & \qw     \\
    \lstick{$t$}   & \qw      & \targ{}  & \targ{}  & \qw         & \qw
\end{quantikz}
\caption{After one rewrite by~\cref{rule:R-2}: the new leftmost violating pair (dashed box); the rewritten fragment (dotted box); the current suffix (solid box).}
\end{subfigure}
\\[2pt]
\begin{subfigure}{0.48\linewidth}
\centering
\begin{quantikz}[row sep=0.15cm]
    \lstick{$q_1$} & \qw\gategroup[wires=5,steps=1,style={solid,rounded corners,draw=blue, fill=none, inner xsep=2pt},label style={label position=above,anchor=north,xshift=-0.1cm,yshift=0.4cm}]{{\rm \blue{prefix}}}\gategroup[wires=5,steps=2,style={dotted,rounded corners, fill=none, inner sep=1pt},label style={label position=above,anchor=north,xshift=-0.1cm,yshift=0.4cm}]{}      & \qw\gategroup[wires=5,steps=2,style={dashed,rounded corners,draw=red, fill=none, inner xsep=2pt},label style={label position=above,anchor=north,xshift=-0.1cm,yshift=0.4cm}]{{\rm \red{violating}}}      & \ctrl{4} & \ctrl{2}\gategroup[wires=5,steps=1,style={solid,rounded corners,draw=blue, fill=none, inner xsep=2pt},label style={label position=above,anchor=north,xshift=-0.1cm,yshift=0.4cm}]{{\rm \blue{suffix}}}    & \ghost{}\qw     \\
    \lstick{$q_2$} & \qw      & \qw      & \ctrl{3} & \ctrl{1}    & \qw     \\
    \lstick{$a$}   & \ctrl{2} & \gate{Z} & \qw      & \targ{}     & \qw     \\
    \lstick{$q_3$} & \ctrl{1} & \qw      & \ctrl{1} & \qw         & \qw     \\
    \lstick{$t$}   & \targ{}  & \qw      & \targ{}  & \qw         & \qw
\end{quantikz}
\caption{After another rewrite by~\cref{rule:R-5}: updated violating pair (dashed box); newly rewritten fragment (dotted box); current prefix/suffix (solid boxes).}
\end{subfigure}
\ &\ 
\begin{subfigure}{0.48\linewidth}
\centering
\begin{quantikz}[row sep=0.15cm]
    \lstick{$q_1$} & \qw\gategroup[wires=5,steps=2,style={solid,rounded corners,draw=blue, fill=none, inner xsep=2pt},label style={label position=above,anchor=north,xshift=-0.1cm,yshift=0.4cm}]{{\rm \blue{prefix}}}                & \ctrl{4}\gategroup[wires=5,steps=2,style={dotted,rounded corners, fill=none, inner sep=1pt},label style={label position=above,anchor=north,xshift=-0.1cm,yshift=0.4cm}]{} & \qw\gategroup[wires=5,steps=2,style={solid,rounded corners,draw=blue, fill=none, inner xsep=2pt},label style={label position=above,anchor=north,xshift=-0.1cm,yshift=0.4cm}]{{\rm \blue{suffix}}}     & \ctrl{2}    & \ghost{}\qw     \\
    \lstick{$q_2$} & \qw      & \ctrl{3} & \qw      & \ctrl{1}    & \qw     \\
    \lstick{$a$}   & \ctrl{2} & \qw      & \gate{Z} & \targ{}     & \qw     \\
    \lstick{$q_3$} & \ctrl{1} & \ctrl{1} & \qw      & \qw         & \qw     \\
    \lstick{$t$}   & \targ{}  & \targ{}  & \qw      & \qw         & \qw
\end{quantikz}
\caption{After the final rewrite by~\cref{rule:R-5}: $G'$ is in normal form, with a prefix--suffix separation (solid boxes).\\ }
\end{subfigure}
\end{tabular}
\caption{Rewrite-based normalization on a running example.
Each panel is the circuit after one rewrite step.
Dotted box highlights the rewrite result of the \emph{previous} panel's violating pair; dashed box marks the next leftmost violating pair.
The prefix/suffix grows monotonically, visualizing progress toward the normal form.}
\Description{
Four circuit diagrams showing successive rewrite steps that normalize a quantum
circuit. Across the four panels, the violating pair moves
until none remains, yielding the final prefix--suffix normal form.
}
\label{fig:rwun-case-normalize}
\end{figure}

\paragraph{Synthesis}
\Cref{fig:rwun-case-synth} illustrates how the normal form suggests regular synthesis.
Since $\sem{G'}=\sem{G}$, we can synthesize an uncomputation for $G$ \emph{by operating on $G'$ only} and directly output the resulting circuit as the uncomputed version of $G$ (without modifying $G$ itself).

For clean ancillas, this is immediate from the normal form: we append, in reverse order, the inverses of all \qfree gates in the suffix that target the ancilla.
(In this example, for simplicity, this reduces to removing the last $\tToffoli$ gate, shown in gray.)
For dirty ancillas, informally, we eliminate \emph{all} gates that touch the ancillas.

\begin{figure}[htbp]
\centering
\begin{subfigure}{0.48\linewidth}
\centering
\begin{quantikz}[row sep=0.15cm]
    \lstick{$q_1$} & \qw\gategroup[wires=5,steps=2,style={solid,rounded corners,draw=blue, fill=none, inner xsep=2pt},label style={label position=above,anchor=north,xshift=-0.1cm,yshift=0.4cm}]{{\rm \blue{prefix}}}                & \ctrl{4} & \qw\gategroup[wires=5,steps=2,style={solid,rounded corners,draw=blue, fill=none, inner xsep=2pt},label style={label position=above,anchor=north,xshift=-0.1cm,yshift=0.4cm}]{{\rm \blue{suffix}}}                                      & \fctrl{2}   & \ghost{}\qw     \\
    \lstick{$q_2$} & \qw      & \ctrl{3} & \qw      & \fctrl{1}   & \qw     \\
    \lstick{$a$}   & \ctrl{2} & \qw      & \gate{Z} & \ftarg      & \qw     \\
    \lstick{$q_3$} & \ctrl{1} & \ctrl{1} & \qw      & \qw         & \qw     \\
    \lstick{$t$}   & \targ{}  & \targ{}  & \qw      & \qw         & \qw
\end{quantikz}
\caption{Clean uncomputation.\\ }
\end{subfigure}
\hfill
\begin{subfigure}{0.48\linewidth}
\centering
\begin{quantikz}[row sep=0.15cm]
    \lstick{$q_1$} & \qw\gategroup[wires=5,steps=2,style={solid,rounded corners,draw=blue, fill=none, inner xsep=2pt},label style={label position=above,anchor=north,xshift=-0.1cm,yshift=0.4cm}]{{\rm \blue{prefix}}}                 & \ctrl{4} & \qw\gategroup[wires=5,steps=2,style={solid,rounded corners,draw=blue, fill=none, inner xsep=2pt},label style={label position=above,anchor=north,xshift=-0.1cm,yshift=0.4cm}]{{\rm \blue{suffix}}}                                         & \fctrl{2}   & \ghost{}\qw     \\
    \lstick{$q_2$} & \qw       & \ctrl{3} & \qw        & \fctrl{1}   & \qw     \\
    \lstick{$a$}   & \fctrl{2} & \qw      & \fgate{Z}  & \ftarg      & \qw     \\
    \lstick{$q_3$} & \fctrl{1} & \ctrl{1} & \qw        & \qw         & \qw     \\
    \lstick{$t$}   & \ftarg    & \targ{}  & \qw        & \qw         & \qw
\end{quantikz}
\caption{Dirty uncomputation: further remove all ancilla-touching gates.}
\end{subfigure}
\caption{Regular synthesis once $G'$ is in normal form.
Faded gates indicate removal.}
\Description{
Two circuit diagrams showing synthesis from the normalized circuit. The left
panel removes the final ancilla update to obtain clean uncomputation. The right
panel further removes the remaining gates that touch the ancilla to obtain
dirty uncomputation.
}
\label{fig:rwun-case-synth}
\end{figure}

In the end, only $\tMCX{3}[q_1,q_2,q_3,t]$ remains, which is exactly the intended effect on the working qubits.
We treat $\tMCX{i}$ as an atomic gate here; its decomposition into $\tToffoli$ gates is orthogonal and has been extensively studied~\cite{gidney2015, conditionalcleanqubits, riseofconditional, zindorf2025efficientimplementationmulticontrolledquantum}.

\subsection{Discussion: Limitations and Potential}\label{sec:rwun-limits} 

\paragraph{Limits of \RwUnx and Failure Conditions.}
The normal form captures a broader class of circuits than those reachable by \RwUnx.
\RwUnx is a leftmost-first, \emph{pairwise} normalization: it succeeds whenever every leftmost violating pair it encounters is locally normalizable, in which case repeated local progress reaches the global normal form.
However, global normalizability does not imply this pairwise \emph{continuity}: a circuit may have a normal-form witness, yet \RwUnx can get stuck on a leftmost pair that is not locally normalizable, even though a wider-scope rewrite could bypass it. For clarity, we list the failure conditions below.
\begin{enumerate}
  \item \emph{Uncovered violating pairs.}
  \begin{enumerate}
    \item \emph{Mutual control:}
    \[
      \tMCX{m}[\olq, a];\ \tMCX{n}[\olp, t]
      \quad\text{with}\quad
      a \in \olp \ \land\ t \in \olq.
    \]
    \Cref{fig:rwun-fail-circuit} shows an example.
    This local pattern alone does not rule out uncomputation; however, in many practical circuits it correlates with non-uncomputability.

    \item \emph{Commuting through $\tH$:}
    \[
      \tMCX{m}[\olq, a];\ \tH[t] \quad\text{with}\quad t \in \olq,
    \]
    for which no semantically equivalent circuit exists that satisfies the normal form.
  \end{enumerate}

  \item \emph{Incompleteness of the normal form.}
  A $\tH$ gate targeting an ancilla is excluded by \cref{prop:quantumnormal}.
\end{enumerate}

\begin{figure}[tbp]
  \centering
  \begin{subfigure}[t]{0.31\textwidth}
    \begin{minipage}[t][1.7cm][t]{\linewidth}
      \vspace{0pt}
      \centering
      \begin{quantikz}[row sep=0.2cm]
        \lstick{$q$}  & \ctrl{1}  & \targ{}   & \qw \\
        \lstick{$a$}  & \targ{}   & \ctrl{-1} & \qw 
      \end{quantikz}
    \end{minipage}
    \caption{}
    \label{fig:rwun-fail-circuit}
  \end{subfigure}\hfill
  \begin{subfigure}[t]{0.31\textwidth}
    \begin{minipage}[t][1.7cm][t]{\linewidth}
      \vspace{0pt}
      \centering
      \includegraphics[height=1.6cm]{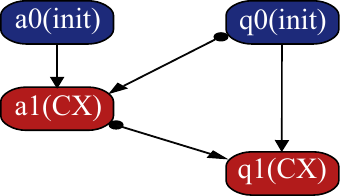}
    \end{minipage}
    \caption{}
    \label{fig:rwun-fail-graph-acyclic}
  \end{subfigure}\hfill
  \begin{subfigure}[t]{0.31\textwidth}
    \begin{minipage}[t][1.7cm][t]{\linewidth}
      \vspace{0pt}
      \centering
      \includegraphics[height=1.6cm]{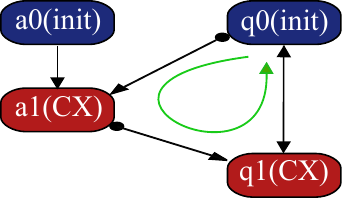}
    \end{minipage}
    \caption{}
    \label{fig:rwun-fail-graph-cycle}
  \end{subfigure}
  \caption{From left to right: a mutual-control circuit $C$ between an ancilla $a$ and a working qubit $q$, its circuit graph $G$, and the modified graph $G^\leftrightarrow$ where all target edges on working qubits are made bidirected; $G^\leftrightarrow$ contains an aw-cycle, highlighted in green.}
  \Description{
    Three panels showing a mutual-control circuit and its graph representations.
    The left panel is the circuit between working qubit q and ancilla a, the middle
    panel is its circuit graph, and the right panel is the modified graph with
    bidirected target edges on working qubits.
    }
  \label{fig:rwun-failure-example}
\end{figure}

\paragraph{A Sufficient Success Condition of \RwUnx.}

\RwUnx succeeds only if (i) the input contains no uncovered violating pair and (ii) rewriting never \emph{introduces} an uncovered one; the latter is subtle because compensating-gate insertions may create new violating pairs and cause the procedure to get stuck.
We therefore seek a sufficient condition that strengthens the input restriction so that ``no uncovered pair initially'' implies ``no uncovered pair will be introduced.''

We focus on \qfree circuits, where the two \qfree rules \cref{rule:R-2} and \cref{rule:R-3} form the core of \RwUnx and may insert a compensating gate.
We then capture the mutual-control pattern via the \emph{circuit graph}---acyclic directed graph---from~\cite{unqomp}.
If we make all target edges ($\to$) of working qubits bidirected ($\leftrightarrow$), mutual control corresponds to a directed cycle that visits both ancilla and working-qubit nodes (cf.~\cref{fig:rwun-failure-example}). Call such a cycle an \emph{aw-cycle} if it uses only target and control edges ($\ctlarrow$) (anti-dependency edges can be ignored here). Since normalization gets stuck only at an aw-cycle, by proving that \cref{rule:R-1,rule:R-2,rule:R-3} cannot introduce an aw-cycle when none exists originally, we can show that normalization never gets stuck and hence succeeds, which yields the following sufficient condition.

\begin{proposition}\label{prop:sufficient_rwun_success}
  For a \qfree circuit $C$ and its circuit graph $G$, let $G^\leftrightarrow$ be obtained from $G$ by making all target edges on working qubits bidirected.
  If $G^\leftrightarrow$ has no aw-cycle, \RwUnx succeeds on $C$.
\end{proposition}

\paragraph{Potential and Challenge.}

\Cref{prop:complete-qfreenormal} suggests that uncomputation for \qfree circuits could be solved by finding an equivalent witness circuit in the required target form, from which both existence and synthesis follow directly. This motivates the design of a stronger rewriting procedure for discovering such a witness, potentially informed by complete rewriting systems (possibly with rules in~\cite{completeReversibleCircuitRules,completeQuantumRules,ExtAndSimpcompleteQuantumRules,minimalCompleteQuantumRules}) developed for other circuit tasks. However, important challenges remain: in many such settings, rewriting is guided by an explicit target or by two concrete circuits that must be transformed into a common form, whereas in our setting only a single circuit is given, and the desired witness structure is itself unknown.

\section{Evaluation}\label{sec:evaluation}

We implement \RwUnx as an algorithmic uncomputation procedure in Qiskit \cite{qiskit}. It takes a circuit as input and either outputs an uncomputed version upon success, or reports an error upon failure.
We use Qiskit’s built-in \texttt{AncillaRegister} type to explicitly mark ancillas. Our evaluation focuses on the following research questions (RQs):

\begin{itemize}
    \item \textbf{RQ1 (Applicability)}: Is our method applicable to a wide range of circuits?
    \item \textbf{RQ2 (Scalability)}: Is the method scalable enough for realistic circuit sizes?
\end{itemize}

\subsection{Evaluation Setting}

\paragraph{Implementation Details}

We implement \RwUnx with four heuristic strategies for deciding the uncomputation order of ancillas. \textsc{Sequential} uncomputes the ancillas in their original order. \textsc{Reverse} uncomputes them in the reverse order. \textsc{Jointly} uncomputes all ancillas jointly, most closely matching the procedure in \cref{alg:rewrite-normalization} and aligning with \cref{prop:sufficient_rwun_success}. \textsc{Lifetime} selects the ancilla with the shortest remaining lifetime to uncompute first, aiming to be more efficient. The first three strategies are designed for clean uncomputation, while \textsc{Lifetime} is applied to dirty uncomputation for \qfree circuits.
Moreover, \RwUnx cancels adjacent pairs of inverse gates and can reorder commuting gates (\cref{rule:R-1}) to make such cancellations possible.

\paragraph{Benchmarks}
\begin{itemize}
    \item \textit{Benchmark 1: practical circuits}. We include benchmark circuits from \Unqompx~\cite{unqomp} and \Reqompx~\cite{reqomp}, and further provide dirty implementations of some circuits without uncomputation. 
    We also introduce the following benchmark circuits: Incrementer, HighestBitConstAdder~\cite{factoring2n+2}, Gidney's Incrementer~\cite{gidney2015}, ConditionalMCS~\cite{conditionalcleanqubits}, and RiseConditionalMCS~\cite{riseofconditional}.
    \item \textit{Benchmark 2: randomly generated \qfree circuits}. 
    To further evaluate the applicability and scalability of our method on \qfree circuits, we generate random batches of \qfree circuits (seed 42). 
    Each batch contains 100 \emph{uncomputable} circuits with fixed width and gate count, using gates from $\{\tX, \tCNOT, \tToffoli\}$. 
    For a given width and gate count, we sample each gate uniformly from this set and then sample its operand qubits uniformly at random (respecting the gate arity). 
    We use a SAT solver to decide whether a sampled circuit is uncomputable based on \cref{prop:juequalrkf}, and keep sampling until we collect 100 uncomputable circuits.
    \item \textit{Benchmark 3: randomly generated quantum circuits}. 
    We also generate random batches of general quantum circuits (seed 42) using the same sampling strategy, with gates from $\{\tX, \tZ, \tH, \tS, \tT, \tCNOT, \tToffoli\}$. 
    We additionally enforce that $\tH$ never acts on an ancilla. 
    Unlike Benchmark~2, we do not filter these circuits to be uncomputable, since deciding uncomputability for general quantum circuits is substantially more expensive.
\end{itemize}

\paragraph{Baseline}
Our approach is most closely related to \Unqompx~\cite{unqomp} and \Reqompx~\cite{reqomp}, which also ensure safety for synthesized outputs. 
Since \Reqompx has been shown to offer broader applicability than \Unqompx~\cite{reqomp}, we adopt it as our baseline for comparison.
Note that \RwUnx is synthesis-oriented and operates at a fine-grained level, whereas Silq~\cite{silq} is a high-level language without any guidance for practical synthesis, making it not the preferred baseline for our evaluation.

\paragraph{Metrics}

To evaluate the applicability of \RwUnx, we introduce a dependency-complexity metric, \bAwDep, defined on the circuit graph (see details in appendix). 
Intuitively, \bAwDep captures how difficult it is to uncompute an ancilla (in the following, $a$) by measuring how often the qubits it depends on (here, $r$ and $q$) are modified \emph{in a way that can causally propagate back to the ancilla}.
The propagation-back requirement is crucial. A downstream modification on a depended qubit can be irrelevant if it lies in a region that is independent of the ancilla and thus can be safely ignored during uncomputation (e.g., the $\tX[r]$ outside the solid box). In contrast, a modification matters when it can later influence the ancilla again: in the figure, the $\tCNOT[r,q]$ may affect $q$, and the subsequent $\tCNOT[q,a]$ (dashed box) can feed this changed state back to $a$. 
Accordingly, the $\tCNOT[r,q]$ contributes $+1$ to \bAwDep, whereas the $\tX[r]$ contributes $+0$.

\begin{wrapfigure}{r}{6.2cm}
    \centering
    \begin{minipage}[t]{0.47\textwidth}
        \begin{quantikz}[row sep=0.15cm]
            \lstick{$r$}          & \ctrl{2}\gategroup[wires=4,steps=5,style={solid,rounded corners,draw=blue, fill=none, inner xsep=1pt},label style={label position=above,anchor=north,xshift=-0.1cm,yshift=0.4cm}]{}  & \qw       & \bluectrl{1}\gategroup[wires=1,steps=1,style={solid,rounded corners,draw=none, fill=none, inner xsep=2pt},label style={label position=above,anchor=north,xshift=0.0cm,yshift=-1.35cm}]{{\rm \blue{+1}}}    & \qw             &  \qw      & \gate[style={draw=blue, text=blue, solid}]{\blue{X}}\gategroup[wires=1,steps=1,style={solid,rounded corners,draw=none, fill=none, inner xsep=2pt},label style={label position=above,anchor=north,xshift=0.0cm,yshift=0.0cm}]{{\rm \blue{+0}}}   & \qw \\
            \lstick{$q$}          & \qw       & \ctrl{1}  & \bluetarg       & \ctrl{1}\gategroup[wires=2,steps=1,style={dashed,rounded corners,draw=red, fill=none, inner xsep=2pt},label style={label position=above,anchor=north,xshift=-0.1cm,yshift=0.4cm}]{}        &  \qw      & \qw              & \qw \\
            \lstick{$a$}          & \targ{}   & \targ{}   & \qw             & \targ{}         &  \ctrl{1} & \qw              & \qw \\
            \lstick{$p$}          & \qw       & \qw       & \qw             & \qw             &  \targ{}  & \qw              & \qw
        \end{quantikz}
    \end{minipage}
    \caption{Illustration of \bAwDep computation.}
    \Description{
    A circuit diagram illustrating the computation of Aw-Dep for a
    selected gate.
    }
    \label{fig:aw-dep-case}
\end{wrapfigure}

We group all reported metrics as follows:

\begin{itemize}
    \item \textbf{Aw-Dep}: the maximum dependency-complexity among all ancillas of a circuit, since non-uncomputability is typically driven by the “hardest” ancilla.
    \item \textbf{Aw-Cyc}: whether an aw-cycle exists in a circuit, as described in \cref{prop:sufficient_rwun_success}.
    \item \textbf{Circuit Width} and \textbf{Size}.
\end{itemize}

\paragraph{Execution Environment} All experiments are conducted on a Linux environment under WSL2, running on an Intel Core i9-14900HX CPU with 64\,GB of RAM.

\subsection{Results and Analysis}

\begin{table}[tb]
    \centering
    \small
    \caption{Largest scale parameter ($n$) of circuits with ancillas uncomputed by \RwUnx and \Reqompx within 30 seconds (\textbf{RQ2}), sorted in increasing order of \bAwDep. The concrete parameters are reported in appendix. $61^{\boldsymbol{\circlearrowright}}$ indicates existence of aw-cycle. We stop scaling at $n=1000$. \ding{55} indicates failure; fewer failures imply better applicability (\textbf{RQ1})}
    \label{tab:practical-circuits}
    \begin{adjustbox}{max width=\textwidth}
        \begin{tabular}{lccccccc}
        \toprule
        \multirow{2}{*}{\textbf{Circuit}} & \multirow{2}{*}{\textbf{Param}($n$)} & \multirow{2}{*}{\bAwDep} & \multicolumn{4}{c}{\textbf{\RwUnx}} & \multirow{2}{*}{\textbf{\Reqompx}} \\
        \cmidrule(lr){4-7}
                                &        &    & \textsc{Sequential} & \textsc{Reverse} & \textsc{Jointly} & \textsc{Lifetime} \\
        \midrule
        Clean IntegerComparator & controls & 0 & 10 & 130 & 20 & 20 & \textbf{180} \\
        Clean MCX & controls & 0 & 30 & 480 & 240 & \textbf{1000+} & 190 \\
        MCRY($\pi/2$) & controls & 0 & \textbf{1000+} & \textbf{1000+} & \textbf{1000+} & \textbf{1000+} & 14 \\
        Clean Incrementer & operand & 0 & 9 & 140 & 89 & \textbf{350} & 170 \\
        Deutsch--Jozsa & controls & 0 & \ding{55} & \ding{55} & \ding{55} & \textbf{1000+} & 170 \\
        Grover's algorithm~\cite{grover} & controls & 0 & \ding{55} & \ding{55} & \ding{55} & \textbf{15} & 13 \\
        PiecewiseLinearR & state qubits & 0 & \ding{55} & \ding{55} & \ding{55} & \ding{55} & \textbf{9} \\
        PolynomialPauliR & state qubits & 0 & \ding{55} & \ding{55} & \ding{55} & \ding{55} & \textbf{10} \\
        WeightedAdder & state qubits & 0 & 3 & 3 & 3 & 84 & \textbf{219} \\
        Dirty MCX & controls & 1 & 70 & 70 & 100 & \textbf{1000+} & \textbf{1000+} \\
        Clean Adder & per operand & 1 & 6 & 80 & 9 & 5 & \textbf{160} \\
        Dirty IntegerComparator & controls & 4 & 13 & 13 & 19 & 15 & \textbf{850} \\
        HighestBitConstAdder & operand & 4 & 9 & 8 & 12 & 9 & \textbf{970} \\
        RiseConditionalCleanMCS & controls & 5 & \textbf{11} & 9 & 7 & 9 & \ding{55} \\
        ConditionalCleanMCS & controls & 6 & \textbf{1000+} & \textbf{1000+} & \textbf{1000+} & \textbf{1000+} & \ding{55} \\
        RiseConditionalDirtyMCS & controls & 7 & 7 & 7 & 7 & \textbf{9} & \ding{55} \\
        ConditionalDirtyMCS & controls & 11 & \textbf{1000+} & \textbf{1000+} & \textbf{1000+} & \textbf{1000+} & \ding{55} \\
        Multiplier & per operand & 27 & 3 & 3 & 3 & 6 & \textbf{38} \\
        Gidney's Incrementer & operand & $61^{\boldsymbol{\circlearrowright}}$ & \ding{55} & \ding{55} & \textbf{45} & \ding{55} & \ding{55} \\
        Dirty Adder & per operand & 601 & 6 & 6 & 5 & \textbf{9} & \ding{55} \\
        Dirty Incrementer & operand & 803 & 10 & 10 & 10 & \textbf{90} & \ding{55} \\

        \bottomrule
    \end{tabular}
    \end{adjustbox}
\end{table}

\begin{figure}[t]
    \centering
    \vspace{-0.5cm}
    \includegraphics[page=1,width=1.0\linewidth]{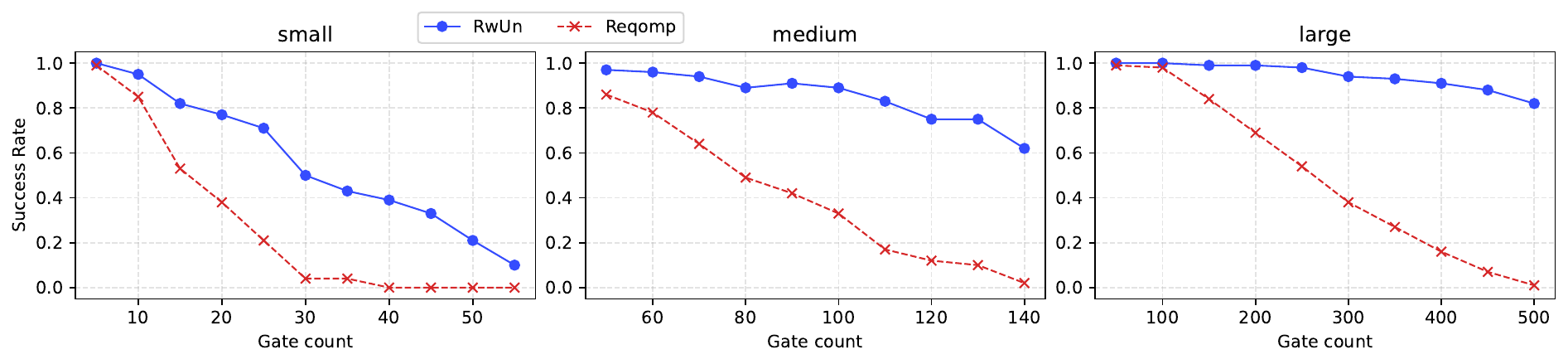}
    \includegraphics[page=2,width=1.0\linewidth]{fig/success_qfree.pdf}
    \includegraphics[page=3,width=1.0\linewidth]{fig/success_qfree.pdf}
    \includegraphics[page=4,width=1.0\linewidth]{fig/success_qfree.pdf}
    \vspace{-0.5cm}
    \caption{
    Uncomputation results of \RwUnx (\textsc{Jointly}) and \Reqompx on random \qfree circuits.
    The three columns correspond to increasing circuit width: small (5 working qubits + 5 ancillas), medium (40+40), and large (200+200).
    The four rows (top to bottom) report: (1) overall success rate (number of successful runs out of 100 uncomputable circuits);
    (2) \bAwDep, reported as the average \bAwDep over successful runs;
    (3) success rate on the subset of these 100 circuits without aw-cycles;
    and (4) success rate on the subset of these 100 circuits with aw-cycles.
    In Row~2, the solid line with dot markers reports
    the average \bAwDep for circuits handled only by
    \RwUnx, while the dashed line with x markers reports
    the average for circuits handled by both tools.
    Points with empty subsets are omitted.
    In the other three rows, solid lines denote \RwUnx and dashed lines denote \Reqompx.
    Higher values indicate better applicability (\textbf{RQ1}).
    }
    \Description{
    Twelve line plots comparing RwUn and Reqomp on random qfree circuits of
    increasing widths and sizes, showing success rates and aw-dependencies. RwUn generally matches or outperforms Reqomp.
    }
    \label{fig:random-qfree}
\end{figure}

\begin{figure}[htbp]
    \centering
    \includegraphics[width=1.0\linewidth]{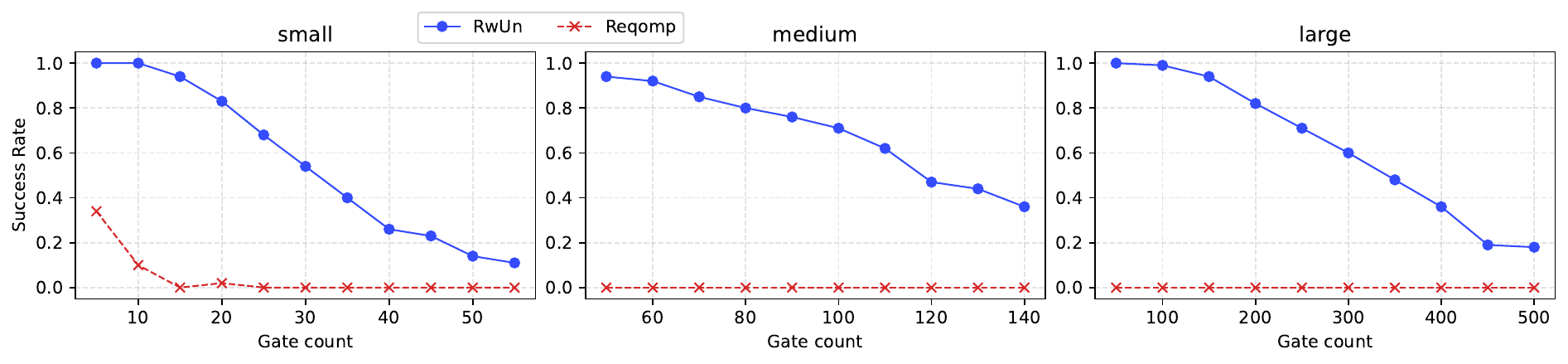}
    \vspace{-0.9cm}
    \caption{
    Success rate of \RwUnx (\textsc{Lifetime}) and \Reqompx on random quantum circuits (\textbf{RQ1}).
    }
    \Description{
    Line plots comparing the success rates of RwUn and Reqomp on random quantum circuits.
    }
    \label{fig:random-quantum}
    \vspace{-0.4cm}
\end{figure}

\paragraph{Applicability}

Table~\ref{tab:practical-circuits} summarizes the applicability of \RwUnx in comparison with \Reqompx on Benchmark~1. With its full set of heuristics, \RwUnx fails on 2/21 cases, whereas \Reqompx fails on 7/21 of them. The two cases where \RwUnx fails, PiecewiseLinearR and PolynomialPauliR, contain controlled-rotation gates outside the gate set covered by our rewriting rules and are rejected during the initial gate-set check. We also observe that \Reqompx fails on all tested cases with \bAwDep \(\ge 5\), except for Multiplier.

Figure~\ref{fig:random-qfree} summarizes the applicability of \RwUnx (\textsc{Jointly}) and \Reqompx on Benchmark~2, and shows an overall advantage of \RwUnx (Row~1).
Row~2 helps explain this gap: circuits that are uncomputed only by \RwUnx typically have higher dependency complexity than those successfully handled by both tools.
Moreover, \RwUnx (\textsc{Jointly}) covers \qfree
circuits without aw-cycles
(cf.~\cref{prop:sufficient_rwun_success}).
Consistently, Row~3 reports a success rate of 1
for \RwUnx on every nonempty no-aw-cycle subset.
Batches with an empty no-aw-cycle subset have
no defined conditional success rate and are omitted.
Even on the complementary subset with aw-cycles (Row~4), \RwUnx still achieves a generally higher success rate than \Reqompx. 
\Cref{fig:random-quantum} illustrates the applicability of \RwUnx (\textsc{Lifetime}) in comparison with \Reqompx on Benchmark~3, where phase-type gates are allowed to act on ancillas—an extension that is outside the scope of \Reqompx.

The failure conditions are discussed in \cref{sec:rwun-limits}: the presence of aw-cycles for \qfree circuits, and $\tH$ gates additionally for general quantum circuits. However, the cancellation-based optimization in \RwUnx can sometimes eliminate an aw-cycle by cancelling some of its constituent gates, thereby breaking the cycle. This explains cases such as Gidney’s Incrementer in \cref{tab:practical-circuits}, and also why \RwUnx achieves a higher success rate on random \qfree circuits than what would be expected if it were strictly constrained by the aw-cycle existence rate.

The applicability of \RwUnx is also limited by the gate set covered by its rewriting rules, as illustrated by the two failures in Benchmark~1.
One possible strategy is to decompose unsupported gates into supported ones, as we did for the MCRY($\pi/2$) case.
However, the choice of decomposition may affect the success of rewriting. Different decompositions of the same operation may introduce different  structures, potentially creating or avoiding patterns such as aw-cycles that can obstruct normalization.
An alternative, potentially more practical strategy is to extend the rewriting rules to cover additional gate families.

\paragraph{Scalability}

As a reference, \Reqompx completes all its benchmarks except Clean MCX within 5 seconds, and requires about 30 seconds for Clean MCX. Therefore, we report the largest circuit scale that can be executed within 30 seconds for Benchmark 1.

From the results for Benchmark 1 in \cref{tab:practical-circuits}, \RwUnx with the \textsc{Lifetime} heuristic is scalable in many cases. 
\RwUnx becomes inefficient for instances where too many intermediate gates are introduced while rewriting. 
\Reqompx, on the other hand, tends to be consistently scalable in most cases where it is applicable and outperforms \RwUnx on six cases.

For Benchmarks~2 and~3, the average uncomputation time at each batch point is computed only over successful runs (failures are excluded).
On Benchmark~2, the maximum average time is 0.55\,s for \Reqompx and 151.92\,s for \RwUnx (\textsc{Jointly}).
On Benchmark~3, the averages are below 1\,s in most cases, and the maximum is 4.92\,s for \RwUnx (\textsc{Lifetime}).

The main scalability bottleneck of \RwUnx comes from the intermediate gates introduced by the two \qfree rules. These rules increase the circuit size and, in the worst case, can lead to an exponential blow-up (informally).
In practice, \RwUnx scales well on circuits where these two rules are applied infrequently, and its simple optimization can also reduce the intermediate circuit size.

\section{Related Work}\label{relatedwork}

\paragraph{Early Attempts for Automatic Uncomputation}
Quipper~\cite{quipper} provides convenient constructs for inserting uncomputation blocks into user-defined \qfree circuits. Q\#~\cite{qpunch} extends this idea by allowing programmers to explicitly declare the use of dirty ancillas. SQUARE~\cite{square} focuses on the allocation and reclamation of clean ancillas to enable qubit reuse and reduce resource costs. However, these early attempts do not guarantee the correctness or safety of uncomputation. 

\paragraph{Safe Automatic Uncomputation for Clean Ancillas}
Silq \cite{silq} is the first quantum programming language that guarantees the existence of uncomputation via type checking.
Qurts \cite{qurts} utilizes the lifetime mechanism inspired by Rust's type system, to enforce that qubits controlling clean ancillas remain constant before uncomputation through type checking. Qurts supports flexible uncomputation strategies by integrating reversible pebble game, though it lacks compiler support.

\Unqompx~\cite{unqomp} is the first approach to automatically synthesize uncomputation within quantum circuits. It represents circuits as graphs and identifies appropriate positions to insert uncomputation gates. \Reqompx~\cite{reqomp} extends this by introducing a well-valued circuit graph, which more precisely tracks the state of clean ancillas and supports recycling. \Reqompx also focuses on recycling and reusing ancillas, achieving significant reductions in ancilla counts.
Qrisp~\cite{qrisp} builds upon \Unqompx to provide a high-level programming framework that supports automatic uncomputation. \cite{modularuncomp} proposes an intermediate representation (IR) and a modular algorithm for uncomputation synthesis.

\paragraph{Reduce Uncomputation Cost}
Measurement-based uncomputation (MBU)~\cite{Gidney2018halvingcostof, gidney2019windowedquantumarithmetic, gidney2019approximateencodedpermutationspiecewise, luongo2024measurementbaseduncomputationquantumcircuits, Kornerup2025tightboundsspooky} reduces cost probabilistically through measurement. However, MBU requires intricate circuit design and is unsuitable for dirty ancillas due to measurement collapse.
Sharma et al.~\cite{spare} propose SPARE, a rewrite-based optimizer that targets compute–uncompute sequences of clean ancillas. By restructuring these sequences and leveraging ancilla-state information, SPARE significantly reduces ancilla count, gate count, and circuit depth. Many studies have explored qubit reuse techniques to minimize qubit count~\cite{wirerecycling, complierqubitreuse, qubitreusecompilation, quantumcircuitresizing}. Jiang~\cite{qubitrecycling} introduces qubit dependency graphs (QDGs) as an abstraction for qubit recycling, proves the NP-hardness of the problem, and proposes a heuristic solver that achieves near-optimal recycling performance.

\paragraph{Verification of Safe Uncomputation}
REVS~\cite{revs} and ReVerC~\cite{reverc} are reversible circuit compilers that support automatic uncomputation for clean ancillas, with the latter being formally verified.
ReQwire~\cite{reqwire} presents methods for verifying that a manual uncomputation of clean ancillas is safe.
Su et al. \cite{su2024bibasedreasoningquantumprograms} model the allocation and deallocation of dirty ancillas as heap manipulations. They first develop a quantum BI-style logic to characterize heap operations, and then build on it to present a quantum separation logic that enables reasoning about the correct usage of dirty ancillas. Su et al. \cite{su2025borrowingdirtyqubitsquantum} formalize the semantics of dirty-qubit borrowing in quantum programming languages and introduce a notion of safe uncomputation for dirty qubits. They also propose an efficient algorithm for verifying safe uncomputation in certain quantum circuits. 
Recently, Li et al. \cite{Li2026AncillaSafety} verify clean and dirty ancilla safety through Pauli-X/Z commutativity checks.

\begin{acks}
We thank the anonymous reviewers for their constructive feedback, which significantly helped us improve the paper.
We are grateful to Bonan Su for the initial discussions and assistance with the submission, to Minbo Gao and Qifan Huang for helpful technical discussions, and to Zhenhao Li for assistance with the implementation.
This work was supported in part by the Beijing Major Science and Technology Project under Contract No. Z251100008125035.
This work was also supported by
Beijing Academy of Artificial Intelligence (BAAI).
\end{acks}

\section*{Data-Availability Statement}

The artifact accompanying \cref{sec:evaluation}, including the implementation of \RwUn, benchmark circuits, and scripts used for the evaluation, is publicly available on Zenodo \cite{liu_2026_22168225}.

\bibliography{references}

@article{deutsch,
    author = {Deutsch, David and Jozsa, Richard},
    title = {Rapid solution of problems by quantum computation},
    journal = {Proceedings of the Royal Society of London. Series A: Mathematical and Physical Sciences},
    volume = {439},
    number = {1907},
    pages = {553-558},
    year = {1992},
    month = {12},
    issn = {0962-8444},
    doi = {10.1098/rspa.1992.0167},
    url = {https://doi.org/10.1098/rspa.1992.0167},
    eprint = {https://royalsocietypublishing.org/rspa/article-pdf/439/1907/553/68698/rspa.1992.0167.pdf},
}

@article{grover,
  title = {Quantum Mechanics Helps in Searching for a Needle in a Haystack},
  author = {Grover, Lov K.},
  journal = {Phys. Rev. Lett.},
  volume = {79},
  issue = {2},
  pages = {325--328},
  numpages = {0},
  year = {1997},
  month = {Jul},
  publisher = {American Physical Society},
  doi = {10.1103/PhysRevLett.79.325},
  url = {https://link.aps.org/doi/10.1103/PhysRevLett.79.325}
}

@article{shoralgorithm,
author = {Shor, Peter W.},
title = {Polynomial-Time Algorithms for Prime Factorization and Discrete Logarithms on a Quantum Computer},
journal = {SIAM Journal on Computing},
volume = {26},
number = {5},
pages = {1484-1509},
year = {1997},
doi = {10.1137/S0097539795293172},

URL = { 
    
        https://doi.org/10.1137/S0097539795293172
    
    

},
eprint = { 
    
        https://doi.org/10.1137/S0097539795293172
    
    

}
}

@article{factoring2n+2,
author = {H\"{a}ner, Thomas and Roetteler, Martin and Svore, Krysta M.},
title = {Factoring using 2n + 2 qubits with Toffoli based modular multiplication},
year = {2017},
issue_date = {June 2017},
publisher = {Rinton Press, Incorporated},
address = {Paramus, NJ},
volume = {17},
number = {7–8},
issn = {1533-7146},
journal = {Quantum Info. Comput.},
month = jun,
pages = {673–684},
numpages = {12},
doi = {10.26421/QIC17.7-8-7}
}

@misc{factoringn+2,
      title={Factoring with n+2 clean qubits and n-1 dirty qubits}, 
      author={Craig Gidney},
      year={2018},
      eprint={1706.07884},
      archivePrefix={arXiv},
      primaryClass={quant-ph},
      url={https://arxiv.org/abs/1706.07884}, 
}

@article{Low2024tradingtgatesdirty,
  doi = {10.22331/q-2024-06-17-1375},
  url = {https://doi.org/10.22331/q-2024-06-17-1375},
  title = {Trading {T} gates for dirty qubits in state preparation and unitary synthesis},
  author = {Low, Guang Hao and Kliuchnikov, Vadym and Schaeffer, Luke},
  journal = {{Quantum}},
  issn = {2521-327X},
  publisher = {{Verein zur F{\"{o}}rderung des Open Access Publizierens in den Quantenwissenschaften}},
  volume = {8},
  pages = {1375},
  month = jun,
  year = {2024}
}

@InProceedings{quantumimplementationmininalT,
author="Huang, Zhenyu
and Zhang, Fuxin
and Lin, Dongdai",
editor="Fehr, Serge
and Fouque, Pierre-Alain",
title="Constructing Quantum Implementations with the Minimal T-depth or Minimal Width and Their Applications",
booktitle="Advances in Cryptology -- EUROCRYPT 2025",
year="2025",
publisher="Springer Nature Switzerland",
address="Cham",
pages="155--185",
isbn="978-3-031-91107-1",
doi="10.1007/978-3-031-91107-1_6"
}

@article{elementgates,
  title = {Elementary gates for quantum computation},
  author = {Barenco, Adriano and Bennett, Charles H. and Cleve, Richard and DiVincenzo, David P. and Margolus, Norman and Shor, Peter and Sleator, Tycho and Smolin, John A. and Weinfurter, Harald},
  journal = {Phys. Rev. A},
  volume = {52},
  issue = {5},
  pages = {3457--3467},
  numpages = {0},
  year = {1995},
  month = {Nov},
  publisher = {American Physical Society},
  doi = {10.1103/PhysRevA.52.3457},
  url = {https://link.aps.org/doi/10.1103/PhysRevA.52.3457}
}

@misc{gidney2015,
  author    = {Craig Gidney},
  title     = {Constructing Large Controlled Nots},
  howpublished = {\url{https://algassert.com/circuits/2015/06/05/Constructing-Large-Controlled-Nots.html}},
  note      = {Accessed: 2025-09-05},
  year      = {2015}
}

@misc{conditionalcleanqubits,
      title={Quantum circuit for multi-qubit Toffoli gate with optimal resource}, 
      author={Junhong Nie and Wei Zi and Xiaoming Sun},
      year={2024},
      eprint={2402.05053},
      archivePrefix={arXiv},
      primaryClass={quant-ph},
      url={https://arxiv.org/abs/2402.05053}, 
}

@article{riseofconditional,
   title={Rise of conditionally clean ancillae for efficient quantum circuit constructions},
   volume={9},
   ISSN={2521-327X},
   url={http://dx.doi.org/10.22331/q-2025-05-21-1752},
   DOI={10.22331/q-2025-05-21-1752},
   journal={Quantum},
   publisher={Verein zur Forderung des Open Access Publizierens in den Quantenwissenschaften},
   author={Khattar, Tanuj and Gidney, Craig},
   year={2025},
   month=may, pages={1752} 
}

@misc{zindorf2025efficientimplementationmulticontrolledquantum,
      title={Efficient Implementation of Multi-Controlled Quantum Gates}, 
      author={Ben Zindorf and Sougato Bose},
      year={2025},
      eprint={2404.02279},
      archivePrefix={arXiv},
      primaryClass={quant-ph},
      url={https://arxiv.org/abs/2404.02279}, 
}

@software{qiskit,
  author       = {Gadi Aleksandrowicz and
                  Thomas Alexander and
                  Panagiotis Barkoutsos and
                  Luciano Bello and
                  Yael Ben-Haim and
                  David Bucher and
                  Francisco Jose Cabrera-Hernández and
                  Jorge Carballo-Franquis and
                  Adrian Chen and
                  Chun-Fu Chen and
                  Jerry M. Chow and
                  Antonio D. Córcoles-Gonzales and
                  Abigail J. Cross and
                  Andrew Cross and
                  Juan Cruz-Benito and
                  Chris Culver and
                  Salvador De La Puente González and
                  Enrique De La Torre and
                  Delton Ding and
                  Eugene Dumitrescu and
                  Ivan Duran and
                  Pieter Eendebak and
                  Mark Everitt and
                  Ismael Faro Sertage and
                  Albert Frisch and
                  Andreas Fuhrer and
                  Jay Gambetta and
                  Borja Godoy Gago and
                  Juan Gomez-Mosquera and
                  Donny Greenberg and
                  Ikko Hamamura and
                  Vojtech Havlicek and
                  Joe Hellmers and
                  Łukasz Herok and
                  Hiroshi Horii and
                  Shaohan Hu and
                  Takashi Imamichi and
                  Toshinari Itoko and
                  Ali Javadi-Abhari and
                  Naoki Kanazawa and
                  Anton Karazeev and
                  Kevin Krsulich and
                  Peng Liu and
                  Yang Luh and
                  Yunho Maeng and
                  Manoel Marques and
                  Francisco Jose Martín-Fernández and
                  Douglas T. McClure and
                  David McKay and
                  Srujan Meesala and
                  Antonio Mezzacapo and
                  Nikolaj Moll and
                  Diego Moreda Rodríguez and
                  Giacomo Nannicini and
                  Paul Nation and
                  Pauline Ollitrault and
                  Lee James O'Riordan and
                  Hanhee Paik and
                  Jesús Pérez and
                  Anna Phan and
                  Marco Pistoia and
                  Viktor Prutyanov and
                  Max Reuter and
                  Julia Rice and
                  Abdón Rodríguez Davila and
                  Raymond Harry Putra Rudy and
                  Mingi Ryu and
                  Ninad Sathaye and
                  Chris Schnabel and
                  Eddie Schoute and
                  Kanav Setia and
                  Yunong Shi and
                  Adenilton Silva and
                  Yukio Siraichi and
                  Seyon Sivarajah and
                  John A. Smolin and
                  Mathias Soeken and
                  Hitomi Takahashi and
                  Ivano Tavernelli and
                  Charles Taylor and
                  Pete Taylour and
                  Kenso Trabing and
                  Matthew Treinish and
                  Wes Turner and
                  Desiree Vogt-Lee and
                  Christophe Vuillot and
                  Jonathan A. Wildstrom and
                  Jessica Wilson and
                  Erick Winston and
                  Christopher Wood and
                  Stephen Wood and
                  Stefan Wörner and
                  Ismail Yunus Akhalwaya and
                  Christa Zoufal},
  title        = {Qiskit: An Open-source Framework for Quantum
                   Computing
                  },
  month        = jan,
  year         = 2019,
  publisher    = {Zenodo},
  version      = {0.7.2},
  doi          = {10.5281/zenodo.2562111},
  url          = {https://doi.org/10.5281/zenodo.2562111},
}

@article{reqwire,
   title={ReQWIRE: Reasoning about Reversible Quantum Circuits},
   volume={287},
   ISSN={2075-2180},
   url={http://dx.doi.org/10.4204/EPTCS.287.17},
   DOI={10.4204/eptcs.287.17},
   journal={Electronic Proceedings in Theoretical Computer Science},
   publisher={Open Publishing Association},
   author={Rand, Robert and Paykin, Jennifer and Lee, Dong-Ho and Zdancewic, Steve},
   year={2019},
   month=jan, pages={299–312} }

@article{quipper,
author = {Green, Alexander S. and Lumsdaine, Peter LeFanu and Ross, Neil J. and Selinger, Peter and Valiron, Beno\^{\i}t},
title = {Quipper: a scalable quantum programming language},
year = {2013},
issue_date = {June 2013},
publisher = {Association for Computing Machinery},
address = {New York, NY, USA},
volume = {48},
number = {6},
issn = {0362-1340},
url = {https://doi.org/10.1145/2499370.2462177},
doi = {10.1145/2499370.2462177},
journal = {SIGPLAN Not.},
month = jun,
pages = {333–342},
numpages = {10}
}

@inproceedings{qpunch,
author = {Svore, Krysta and Geller, Alan and Troyer, Matthias and Azariah, John and Granade, Christopher and Heim, Bettina and Kliuchnikov, Vadym and Mykhailova, Mariia and Paz, Andres and Roetteler, Martin},
title = {Q\#: Enabling Scalable Quantum Computing and Development with a High-level DSL},
year = {2018},
isbn = {9781450363556},
publisher = {Association for Computing Machinery},
address = {New York, NY, USA},
url = {https://doi.org/10.1145/3183895.3183901},
doi = {10.1145/3183895.3183901},
articleno = {7},
numpages = {10},
location = {Vienna, Austria},
series = {RWDSL2018}
}

@INPROCEEDINGS{square,
  author={Ding, Yongshan and Wu, Xin-Chuan and Holmes, Adam and Wiseth, Ash and Franklin, Diana and Martonosi, Margaret and Chong, Frederic T.},
  booktitle={2020 ACM/IEEE 47th Annual International Symposium on Computer Architecture (ISCA)}, 
  title={SQUARE: Strategic Quantum Ancilla Reuse for Modular Quantum Programs via Cost-Effective Uncomputation}, 
  year={2020},
  volume={},
  number={},
  pages={570-583},
  doi={10.1109/ISCA45697.2020.00054}}

@inproceedings{silq,
author = {Bichsel, Benjamin and Baader, Maximilian and Gehr, Timon and Vechev, Martin},
title = {Silq: a high-level quantum language with safe uncomputation and intuitive semantics},
year = {2020},
isbn = {9781450376136},
publisher = {Association for Computing Machinery},
address = {New York, NY, USA},
url = {https://doi.org/10.1145/3385412.3386007},
doi = {10.1145/3385412.3386007},
booktitle = {Proceedings of the 41st ACM SIGPLAN Conference on Programming Language Design and Implementation},
pages = {286–300},
numpages = {15},
location = {London, UK},
series = {PLDI 2020}
}

@inproceedings{unqomp,
author = {Paradis, Anouk and Bichsel, Benjamin and Steffen, Samuel and Vechev, Martin},
title = {Unqomp: synthesizing uncomputation in Quantum circuits},
year = {2021},
isbn = {9781450383912},
publisher = {Association for Computing Machinery},
address = {New York, NY, USA},
url = {https://doi.org/10.1145/3453483.3454040},
doi = {10.1145/3453483.3454040},
booktitle = {Proceedings of the 42nd ACM SIGPLAN International Conference on Programming Language Design and Implementation},
pages = {222–236},
numpages = {15},
location = {Virtual, Canada},
series = {PLDI 2021}
}

@article{reqomp,
  doi = {10.22331/q-2024-02-19-1258},
  url = {https://doi.org/10.22331/q-2024-02-19-1258},
  title = {Reqomp: {S}pace-constrained {U}ncomputation for {Q}uantum {C}ircuits},
  author = {Paradis, Anouk and Bichsel, Benjamin and Vechev, Martin},
  journal = {{Quantum}},
  issn = {2521-327X},
  publisher = {{Verein zur F{\"{o}}rderung des Open Access Publizierens in den Quantenwissenschaften}},
  volume = {8},
  pages = {1258},
  month = feb,
  year = {2024}
}

@article{modularuncomp,
author = {Venev, Hristo and Gehr, Timon and Dimitrov, Dimitar and Vechev, Martin},
title = {Modular Synthesis of Efficient Quantum Uncomputation},
year = {2024},
issue_date = {October 2024},
publisher = {Association for Computing Machinery},
address = {New York, NY, USA},
volume = {8},
number = {OOPSLA2},
url = {https://doi.org/10.1145/3689785},
doi = {10.1145/3689785},
journal = {Proc. ACM Program. Lang.},
month = oct,
articleno = {345},
numpages = {28}
}

@article{qurts,
author = {Hirata, Kengo and Heunen, Chris},
title = {Qurts: Automatic Quantum Uncomputation by Affine Types with Lifetime},
year = {2025},
issue_date = {January 2025},
publisher = {Association for Computing Machinery},
address = {New York, NY, USA},
volume = {9},
number = {POPL},
url = {https://doi.org/10.1145/3704842},
doi = {10.1145/3704842},
journal = {Proc. ACM Program. Lang.},
month = jan,
articleno = {6},
numpages = {28}
}

@misc{su2024bibasedreasoningquantumprograms,
      title={BI-based Reasoning about Quantum Programs with Heap Manipulations}, 
      author={Bonan Su and Li Zhou and Yuan Feng and Mingsheng Ying},
      year={2024},
      eprint={2409.10153},
      archivePrefix={arXiv},
      primaryClass={quant-ph},
      url={https://arxiv.org/abs/2409.10153}, 
}

@inproceedings{su2025borrowingdirtyqubitsquantum, author = {Su, Bonan and Zhou, Li and Feng, Yuan and Ying, Mingsheng}, title = {Borrowing Dirty Qubits in Quantum Programs}, year = {2026}, isbn = {9798400723599}, publisher = {Association for Computing Machinery}, address = {New York, NY, USA}, url = {https://doi.org/10.1145/3779212.3790134}, doi = {10.1145/3779212.3790134}, booktitle = {Proceedings of the 31st ACM International Conference on Architectural Support for Programming Languages and Operating Systems, Volume 2}, pages = {274–289}, numpages = {16}, location = {USA}, series = {ASPLOS '26} }

@article{qubitrecycling,
author = {Jiang, Hanru},
title = {Qubit Recycling Revisited},
year = {2024},
issue_date = {June 2024},
publisher = {Association for Computing Machinery},
address = {New York, NY, USA},
volume = {8},
number = {PLDI},
url = {https://doi.org/10.1145/3656428},
doi = {10.1145/3656428},
month = jun,
articleno = {198},
numpages = {24}
}

@article{spare,
author = {Sharma, Ritvik and Achour, Sara},
title = {Optimizing Ancilla-Based Quantum Circuits with SPARE},
year = {2025},
issue_date = {June 2025},
publisher = {Association for Computing Machinery},
address = {New York, NY, USA},
volume = {9},
number = {PLDI},
url = {https://doi.org/10.1145/3729253},
doi = {10.1145/3729253},
journal = {Proc. ACM Program. Lang.},
month = jun,
articleno = {154},
numpages = {25}
}

@book{Nielsen_Chuang_2010, place={Cambridge}, title={Quantum Computation and Quantum Information: 10th Anniversary Edition}, publisher={Cambridge University Press}, author={Nielsen, Michael A. and Chuang, Isaac L.}, year={2010}}

@InProceedings{reversiblecomputing,
author="Toffoli, Tommaso",
editor="de Bakker, Jaco
and van Leeuwen, Jan",
title="Reversible computing",
booktitle="Automata, Languages and Programming",
year="1980",
publisher="Springer Berlin Heidelberg",
address="Berlin, Heidelberg",
pages="632--644",
isbn="978-3-540-39346-7"
}

@article{conservativelogic,
  author  = {Edward Fredkin and Tommaso Toffoli},
  title   = {Conservative Logic},
  journal = {International Journal of Theoretical Physics},
  volume  = {21},
  number  = {3-4},
  pages   = {219--253},
  year    = {1982},
  doi     = {10.1007/BF01857727}
}

@inproceedings{cooktheorem,
author = {Cook, Stephen A.},
title = {The complexity of theorem-proving procedures},
year = {1971},
isbn = {9781450374644},
publisher = {Association for Computing Machinery},
address = {New York, NY, USA},
url = {https://doi.org/10.1145/800157.805047},
doi = {10.1145/800157.805047},
booktitle = {Proceedings of the Third Annual ACM Symposium on Theory of Computing},
pages = {151–158},
numpages = {8},
location = {Shaker Heights, Ohio, USA},
series = {STOC '71}
}

@article{strongequivalence,
author = {Jordan, Stephen P.},
title = {Strong equivalence of reversible circuits is coNP-complete},
year = {2014},
issue_date = {November 2014},
publisher = {Rinton Press, Incorporated},
address = {Paramus, NJ},
volume = {14},
number = {15–16},
issn = {1533-7146},
journal = {Quantum Info. Comput.},
month = nov,
pages = {1302–1307},
numpages = {6},
doi={10.26421/QIC14.15-16-3}
}

@article{quantumbooleanfunctions,
  title={Quantum boolean functions},
  author={Ashley Montanaro and Tobias J. Osborne},
  journal={Chic. J. Theor. Comput. Sci.},
  year={2008},
  volume={2010},
  url={https://api.semanticscholar.org/CorpusID:15370057}
}

@article{boundedwidthpolysizebranching,
title = {Bounded-width polynomial-size branching programs recognize exactly those languages in NC1},
journal = {Journal of Computer and System Sciences},
volume = {38},
number = {1},
pages = {150-164},
year = {1989},
issn = {0022-0000},
doi = {https://doi.org/10.1016/0022-0000(89)90037-8},
url = {https://www.sciencedirect.com/science/article/pii/0022000089900378},
author = {David A. Barrington}
}

@article{Gidney2018halvingcostof,
  doi = {10.22331/q-2018-06-18-74},
  url = {https://doi.org/10.22331/q-2018-06-18-74},
  title = {Halving the cost of quantum addition},
  author = {Gidney, Craig},
  journal = {{Quantum}},
  issn = {2521-327X},
  publisher = {{Verein zur F{\"{o}}rderung des Open Access Publizierens in den Quantenwissenschaften}},
  volume = {2},
  pages = {74},
  month = jun,
  year = {2018}
}

@misc{gidney2019approximateencodedpermutationspiecewise,
      title={Approximate encoded permutations and piecewise quantum adders}, 
      author={Craig Gidney},
      year={2019},
      eprint={1905.08488},
      archivePrefix={arXiv},
      primaryClass={quant-ph},
      url={https://arxiv.org/abs/1905.08488}, 
}

@misc{gidney2019windowedquantumarithmetic,
      title={Windowed quantum arithmetic}, 
      author={Craig Gidney},
      year={2019},
      eprint={1905.07682},
      archivePrefix={arXiv},
      primaryClass={quant-ph},
      url={https://arxiv.org/abs/1905.07682}, 
}

@INPROCEEDINGS{luongo2024measurementbaseduncomputationquantumcircuits,
  author={Luongo, Alessandro and Miti, Antonio Michele and Narasimhachar, Varun and Sireesh, Adithya},
  booktitle={2025 62nd ACM/IEEE Design Automation Conference (DAC)}, 
  title={Measurement-based uncomputation of quantum circuits for modular arithmetic}, 
  year={2025},
  volume={},
  number={},
  pages={1-7},
  doi={10.1109/DAC63849.2025.11132981}}

@article{Kornerup2025tightboundsspooky,
  doi = {10.22331/q-2025-02-18-1636},
  url = {https://doi.org/10.22331/q-2025-02-18-1636},
  title = {Tight {B}ounds on the {S}pooky {P}ebble {G}ame: {R}ecycling {Q}ubits with {M}easurements},
  author = {Kornerup, Niels and Sadun, Jonathan and Soloveichik, David},
  journal = {{Quantum}},
  issn = {2521-327X},
  publisher = {{Verein zur F{\"{o}}rderung des Open Access Publizierens in den Quantenwissenschaften}},
  volume = {9},
  pages = {1636},
  month = feb,
  year = {2025}
}

@inproceedings{architectureuncomputation,
author = {Schuchman, Ethan and Vijaykumar, T. N.},
title = {A program transformation and architecture support for quantum uncomputation},
year = {2006},
isbn = {1595934510},
publisher = {Association for Computing Machinery},
address = {New York, NY, USA},
url = {https://doi.org/10.1145/1168857.1168889},
doi = {10.1145/1168857.1168889},
booktitle = {Proceedings of the 12th International Conference on Architectural Support for Programming Languages and Operating Systems},
pages = {252–263},
numpages = {12},
location = {San Jose, California, USA},
series = {ASPLOS XII}
}

@article{wirerecycling,
  title = {Wire recycling for quantum circuit optimization},
  author = {Paler, Alexandru and Wille, Robert and Devitt, Simon J.},
  journal = {Phys. Rev. A},
  volume = {94},
  issue = {4},
  pages = {042337},
  numpages = {8},
  year = {2016},
  month = {Oct},
  publisher = {American Physical Society},
  doi = {10.1103/PhysRevA.94.042337},
  url = {https://link.aps.org/doi/10.1103/PhysRevA.94.042337}
}

@inproceedings{complierqubitreuse,
author = {Hua, Fei and Jin, Yuwei and Chen, Yanhao and Vittal, Suhas and Krsulich, Kevin and Bishop, Lev S. and Lapeyre, John and Javadi-Abhari, Ali and Zhang, Eddy Z.},
title = {CaQR: A Compiler-Assisted Approach for Qubit Reuse through Dynamic Circuit},
year = {2023},
isbn = {9781450399180},
publisher = {Association for Computing Machinery},
address = {New York, NY, USA},
url = {https://doi.org/10.1145/3582016.3582030},
doi = {10.1145/3582016.3582030},
booktitle = {Proceedings of the 28th ACM International Conference on Architectural Support for Programming Languages and Operating Systems, Volume 3},
pages = {59–71},
numpages = {13},
location = {Vancouver, BC, Canada},
series = {ASPLOS 2023}
}

@article{qubitreusecompilation,
  title = {Qubit-Reuse Compilation with Mid-Circuit Measurement and Reset},
  author = {DeCross, Matthew and Chertkov, Eli and Kohagen, Megan and Foss-Feig, Michael},
  journal = {Phys. Rev. X},
  volume = {13},
  issue = {4},
  pages = {041057},
  numpages = {22},
  year = {2023},
  month = {Dec},
  publisher = {American Physical Society},
  doi = {10.1103/PhysRevX.13.041057},
  url = {https://link.aps.org/doi/10.1103/PhysRevX.13.041057}
}

@misc{quantumcircuitresizing,
      title={Quantum Circuit Resizing}, 
      author={Movahhed Sadeghi and Soheil Khadirsharbiyani and Mahmut Taylan Kandemir},
      year={2022},
      eprint={2301.00720},
      archivePrefix={arXiv},
      primaryClass={cs.ET},
      url={https://arxiv.org/abs/2301.00720}, 
}

@inbook{Remaud2025,
   title={Ancilla-Free Quantum Adder with Sublinear Depth},
   ISBN={9783031970634},
   ISSN={1611-3349},
   url={http://dx.doi.org/10.1007/978-3-031-97063-4_11},
   DOI={10.1007/978-3-031-97063-4_11},
   booktitle={Reversible Computation},
   publisher={Springer Nature Switzerland},
   author={Remaud, Maxime and Vandaele, Vivien},
   year={2025},
   pages={137–154} }

@article{takahashi2010,
author = {Takahashi, Yasuhiro and Tani, Seiichiro and Kunihiro, Noboru},
title = {Quantum addition circuits and unbounded fan-out},
year = {2010},
issue_date = {September 2010},
publisher = {Rinton Press, Incorporated},
address = {Paramus, NJ},
volume = {10},
number = {9},
issn = {1533-7146},
journal = {Quantum Info. Comput.},
month = sep,
pages = {872–890},
numpages = {19},
doi = {10.26421/QIC10.9-10-12}
}

@article{lawsofquantumprog,
author = {Ying, Mingsheng and Zhou, Li and Barthe, Gilles},
title = {Laws of Quantum Programming},
year = {2025},
publisher = {Association for Computing Machinery},
address = {New York, NY, USA},
issn = {1049-331X},
url = {https://doi.org/10.1145/3765903},
doi = {10.1145/3765903},
note = {Just Accepted},
journal = {ACM Trans. Softw. Eng. Methodol.},
month = sep
}

@inproceedings{qrisp,
author = {Seidel, Raphael and Tcholtchev, Nikolay and Bock, Sebastian and Hauswirth, Manfred},
title = {Uncomputation in the Qrisp High-Level Quantum Programming Framework},
year = {2023},
isbn = {978-3-031-38099-0},
publisher = {Springer-Verlag},
address = {Berlin, Heidelberg},
url = {https://doi.org/10.1007/978-3-031-38100-3_11},
doi = {10.1007/978-3-031-38100-3_11},
pages = {150–165},
numpages = {16},
location = {Giessen, Germany}
}

@misc{ying2023quantumrecursiveprogrammingquantum,
      title={Quantum Recursive Programming with Quantum Case Statements}, 
      author={Mingsheng Ying and Zhicheng Zhang},
      year={2023},
      eprint={2311.01725},
      archivePrefix={arXiv},
      primaryClass={cs.PL},
      url={https://arxiv.org/abs/2311.01725}, 
}

@article{table-n-CNOT,
	author = {Claudon, Baptiste and Zylberman, Julien and Feniou, C{\'e}sar and Debbasch, Fabrice and Peruzzo, Alberto and Piquemal, Jean-Philip},
	date = {2024/07/13},
	doi = {10.1038/s41467-024-50065-x},
	id = {Claudon2024},
	isbn = {2041-1723},
	journal = {Nature Communications},
	number = {1},
	pages = {5886},
	title = {Polylogarithmic-depth controlled-NOT gates without ancilla qubits},
	url = {https://doi.org/10.1038/s41467-024-50065-x},
	volume = {15},
	year = {2024}}

@misc{draper2000addition,
      title={Addition on a Quantum Computer}, 
      author={Thomas G. Draper},
      year={2000},
      eprint={quant-ph/0008033},
      archivePrefix={arXiv},
      primaryClass={quant-ph},
      url={https://arxiv.org/abs/quant-ph/0008033}, 
}

@misc{cuccaro2004new,
      title={A new quantum ripple-carry addition circuit}, 
      author={Steven A. Cuccaro and Thomas G. Draper and Samuel A. Kutin and David Petrie Moulton},
      year={2004},
      eprint={quant-ph/0410184},
      archivePrefix={arXiv},
      primaryClass={quant-ph},
      url={https://arxiv.org/abs/quant-ph/0410184}, 
}

@misc{zalka,
      title={Shor's algorithm with fewer (pure) qubits}, 
      author={Christof Zalka},
      year={2006},
      eprint={quant-ph/0601097},
      archivePrefix={arXiv},
      primaryClass={quant-ph},
      url={https://arxiv.org/abs/quant-ph/0601097}, 
}

@inproceedings{Shende, author = {Shende, Vivek V. and Bullock, Stephen S. and Markov, Igor L.}, title = {Synthesis of quantum logic circuits}, year = {2005}, isbn = {0780387376}, publisher = {Association for Computing Machinery}, address = {New York, NY, USA}, url = {https://doi.org/10.1145/1120725.1120847}, doi = {10.1145/1120725.1120847}, booktitle = {Proceedings of the 2005 Asia and South Pacific Design Automation Conference}, pages = {272–275}, numpages = {4}, location = {Shanghai, China}, series = {ASP-DAC '05} }

@article{sun, author = {Sun, Xiaoming and Tian, Guojing and Yang, Shuai and Yuan, Pei and Zhang, Shengyu}, title = {Asymptotically Optimal Circuit Depth for Quantum State Preparation and General Unitary Synthesis}, year = {2023}, issue_date = {Oct. 2023}, publisher = {IEEE Press}, volume = {42}, number = {10}, issn = {0278-0070}, url = {https://doi.org/10.1109/TCAD.2023.3244885}, doi = {10.1109/TCAD.2023.3244885}, journal = {Trans. Comp.-Aided Des. Integ. Cir. Sys.}, month = oct, pages = {3301–3314}, numpages = {14} }

@article{faster-integer,
author = {F\"{u}rer, Martin},
title = {Faster Integer Multiplication},
journal = {SIAM Journal on Computing},
volume = {39},
number = {3},
pages = {979-1005},
year = {2009},
doi = {10.1137/070711761},

URL = { 
    
        https://doi.org/10.1137/070711761
    
    

},
eprint = { 
    
        https://doi.org/10.1137/070711761
    
    

}
}

@misc{xu2026bonaautomaticmanagementdirty,
      title={Bona: Automatic Management of Dirty Ancilla Borrowing in Quantum Circuits}, 
      author={Xiaoquan Xu and Chenke Liu and Boning Meng and Zihao Shen and Li Zhou},
      year={2026},
      eprint={2608.08765},
      archivePrefix={arXiv},
      primaryClass={cs.PL},
      url={https://arxiv.org/abs/2608.08765}, 
}

@InProceedings{reverc,
author="Amy, Matthew
and Roetteler, Martin
and Svore, Krysta M.",
editor="Majumdar, Rupak
and Kun{\v{c}}ak, Viktor",
title="Verified Compilation of Space-Efficient Reversible Circuits",
booktitle="Computer Aided Verification",
year="2017",
publisher="Springer International Publishing",
address="Cham",
pages="3--21",
isbn="978-3-319-63390-9",
doi="10.1007/978-3-319-63390-9_1"
}

@misc{revs,
      title={Reversible circuit compilation with space constraints}, 
      author={Alex Parent and Martin Roetteler and Krysta M. Svore},
      year={2015},
      eprint={1510.00377},
      archivePrefix={arXiv},
      primaryClass={quant-ph},
      url={https://arxiv.org/abs/1510.00377}, 
}

@ARTICLE{completeReversibleCircuitRules,
  author={Feng, Shiguang and Li, Lvzhou},
  journal={IEEE Transactions on Computer-Aided Design of Integrated Circuits and Systems}, 
  title={A complete set of transformation rules for reversible circuits}, 
  year={2025},
  volume={},
  number={},
  pages={1-1},
  doi={10.1109/TCAD.2025.3641533}}

@INPROCEEDINGS{completeQuantumRules,
  author={Clément, Alexandre and Heurtel, Nicolas and Mansfield, Shane and Perdrix, Simon and Valiron, Benoît},
  booktitle={2023 38th Annual ACM/IEEE Symposium on Logic in Computer Science (LICS)}, 
  title={A Complete Equational Theory for Quantum Circuits}, 
  year={2023},
  volume={},
  number={},
  pages={1-13},
  doi={10.1109/LICS56636.2023.10175801}}

@InProceedings{ExtAndSimpcompleteQuantumRules,
  author =	{Cl\'{e}ment, Alexandre and Delorme, No\'{e} and Perdrix, Simon and Vilmart, Renaud},
  title =	{{Quantum Circuit Completeness: Extensions and Simplifications}},
  booktitle =	{32nd EACSL Annual Conference on Computer Science Logic (CSL 2024)},
  pages =	{20:1--20:23},
  series =	{Leibniz International Proceedings in Informatics (LIPIcs)},
  ISBN =	{978-3-95977-310-2},
  ISSN =	{1868-8969},
  year =	{2024},
  volume =	{288},
  editor =	{Murano, Aniello and Silva, Alexandra},
  publisher =	{Schloss Dagstuhl -- Leibniz-Zentrum f{\"u}r Informatik},
  address =	{Dagstuhl, Germany},
  URL =		{https://drops.dagstuhl.de/entities/document/10.4230/LIPIcs.CSL.2024.20},
  URN =		{urn:nbn:de:0030-drops-196639},
  doi =		{10.4230/LIPIcs.CSL.2024.20}
}

@inproceedings{minimalCompleteQuantumRules,
author = {Cl\'{e}ment, Alexandre and Delorme, No\'{e} and Perdrix, Simon},
title = {Minimal Equational Theories for Quantum Circuits},
year = {2024},
isbn = {9798400706608},
publisher = {Association for Computing Machinery},
address = {New York, NY, USA},
url = {https://doi.org/10.1145/3661814.3662088},
doi = {10.1145/3661814.3662088},
booktitle = {Proceedings of the 39th Annual ACM/IEEE Symposium on Logic in Computer Science},
articleno = {27},
numpages = {14},
location = {Tallinn, Estonia},
series = {LICS '24}
}

@inproceedings{Li2026AncillaSafety,
  author    = {Li, Jiqi and Mei, Jingyi and Fang, Wang and Guan, Ji},
  editor    = {Darulova, Eva and Lin, Anthony W. and R{\"u}mmer, Philipp},
  title     = {Formal Verification of {Quantum Ancilla Safety}},
  booktitle = {Computer Aided Verification},
  series    = {Lecture Notes in Computer Science},
  volume    = {16684},
  pages     = {326--348},
  year      = {2026},
  publisher = {Springer Nature Switzerland},
  address   = {Cham},
  doi       = {10.1007/978-3-032-32537-2_16},
  isbn      = {978-3-032-32537-2},
  url       = {https://doi.org/10.1007/978-3-032-32537-2_16}
}

@software{liu_2026_22168225,
  author       = {Liu, Chenke},
  title        = {Quantum Uncomputation of Clean and Dirty Ancilla
                   Qubits: Artifact
                  },
  version      = {v2},
  month        = oct,
  year         = 2026,
  publisher    = {Zenodo},
  doi          = {10.5281/zenodo.22168225},
  url          = {https://doi.org/10.5281/zenodo.22168225},
}

\newpage
\appendix

{\centering\LARGE\textbf{Appendices}}

\section{Store--Use Templates with a Rule-Based Static Reasoning System}\label{app:tpun-full}

A central insight of~\cite{silq,qurts} is that the existence of uncomputation can be enforced by disciplined structural templates.
We extend this insight to enable synthesis-oriented existence checking, where establishing existence also supports regular uncomputation synthesis. 

\subsection{Template Usage of Ancillas}

We begin by introducing the following pieces of syntactic sugar, which factor the body of a borrow statement into two explicit parts and expose a commonly used template for ancillas:
\[
  \borrowcd{a}{C_s}{C_u} \triangleq \borrowccdd{a}{C_s; C_u}
\]

The clean-ancilla instance corresponds to the store--use template body $C_s; C_u$ shown in \cref{fig:storeuse-template}, and the dirty-ancilla instance is obtained by prepending an additional $C_u$, so that the first three components $C_u; C_s; C_u$ form a toggle-detection template.
In both cases, uncomputation can be carried out regularly by appending $C_s^\dagger$.

\begin{figure}[tbp]
  \centering
  \includegraphics[width=0.7\linewidth]{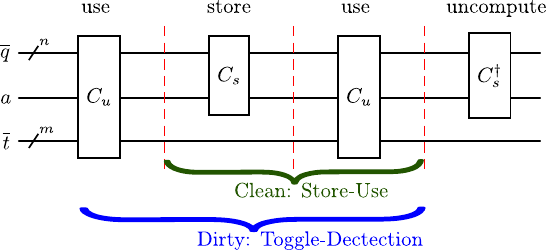}
  \label{fig:storeuse-template}
\end{figure}

The clean-ancilla template captures the pattern in which an ancilla temporarily stores an intermediate value for later use. 
This mirrors the clean-ancilla case discussed in \cref{subsec:safeusecase}:  
the first $\tToffoli[q_1, q_2, a]$ corresponds to $C_s$, which stores an intermediate result in the ancilla,  
and the second $\tToffoli[a, q_3, t]$ corresponds to $C_u$, which uses this intermediate value. The final 
$\tToffoli[q_1, q_2, a]$ serves as $C_s^\dagger$ that 
safely restores the ancilla to its original state.

The dirty-ancilla template specializes to the toggle-detection usage pattern, where the ancilla acts as a temporary flag that is conditionally toggled by intermediate computation results.
Specifically, the construction parallels the dirty-ancilla case discussed in \cref{subsec:safeusecase}, 
where the first and third $\tToffoli[a, q_3, t]$ gates use the intermediate result stored in the ancilla twice—once before and once after it is written—while the second $\tToffoli[q_1, q_2, a]$ corresponds to $C_s$, which performs the storage.  
In this pattern, the intermediate result $1$ is represented by flipping the ancilla qubit:  
if the ancilla flips, only one instance of $C_u$ applies; if it does not, the two applications of $C_u$ cancel each other.  
The final $\tToffoli[q_1, q_2, a]$ serves as $C_{s}^\dagger$ that safely restores the ancilla to its original state.

\subsection{Formal Template}\label{spform}

In this subsection, we formally define the two templates by imposing semantic constraints on $C_s$ and $C_u$ that guarantee (i) uncomputation exists, and (ii) it can be synthesized regularly by appending $C_s^\dagger$ after the corresponding body. 

\subsubsection{Clean-Ancilla Template}

We begin by presenting a core definition. Intuitively, if an ancilla is affected by certain qubits, those qubits should remain invariant afterward, 
as their states are required to uncompute the ancilla. 
The keyword \const in Silq and the lifetime mechanism in Qurts both serve this purpose. 
We extend this intuition as follows.

\begin{definition}\label{def:spform}
    Let $\olq$ and $\olt$ be disjoint $n$ and $m$ qubits ($\olq \cap \olt = \emptyset$). 
    A unitary operator $U$ on $\cH_{\olq} \otimes \cH_{\olt}$ is said to be \emph{constant on $\olq$}, written $U:\spform{\olq}{\olt}$, if there exists a family of unitaries $\{U_i\}_{i=0}^{2^n-1}$ on $\cH_{\olt}$ such that, for all $i \in \{0, \dots, 2^n-1\}$ and all $\ket{\varphi} \in \cH_{\olt}$,
    \[
    U\bigl(\ket{e_i}_{\olq}\otimes \ket{\varphi}_{\olt}\bigr)
    =
    \ket{e_i}_{\olq}\otimes U_i\ket{\varphi}_{\olt}.
    \]
    
    Equivalently, $U=\sum_{i=0}^{2^n-1}\ketbra{e_i}{e_i}_{\olq}\otimes U_i$.
    
\end{definition}

Conceptually, \const corresponds to a quantum case statement: if the state of $\olq$ is $\ket{e_i}_{\olq}$, then apply $U_i$ to $\olt$.
We write invariant qubits in square brackets and potentially varying qubits in parentheses. For example, $\tCNOT[q,a] : \spform{q}{a}$, whereas $U[\olq] : \spform{}{\olq}$ for any unitary operator $U[\olq]$. This notation is chosen deliberately: square brackets convey rigidity, while parentheses suggest flexibility. 

The term \qfree was originally introduced in~\cite{silq}. Here we use it to denote unitaries that implement reversible Boolean functions. Throughout, when an operator is applied in a larger context, we implicitly extend it to the full space by tensoring identities on qubits it does not act on.

\begin{definition}[Clean-ancilla Store–Use Template]\label{def:cleantemplate}
Let $S \equiv \borrowc{a}{C_s}{C_u}$. Let $\olq$ and $\olt$ be disjoint $n$ and $m$ qubits such that $qv(S)=\olq \cup \olt$ (Recall $a \notin qv(S)$). We say that $S$ forms a clean-ancilla Store--Use template if the following conditions hold:

\begin{enumerate}
    \item \emph{Store phase.}  
    $qv(C_s) \subseteq \olq \cup \{a\}$, and
    $\sem{C_s} : \spform{\olq}{a}$ and is $\sem{C_s}$ a \qfree unitary.
    \item \emph{Use phase.}  
    $\sem{C_u} : \spform{\olq, a}{\olt}$.
\end{enumerate}

\end{definition}

The above definition formalizes the clean-ancilla template in \cref{fig:storeuse-template}.
The \const constraints enforce that the store phase preserves $\olq$ and the use phase preserves both $\olq$ and $a$, while the \qfree condition ensures the store phase is invertible in a regular manner.
Together, these constraints witness the existence of uncomputation and enable a synthesis: appending $C_s^\dagger$.

\begin{proposition}\label{prop:cleantempuncomp}
    $\sem{C_s; C_u}$ is uncomputable and $\sem{C_s; C_u; C_s^\dagger} \in \Uncomp(\sem{C_s; C_u}, a)$.
\end{proposition}

Furthermore, a Store--Use template itself is \const on some qubits so that it can serve as a component within another instance of the template.

\begin{proposition}\label{prop:cleantempspform}
    $\sem{C_s; C_u; C_s^\dagger} : \spform{\olq, a}{\olt}$.
\end{proposition}

\subsubsection{Dirty-Ancilla Template}

While the overall templates in~\cref{fig:storeuse-template} are similar for clean and dirty ancillas, the dirty-ancilla template requires additional constraints to enable this regular synthesis.

We have explained that the key intuition behind using dirty ancillas is the principle of \emph{toggle detection}: 
two identical operations cancel each other if no toggle occurs on the dirty ancilla, 
whereas one operation remains if a toggle does occur. 
This behavior enables conditional activation of functionality without assuming the ancilla is initialized to $\ket{0}$.
To capture this property, we introduce the following two notions.

\begin{definition}[Quantum Boolean Function (\qbf)]~\cite{quantumbooleanfunctions}
    A Quantum Boolean Function (\qbf) is a unitary operator $U$ such that $U^2 = I$.
\end{definition}

\begin{definition}[Controlled-QBF (\bcqbf)]
Let $\overline p,\overline q,\overline t$ be disjoint qubits.
Let $U$ be a unitary operator acting on these qubits. We say that $U$ is a \emph{$\cqbf{\overline p}{\overline q}{\overline t}$} if the following conditions hold:

\begin{enumerate}
  \item \emph{\qbf and \const:}
  $U$ is a \qbf and $U : \spform{\olp, \olq}{\olt}$.

  \item \emph{Controlledness:}
  When $\olq \ne \emptyset$, $U$ acts non-trivially on $\olt$ only if all qubits in $\olq$ are in $\ket{\overline{1}}$.
\end{enumerate}
\end{definition}

\begin{definition}[Dirty-ancilla Store-Use Template]\label{def:dirtytemplate}
    Let $S \equiv \borrowd{a}{C_s}{C_u}$. Let $\olq$ and $\olt$ be disjoint $n$ and $m$ qubits such that $qv(S)=\olq \cup \olt$ (Recall $a \notin qv(S)$). We say that $S$ forms a dirty-ancilla Store--Use template if the following conditions hold:

    \begin{enumerate}
        \item \emph{Store phase.}  
        $qv(C_s) \subseteq \olq \cup \{a\}$, and
        $\sem{C_s} : \spform{\olq}{a}$ and $\sem{C_s}$ is a \qfree unitary.
        \item \emph{Use phase.}  
        $\sem{C_u} : \cqbf{\olq}{a}{\olt}$.
    \end{enumerate}
\end{definition}

The above definition formalizes the dirty-ancilla template in~\cref{fig:storeuse-template}.
The \const and the \qfree constraints are simliar with the clean-ancilla case. The additional \bcqbf constraint captures the toggle-detection pattern. Together, these constraints enable a toggle-based construction.

\begin{proposition}\label{prop:dirtytempuncomp}
    $\sem{C_s; C_u}$ is uncomputable and $\sem{C_u; C_s; C_u; C_s^\dagger} \in \Uncomp^*(\sem{C_s; C_u}, a)$.
\end{proposition}

\begin{proposition}\label{prop:dirtytempspform}
    $\sem{C_u; C_s; C_u; C_s^\dagger} : \spform{\olq, a}{\olt}$.
\end{proposition}

\subsection{Static Reasoning System}

The templates already capture a semantic structure for the existence of uncomputation and regular synthesis.
Therefore, checking whether a program conforms to the templates--that is, whether $C_s$ and $C_u$ satisfy the corresponding constraints--realizes  synthesis-oriented existence checking.

Prior work such as~\cite{silq,qurts} adopts a type-system discipline to automatically check conformance to their templates. However, \bcqbf for dirty template is a relatively heavy semantic property to encode in a type system, which motivates a more direct and focused approach in our setting. To automate template conformance checking,
we present a syntax-directed, rule-based static reasoning system for deriving the relevant properties.

\paragraph{Reasoning System.}

Formally, the reasoning system works on a circuit $G \in \QCa$ and operates within a context $\Gamma$, which is a set of circuit semantics judgments accumulated during reasoning:
\begin{itemize}
    \item $C : \spform{\olq}{\olt}$ indicates that $\sem{C} : \spform{\olq}{\olt}$;
    \item $C : \qfree$ indicates that $\sem{C}$ is \qfree;
    \item $C : \qbf$ indicates that $\sem{C}$ is a QBF;
    \item $C : \cqbf{\olp}{\olq}{\olt}$ indicates that $\sem{C}$ is a $\cqbf{\olp}{\olq}{\olt}$;
    \item $[C_1, C_2]$ indicates that $\sem{C_1}$ and $\sem{C_2}$ commute.
\end{itemize}

We design a set of syntactic rules that infer semantic properties of a circuit $C$ based on its syntactic structure.
Each property, such as \const, \qfree, \qbf, and \bcqbf, captures a specific semantic constraint defined in the templates introduced earlier.  
The key correctness guarantee is that every rule inferring a property for a $\borrowcd{a}{C_s}{C_u}$ statement 
explicitly requires $\sem{C_s}$ and $\sem{C_u}$ to satisfy the corresponding template constraints.  
For example:
\[
    \inferrule*[right=CONST-BORROW CLEAN]{ 
        \Gamma \vdash C_s : \spform{\olq}{a} \\ \Gamma \vdash C_u : \spform{\olq, a}{\olt} \\ \Gamma \vdash C_s : \qfree 
    }{ 
        \Gamma \vdash \borrowc{a}{C_s}{C_u} : \spform{\olq}{\olt}
    }
\]

Hence, only programs that use templates correctly can be inferred to possess these semantic properties.  
Since the overall program has \const property if and only if every $\bborrow$ statement has \const property, 
our reasoning system infers a \const property for the whole circuit to ensure its semantic validity.
For the complete set of rules, see Fig.~\ref{fig:sp-rule}--\ref{fig:comm-rule}.  
Formal proofs of their soundness are provided in Appendix~\ref{app:proof-soundreasoning}.

\begin{theorem}[Soundness]\label{thm:reasoning-soundness}
If a circuit $G \in \QCa$ is assigned a \const property by the static reasoning system, 
then for every borrow substatement $B$ in $G$, $\sem{B} \ne \bot$; moreover, $\sem{G} \ne \bot$.
\end{theorem}

\begin{figure}[]
    \centering
    \begin{mathpar}
    
    \inferrule*[right=CONST-SKIP]
    { }
    { \Gamma \vdash \bskip : \spform{}{} }
    
    \inferrule*[right=CONST-UNITARY]
    { }
    { \Gamma \vdash U[q] : \spform{}{q} }
    
    \inferrule*[right=CONST-IDENTITY]
    { U = I}
    { \Gamma \vdash U[q] : \spform{q}{} }

    \inferrule*[right=CONST-SEQ]{
      \Gamma \vdash C_1 : \spform{\overline{q_1}}{\overline{t_1}} \\
      \Gamma \vdash C_2 : \spform{\overline{q_2}}{\overline{t_2}} \\
      \olq = (\overline{q_1} \cup \overline{q_2}) \setminus (\overline{t_1} \cup \overline{t_2}) \\
      \olt = (\overline{q_1} \cup \overline{q_2} \cup \overline{t_1} \cup \overline{t_2}) \setminus \olq
    }{
      \Gamma \vdash C_1 ; C_2 : \spform{\olq}{\olt}
    }

    \inferrule*[right=CONST-QIF]{ 
        \Gamma \vdash C_1 : \spform{\overline{q_1}}{\overline{t_1}} \\ 
        \Gamma \vdash C_0 : \spform{\overline{q_2}}{\overline{t_2}} \\ 
        x \notin \overline{q_1} \cup \overline{q_2} \cup \overline{t_1} \cup \overline{t_2} \\
        \olq =
        (\overline{q_1} \cup \overline{q_2})
        \setminus
        (\overline{t_1} \cup \overline{t_2}) \\
        \olt =
        (\overline{q_1} \cup \overline{q_2}
        \cup \overline{t_1} \cup \overline{t_2})
        \setminus \olq
    }{ 
        \Gamma \vdash \qif{x}{C_1}{C_0} : \spform{x, \olq}{\olt} 
    }

    \inferrule*[right=CONST-BORROW CLEAN]{ 
        \Gamma \vdash C_s : \spform{\olq}{a} \\ \Gamma \vdash C_u : \spform{\olq,a}{\olt} \\ \Gamma \vdash C_s : \qfree 
    }{ 
        \Gamma \vdash \borrowc{a}{C_s}{C_u} : \spform{\olq}{\olt}
    }

    \inferrule*[right=CONST-BORROW DIRTY]
    { \Gamma \vdash C_s : \spform{\olq}{a} \\ \Gamma \vdash C_u : \cqbf{\olq}{a}{\olt} \\ \Gamma \vdash C_s : \qfree }
    { \Gamma \vdash \borrowd{a}{C_s}{C_u} : \spform{\olq}{\olt} }

    \inferrule*[right=CONST-CQBF]{
        \Gamma \vdash C : \cqbf{\olp}{\olq}{\olt}
    } {
        \Gamma \vdash C : \spform{\olp, \olq}{\olt}
    }

    \inferrule*[right=CONST-NEW SQUARE]
    { \Gamma \vdash C : \spform{\olq}{\olt} \\
    a \notin \olq \cup \olt}
    { \Gamma \vdash C : \spform{a, \olq}{\olt} }

    \inferrule*[right=CONST-MOVE SQUARE]
    { \Gamma \vdash C : \spform{\olq,a}{\olt} }
    { \Gamma \vdash C : \spform{\olq}{a,\olt} }

\end{mathpar}
    \caption{\const rules}
    \label{fig:sp-rule}
\end{figure}

\begin{figure}[]
    \centering
    \begin{mathpar}
        \inferrule*[right=QFree-SKIP]
        { }
        { \Gamma \vdash \bskip : \qfree }    
        
        \inferrule*[right=QFree-BASIC]
        { U = I \lor U = X}
        { \Gamma \vdash U[q] : \qfree }    
        
        \inferrule*[right=QFree-SEQ]
        { \Gamma \vdash C_1 : \qfree \\ \Gamma \vdash C_2 : \qfree }
        { \Gamma \vdash C_1; C_2 : \qfree }    
        
        \inferrule*[right=QFree-QIF]
        { \Gamma \vdash C_1 : \qfree \\ \Gamma \vdash C_0 : \qfree }
        { \Gamma \vdash \qif{x}{C_1}{C_0} : \qfree }
            
        \inferrule*[right=QFree-BORROW CLEAN]
        { \Gamma \vdash C_s : \spform{\olq}{a} \\ \Gamma \vdash C_u : \spform{\olq, a}{\olt} \\ \Gamma \vdash C_s : \qfree \\ \Gamma \vdash C_u : \qfree }
        { \Gamma \vdash \borrowc{a}{C_s}{C_u} : \qfree }
            
        \inferrule*[right=QFree-BORROW DIRTY]
        { \Gamma \vdash C_s : \spform{\olq}{a} \\ \Gamma \vdash C_u : \cqbf{\olq}{a}{\olt} \\ \Gamma \vdash C_s : \qfree \\ \Gamma \vdash C_u : \qfree }
        { \Gamma \vdash \borrowd{a}{C_s}{C_u} : \qfree }
    \end{mathpar}
    \caption{\qfree rules}
    \label{fig:qfree-rule}
\end{figure}

\begin{figure}[]
    \centering
    \begin{mathpar}
        \inferrule*[right=QBF-SKIP]
        { }
        { \Gamma \vdash \bskip : \qbf }
        
        \inferrule*[right=QBF-BASIC]
        { U = I \lor U = X \lor U = Y \lor U = Z \lor U = H }
        { \Gamma \vdash U[q] : \qbf }
        
        \inferrule*[right=QBF-CONJUGATE]
        { \Gamma \vdash C_1 : \qbf \\
        C_2 \text{ is borrow-free} }
        { \Gamma \vdash C_2 ; C_1 ; C_2^\dagger  : \qbf }
        
        \inferrule*[right=QBF-TENSOR]
        { \Gamma \vdash C_1 : \qbf \\ \Gamma \vdash C_2 : \qbf \\ qv(C_1) \cap qv(C_2) = \emptyset }
        { \Gamma \vdash C_1 ; C_2 : \qbf }
        
        \inferrule*[right=QBF-SEQ]
        { \Gamma \vdash C_1 : \qbf \\ \Gamma \vdash C_2 : \qbf \\ \Gamma \vdash [C_1, C_2] }
        { \Gamma \vdash C_1; C_2 : \qbf }
        
        \inferrule*[right=QBF-QIF]
        { \Gamma \vdash C_1 : \qbf \\ \Gamma \vdash C_0 : \qbf }
        { \Gamma \vdash \qif{x}{C_1}{C_0} : \qbf }
        
        \inferrule*[right=QBF-BORROW CLEAN]
        { \Gamma \vdash C_s : \spform{\olq}{a} \\ \Gamma \vdash C_u : \spform{\olq,a}{\olt} \\ \Gamma \vdash C_s : \qfree \\ \Gamma \vdash C_u : \qbf }
        { \Gamma \vdash \borrowc{a}{C_s}{C_u} : \qbf }
        
        \inferrule*[right=QBF-BORROW DIRTY]
        { \Gamma \vdash C_s : \spform{\olq}{a} \\ \Gamma \vdash C_u : \cqbf{\olq}{a}{\olt} \\ \Gamma \vdash C_s : \qfree \\ \Gamma \vdash C_u : \qbf }
        { \Gamma \vdash \borrowd{a}{C_s}{C_u} : \qbf }

        \inferrule*[right=QBF-CQBF]{
            \Gamma \vdash C : \cqbf{\olp}{\olq}{\olt}
        }{
            \Gamma \vdash C : \qbf
        }
    
    \end{mathpar}
    \caption{QBF rules}
    \label{fig:qbf-rule}
\end{figure}

\begin{figure}[]
    \centering
    \begin{mathpar}
        
        \inferrule*[right=CQBF-SKIP]
        { }
        { \Gamma \vdash \bskip : \cqbf{}{}{} }
        
        \inferrule*[right=CQBF-BASIC]
        { U = X \lor U = Y \lor U = Z \lor U = H }
        { \Gamma \vdash U[q] : \cqbf{}{}{q} }
        
        \inferrule*[right=CQBF-IDENTITY]
        { U = I }
        { \Gamma \vdash U[q] : \cqbf{}{q}{} }

        \inferrule*[right=CQBF-QIF0]{
            \Gamma \vdash C_0 : \cqbf{\olp}{\olq}{\olt} \\
            x \notin \olp \cup \olq \cup \olt
        }{
            \Gamma \vdash \qif{x}{\bskip}{C_0}
            : \cqbf{x,\olp}{\olq}{\olt}
        }
        
        \inferrule*[right=CQBF-QIF]{
            \Gamma \vdash C_1 : \cqbf{\olp}{\olq}{\olt} \land C_0 = \bskip \\
            x \notin \olp \cup \olq \cup \olt            
        }{ 
            \Gamma \vdash \qif{x}{C_1}{C_0} : \cqbf{\olp}{x, \overline{q}}{\olt} 
        }

        \inferrule*[right=CQBF-BORROW CLEAN]{ 
            \Gamma \vdash C_s : \spform{\olq}{a} \\ 
            \Gamma \vdash C_u : \spform{\olq, a}{\olt} \\ 
            \Gamma \vdash C_s : \qfree \\
            \Gamma \vdash C_s : \cqbf{A, B, C}{D, E, F}{a} \\
            \Gamma \vdash C_u : \cqbf{a, A, D, G}{B, E, H}{\olt} \\
            A, B, C, D, E, F, G, H \text{ are disjoint quantum registers} \\
            \overline{q} = A \cup B \cup C \cup D \cup E \cup F \cup G \cup H
        }{ 
            \Gamma \vdash \borrowc{a}{C_s}{C_u} : \cqbf{A, C, D, F, G}{B, E, H}{\olt} 
        }
        
        \inferrule*[right=CQBF-BORROW DIRTY]{ 
            \Gamma \vdash C_s : \spform{\olq}{a} \\ 
            \Gamma \vdash C_s : \qfree \\
            \Gamma \vdash C_s : \cqbf{A, B, C}{D, E, F}{a} \\
            \Gamma \vdash C_u : \cqbf{A, D, G}{B, E, H, a}{\olt} \\
            A, B, C, D, E, F, G, H \text{ are disjoint quantum registers} \\
            \overline{q} = A \cup B \cup C \cup D \cup E \cup F \cup G \cup H 
        }{ 
            \Gamma \vdash \borrowd{a}{C_s}{C_u} : \cqbf{A, C, G}{B, D, E, F, H}{\olt}
        }

        \inferrule*[right=CQBF-NEW SQUARE]{
            \Gamma \vdash C : \cqbf{\olp}{\olq}{\olt} \\
            a \notin \olp \cup \olq \cup \olt
        }{ 
            \Gamma \vdash C : \cqbf{a, \olp}{\olq}{\olt} 
        }

        \inferrule*[right=CQBF-MOVE SQUARE]{
            \Gamma \vdash C : \cqbf{\olp}{\olq}{\olt}
        }{ 
            \Gamma \vdash C : \cqbf{\olp \setminus x}{\olq}{\olt, x} 
        }

        \inferrule*[right=CQBF-MOVE CURLY]{
            \Gamma \vdash C : \cqbf{\olp}{\olq}{\olt}
        }{ 
            \Gamma \vdash C : \cqbf{\olp, x}{\olq \setminus x}{\olt} 
        }

    \end{mathpar}
    \caption{CQBF rules}
    \label{fig:cqbf-rule}
\end{figure}

\begin{figure}[]
    \centering
    \begin{mathpar}
        \inferrule*[right=COMM-SELF]
        { }
        { \Gamma \vdash [C, C] }
        
        \inferrule*[right=COMM-SKIP]
        { }
        { \Gamma \vdash [\bskip, C] }
        
        \inferrule*[right=COMM-CONST]{
            \Gamma \vdash C_1 : \spform{\overline{q_1}}{\overline{t_1}} \\ 
            \Gamma \vdash C_2 : \spform{\overline{q_2}}{\overline{t_2}} \\
            \overline{t_1} \cap \overline{t_2} = \emptyset \land (\overline{q_1} \cup \overline{q_2}) \cap (\overline{t_1} \cup \overline{t_2}) = \emptyset 
        }{ 
            \Gamma \vdash [C_1, C_2] 
        }
        
        \inferrule*[right=COMM-QBF]
        { \Gamma \vdash C_1 : \qbf \\ \Gamma \vdash C_2 : \qbf \\ \Gamma \vdash C_1 ; C_2 : \qbf }
        { \Gamma \vdash [C_1, C_2] }
    \end{mathpar}
    \caption{Commutativity rules}
    \label{fig:comm-rule}
\end{figure}

\subsection{Synthesis from Templates}

We give the complete definitions of the synthesis procedure \TpUnx for the static system.

\begin{definition}\label{def:tpsynth}
The transformation 
\[
\Tpsynth : \QCa \to \QCa
\]
takes as input a circuit $G \in \QCa$ and recursively synthesizes its subcircuits as follows:
\begin{align*}
\Tpsynth(\bskip) &\triangleq \bskip, \\
\Tpsynth(U[\olq]) &\triangleq U[\olq], \\
\Tpsynth(C_1; C_2) &\triangleq \Tpsynth(C_1); \Tpsynth(C_2), \\
\Tpsynth(\qif{q}{C_1}{C_0}) &\triangleq \qif{q}{\Tpsynth(C_1)}{\Tpsynth(C_0)}, \\
\Tpsynth(\borrowc{a}{C_s}{C_u}) &\triangleq \Tpsynth(C_s); \Tpsynth(C_u); \Tpsynth(C_s)^\dagger, \\
\Tpsynth(\borrowd{a}{C_s}{C_u}) &\triangleq \Tpsynth(C_u); \Tpsynth(C_s); \Tpsynth(C_u);\\
& \quad \quad \Tpsynth(C_s)^\dagger.
\end{align*}
\end{definition}

\begin{definition}\label{def:tpun}
The transformation 
\[
\TpUnx: \QCa \to \QCa
\]
operates on a circuit $G \in \QCa$ as follows:
it first applies the static reasoning system to derive a \const property of $G$.
If the reasoning fails, \TpUnx returns the original circuit $G$.
Otherwise, it outputs the synthesized circuit $\Tpsynth(G)$.
\end{definition}

\begin{theorem}[Synthesis]
\label{thm:tpun-soundness-full}
Let $G$ be a quantum circuit with clean ancillas $\ola$, dirty ancillas $\overline{d}$, and working qubits $\olq$.
If $G$ is assigned a \const property by the static reasoning system, then $\TpUnx(G)$ outputs a circuit containing no $\bborrow$ statement, and
\[
    \sem{\TpUnx(G)}\ket{0}_{\ola}\ket{x}_{\overline{d}}\ket{\varphi}_{\olq}
    = \ket{0}_{\ola}\ket{x}_{\overline{d}}(\sem{G}\ket{\varphi}_{\olq}),
\]
for all $\ket{x}$ and $\ket{\varphi}$.
\end{theorem}

\section{Synthesis from Normal Form}\label{app:rwun-full}

\paragraph{\qfree Circuits.}
We first consider \qfree circuits, for which the synthesis of clean and dirty uncomputation is almost identical and straightforward. Given a \qfree circuit $G$ in the normal form,

\begin{itemize}
    \item Removing all $\tMCX{i}$ gates targeting ancillas in $G$ yields a circuit $\cG$ such that $\sem{\cG} \in \Uncomp(\sem{G}, \ola)$.
    \item Further removing all $\tMCX{i}$ gates that use ancillas as controls yields a circuit $\cG'$ such that $\sem{\cG'} \in \Uncomp^*(\sem{G}, \ola)$.
\end{itemize}

\paragraph{Quantum Circuits}

Once $G$ is normalized into $G'$, we can directly synthesize a circuit $\cG$ satisfying $\sem{\cG} \in \Uncomp(\sem{G}, \ola)$ by reversing all gates that target ancillas, in reverse order:

\begin{itemize}
\item For $\tX$ and $\tMCX{i}$ gates targeting an ancilla, append the same gate;
\item For $\tZ$, $\tS$, and $\tT$ gates targeting an ancilla, append nothing.
\end{itemize}

For dirty ancillas, we further reorder the ancilla-update tail so that all phase-type gates precede the $\tX$ and $\tMCX{i}$ gates targeting ancillas.
Let
\[
\mathcal V =
\{\tZ,\tS,\tS^\dagger,\tT,\tT^\dagger\},
\quad
\mathcal P =
\mathcal V \cup
\{\tCZ,\tCS,\tCS^\dagger,\tCT,\tCT^\dagger\}.
\]
We allow controlled and inverse phase gates as intermediate operations and use the following rewriting rules for $V\in\mathcal V$.
\begin{align*}\setlength{\itemsep}{1pt}
    & \tCNOT[p, t]; V[p] \equiv V[p]; \tCNOT[p, t] \quad
    \tCNOT[p, t]; V[t] \equiv V[p]; V[t]; \mathtt{CV}^\dagger[p, t]; \mathtt{CV}^\dagger[p, t]; \tCNOT[p, t] \\
    & \tX[p]; \mathtt{CV}[p, t] \equiv \mathtt{CV}[p, t]^\dagger; V[t]; \tX[p] 
    \quad\quad\quad\quad\quad
    \tX[t]; \mathtt{CV}[p, t] \equiv \mathtt{CV}[p, t]^\dagger; V[p]; \tX[t] \\
    & \tCNOT[p, a]; \mathtt{CV}[p, t] \equiv \mathtt{CV}[p, t]; \tCNOT[p, a]
    \quad\quad\quad
    \tCNOT[t, a]; \mathtt{CV}[p, t] \equiv \mathtt{CV}[p, t]; \tCNOT[t, a] \\
    & \tCNOT[p, t]; \mathtt{CV}[p, t] \equiv \mathtt{CV}^\dagger[p, t]; V[p]; \tCNOT[p, t] \quad
    \tCNOT[t, p]; \mathtt{CV}[p, t] \equiv \mathtt{CV}^\dagger[p, t]; V[t]; \tCNOT[t, p]
\end{align*}

When such transformation succeeds, removing all gates on dirty ancillas exactly generates a circuit $\cG$ such that $\sem{\cG} \in \Uncomp^*(\sem{G}, \ola)$.

\begin{figure*}[htbp]
    \centering
    \begin{minipage}[t]{0.80\textwidth}
        \begin{algorithm}[H]
        \caption{Rewrite-based phase reordering.}
        \label{alg:rewrite-phasereorderring}
        \begin{algorithmic}[1]
        \Require A quantum circuit $G$ in its normal form, acting on an ancillas $\ola$ and qubits $\olq$.
        \State $G' \gets G.\text{copy}()$
        \While{\textbf{true}}
            \State Find the first adjacent pair in the ancilla-update tail consisting of a $\tMCX{i}$ gate targeting an ancilla followed by a phase-type gate in $\mathcal P$.
            \If{Not found}                
                \If{The tail is not a sequence of phase-type gates in $\mathcal P$ followed by ancilla-targeting
                    $\tMCX{i}$ gates}
                    \State \textbf{error}("Fail.")
                \EndIf
                \State \Return $G'$
            \Else
                \If{Rewriting rules applicable}
                    \State Rewrite in $G'$.
                \Else
                    \State \textbf{error}("Fail.")
                \EndIf
            \EndIf
        \EndWhile
        \end{algorithmic}
        \end{algorithm}
    \end{minipage}
    \caption{Rewrite-based phase reordering.}
\end{figure*}

In the above text, we gave a high-level presentation of the rewrite-based uncomputation procedure \RwUnx and its soundness result. 
Here we provide the complete formal definitions of all intermediate transformations, including \NWsynth, \Rwsynth, and \RwUnx.

\begin{definition}[Normalization-based Synthesis]
\label{def:NWsynth}
The transformation 
\[
\NWsynth : (\QCa,\, a[^*] : qubit) \to \QCa \cup \{\false\}
\]
takes as input a circuit $G \in \QCa$ that is both rewrite-regular and template-regular, together with a qubit $a \in qv(G)$. 
It proceeds as follows:
\begin{enumerate}
    \item \textbf{Normalization.} Normalize $G$ with respect to $a$ according to Algorithm~\ref{alg:rewrite-normalization}. 
    If normalization fails, output $\false$; otherwise, obtain the normalized circuit $G'$.

    \item \textbf{Phase forwarding (for $a^*$ only).} 
    If the input is marked $a^*$, reorder $G'$  with respect to $a$ according to Algorithm~\ref{alg:rewrite-phasereorderring}.
    If this step fails, output $\false$.

    \item \textbf{Clean synthesis (for $a$ only).}
    If the input is unmarked, append, in reverse order, the inverses of all ancilla-targeting classical ($\tX$ and $\tMCX{i}$) gates in the normalized tail.
    Do not append inverses of ancilla-targeting diagonal phase gates.

    \item \textbf{Dirty synthesis (for $a^*$ only).}
    If the input is marked and phase. reordering has succeeded, remove all gates involving $a$ from the phase-reordered circuit, including classical and (possibly controlled) diagonal gates.
    No inverse gates are appended in this branch.
\end{enumerate}
The final result $G'$ is returned as the output.
\end{definition}

\begin{definition}[Recursive Rewrite Synthesis]
\label{def:Rwsynth}
The transformation 
\[
\Rwsynth : \QCa \to \QCa \cup \{\false\}
\]
recursively synthesizes the subcircuits of a rewrite-regular circuit $G \in \QCa$:
\begin{align*}
\Rwsynth(\bskip) &\triangleq \bskip, \\
\Rwsynth(U[\olq]) &\triangleq U[\olq], \\
\Rwsynth(C_1; C_2) &\triangleq \Rwsynth(C_1); \Rwsynth(C_2), \\
\Rwsynth(\qif{q}{C_1}{C_0}) &\triangleq \qif{q}{\Rwsynth(C_1)}{\Rwsynth(C_0)}, \\
\Rwsynth(\borrowcc{a}{G}) &\triangleq \NWsynth(\Rwsynth(G), a), \\
\Rwsynth(\borrowdd{a}{G}) &\triangleq \NWsynth(\Rwsynth(G), a^*).
\end{align*}
\end{definition}

\begin{definition}[Rewrite-based Uncomputation]
\label{def:rwun}
The transformation 
\[
\RwUnx : \QCa \to \QCa
\]
operates on a rewrite-regular circuit $G \in \QCa$ as follows:
\begin{enumerate}
    \item \textbf{Gate validity check.} 
    Verify that all component gates of $G$ lie within the allowed set 
    \[
    \{\tX, \tZ, \tH, \tS, \tT, \tMCX{i}\}.
    \]
    If not, output the original circuit $G$.

    \item \textbf{Synthesis.}
    Compute $\Rwsynth(G)$. 
    If the result is $\false$, output $G$; otherwise, output the synthesized circuit $\Rwsynth(G)$.
\end{enumerate}
\end{definition}

\section{Examples of Circuits with Ancillas Written with \text{Template-Based Borrow}\xspace}\label{app:examples}

In this section, we refer to the sugared statement $\borrowcd{a}{C_s}{C_u}$ as \templateborrow statement.

\subsection{Treating Dirty Ancillas as Clean Ones for Multi-controlled-H Implementations}

The circuit in Figure~\ref{fig:mch-circ} demonstrates a multi-controlled-$\tH$ gate, $\mathtt{MCH}[A, B, C, D, E, T]$, implemented with two dirty ancillas, $a_0$ and $a_2$, and one clean ancilla $a_1$ without uncomputation.

\begin{figure}[!htbp]
\centering
\begin{minipage}{0.48\textwidth}
\centering
\begin{quantikz}[row sep=small, column sep=small]
    \lstick{$A$}               & \qw      & \ctrl{2} & \qw      & \qw       & \qw       & \qw      & \qw     \\
    \lstick{$B$}               & \qw      & \ctrl{1} & \qw      & \qw       & \qw       & \qw      & \qw     \\
    \lstick{$a_0$}             & \ctrl{2} & \targ{}  & \ctrl{2} & \qw       & \qw       & \qw      & \qw     \\
    \lstick{$C$}               & \ctrl{1} & \qw      & \ctrl{1} & \qw       & \qw       & \qw      & \qw     \\
    \lstick{$\ket{0}_{a_1}$}  & \targ{}  & \qw      & \targ{}  & \qw       & \ctrl{2}  & \qw      & \qw     \\
    \lstick{$D$}               & \qw      & \qw      & \qw      & \qw       & \ctrl{1}  & \qw      & \qw     \\
    \lstick{$a_2$}             & \qw      & \qw      & \qw      & \ctrl{2}  & \targ{}   & \ctrl{2} & \qw     \\
    \lstick{$E$}               & \qw      & \qw      & \qw      & \ctrl{1}  & \qw       & \ctrl{1} & \qw     \\
    \lstick{$T$}               & \qw      & \qw      & \qw      & \gate{H}  & \qw       & \gate{H} & \qw
\end{quantikz}
\end{minipage}%
\hfill
\begin{minipage}{0.50\textwidth}
\small
\[
\begin{aligned}
&\ket{A,B}\ket{a_0}\ket{C}\ket{0}\ket{D}\ket{a_2}\ket{E}\ket{T} \\[4pt]
\xmapsto{\tToffoli[a_0,C,a_1]}\;&
\ket{A,B}\ket{a_0}\ket{C}\ket{a_0\land C}\ket{D}\ket{a_2}\ket{E}\ket{T} \\[4pt]
\xmapsto{\tToffoli[A,B,a_0]\ }\;&
\ket{A,B}\ket{a_0\oplus (A\land B)}\ket{C}\ket{a_0\land C} \\ 
& \ket{D}\ket{a_2}\ket{E}\ket{T} \\[4pt]
\xmapsto{\tToffoli[a_0,C,a_1]}\;&
\ket{A,B}\ket{a_0\oplus (A\land B)}\ket{C}\ket{A\land B\land C} \\ 
&\ket{D}\ket{a_2}\ket{E}\ket{T} \\[6pt]
\xmapsto{\mathtt{CCH}[a_2,E,T]\ \ \ \ }\;&
\ket{A,B}\ket{a_0\oplus (A\land B)}\ket{C}\ket{A\land B\land C} \\ 
&\ket{D}\ket{a_2}\ket{E}\;(H^{a_2\land E}\ket{T}) \\[6pt]
\xmapsto{\tToffoli[a_1,D,a_2]}\;&
\ket{A,B}\ket{a_0\oplus (A\land B)}\ket{C}\ket{A\land B\land C} \\ 
&\ket{D}\ket{a_2\oplus (A\land B\land C)\land D}\ket{E}\;(H^{a_2\land E}\ket{T}) \\[6pt]
\xmapsto{\mathtt{CCH}[a_2,E,T]\ \ \ \ }\;&
\ket{A,B}\ket{a_0\oplus (A\land B)}\ket{C}\ket{A\land B\land C} \\ 
&\ket{D}\ket{a_2\oplus A\land B\land C\land D}\ket{E} \\
&(H^{(a_2\oplus A\land B\land C\land D)\land E}H^{a_2\land E}\ket{T}) \\[6pt]
=\;&
\ket{A,B}\ket{a_0\oplus (A\land B)}\ket{C}\ket{A\land B\land C} \\ 
&\ket{D}\ket{a_2\oplus A\land B\land C\land D}\ket{E} \\
&(H^{(A\land B\land C\land D)\land E}\ket{T})
\end{aligned}
\]
\normalsize
\vspace{2mm}

\end{minipage}
\caption{The circuit realizes an operation that applies $H$ to target $T$ iff $A\land B\land C\land D\land E$ holds, using dirty ancillas $a_0, a_2$, clean ancilla $a_1$, $\tToffoli$ gates and controlled-controlled-$\tH$ (CCH) gates.}
\label{fig:mch-circ}
\end{figure}

Figure~\ref{fig:mch-qca} illustrates the implementation of this circuit in $\QCa$, where panel (c) utilizes \templateborrow statements with automatic uncomputation. We now demonstrate how to express this construction step by step. Notably, we can treat dirty ancillas as if they were clean ancillas in this case.

The implementation proceeds in four logical phases:

\textbf{Phase 1:} Store the intermediate result $A \land B$ in dirty ancilla $a_0$:

\begin{lstlisting}[language=QCa, xleftmargin=0.2\linewidth, framexrightmargin=0.2\linewidth]
borrow a0 store {
    TOFFOLI[A, B, a0]
} use {
    // a0 now contains A && B
}
\end{lstlisting}

\textbf{Phase 2:} Compute $A \land B \land C$ and store in clean ancilla $a_1$:
\begin{lstlisting}[language=QCa, xleftmargin=0.2\linewidth, framexrightmargin=0.2\linewidth]
borrow a0 store {
    TOFFOLI[A, B, a0]
} use {
    borrow a1 := |0> store {
        TOFFOLI[a0, C, a1]
    } use {
        // a1 now contains A && B && C
    }
}
\end{lstlisting}

\textbf{Phase 3:} Compute $A \land B \land C \land D$ and store in dirty ancilla $a_2$:
\begin{lstlisting}[language=QCa, xleftmargin=0.2\linewidth, framexrightmargin=0.2\linewidth]
borrow a0 store {
    TOFFOLI[A, B, a0]
} use {
    borrow a1 := |0> store {
        TOFFOLI[a0, C, a1]
    } use {
        borrow a2 store {
            TOFFOLI[a1, D, a2]    
        } use {
            // a2 now contains A && B && C && D
        }
    }
}
\end{lstlisting}

\textbf{Phase 4:} Apply the controlled-Hadamard operation using the accumulated condition:
\begin{lstlisting}[language=QCa, xleftmargin=0.2\linewidth, framexrightmargin=0.2\linewidth]
borrow a0 store {
    TOFFOLI[A, B, a0]
} use {
    borrow a1 := |0> store {
        TOFFOLI[a0, C, a1]
    } use {
        borrow a2 store {
            TOFFOLI[a1, D, a2]    
        } use {
            CCH[a2, E, T]  // Apply H to T if A && B && C && D && E
        }
    }
}
\end{lstlisting}

This nested $\bborrow$ structure ensures that all ancilla qubits are properly uncomputed through the automatic uncomputation mechanism of \templateborrow. The final operation applies the Hadamard gate to target qubit $T$ if and only if the conjunction $A \land B \land C \land D \land E$ is satisfied, achieving the desired multi-controlled-Hadamard functionality.

\begin{figure}[ht]
\centering

\begin{minipage}{0.47\linewidth}
    \begin{subfigure}{\linewidth}
    \centering
    \begin{lstlisting}[language=QCa, xleftmargin=0.2\linewidth, framexrightmargin=0.2\linewidth]
TOFFOLI[a0, C, a1];
TOFFOLI[A, B, a0];
TOFFOLI[a0, C, a1];
CCH[a2, E, T];
TOFFOLI[a1, D, a2];
CCH[a2, E, T]
    \end{lstlisting}
    \caption{Without uncomputation}
    \end{subfigure}

    \vspace{1em} 

    \begin{subfigure}{\linewidth}
    \centering
    \begin{lstlisting}[language=QCa, xleftmargin=0.2\linewidth, framexrightmargin=0.2\linewidth]
    TOFFOLI[a0, C, a1];
    TOFFOLI[A, B, a0];
    TOFFOLI[a0, C, a1];
    CCH[a2, E, T];
    TOFFOLI[a1, D, a2];
    CCH[a2, E, T];
    TOFFOLI[a1, D, a2];
    TOFFOLI[a0, C, a1];
    TOFFOLI[A, B, a0];
    TOFFOLI[a0, C, a1]
    \end{lstlisting}
    \caption{With manual uncomputation}
    \end{subfigure}

\end{minipage}
\hfill
\begin{minipage}{0.47\linewidth}
    \begin{subfigure}{\linewidth}
    \centering
    \begin{lstlisting}[language=QCa]
    borrow a0 store {
        TOFFOLI[A, B, a0]
    } use {
        borrow a1 := |0> store {
            TOFFOLI[a0, C, a1]
        } use {
            borrow a2 store {
                TOFFOLI[a1, D, a2]
            } use {
                CCH[a2, E, T]
            }
        }
    }
    \end{lstlisting}
    \caption{With \templateborrow}
    \end{subfigure}
\end{minipage}

\caption{Multi-controlled-$\tH$ written in $\QCa$ using clean and dirty ancillas: (a) direct implementation without uncomputation, (b) manual uncomputation expansion, and (c) with \templateborrow statements.}
\label{fig:mch-qca}
\end{figure}

\subsection{Incrementer Implemented with Dirty Ancillas}

Figure~\ref{fig:inc-qca} illustrates an incrementer with dirty ancillas written using \templateborrow statements, with automatic uncomputation on the left and manual uncomputation on the right.

\begin{figure}[ht]
\centering
\begin{minipage}{0.47\linewidth}
\centering
\begin{subfigure}{0.98\linewidth}
    \centering
    \begin{lstlisting}[language=QCa, xleftmargin=0.2\linewidth, framexrightmargin=0.2\linewidth]
borrow a0 store {
    TOFFOLI[q0, q1, a0]
} use {
    borrow a1 store {
        TOFFOLI[a0, q2, a1]
    } use {
        CNOT[a1, q3]
    }
};
borrow a0 store {
    TOFFOLI[q0, q1, a0]
} use {
    CNOT[a0, q2]
};
CNOT[q0, q1];
X[q0]
    \end{lstlisting}
    \caption{Incrementer with dirty ancillas in $\QCa$ with \templateborrow.}
\end{subfigure}
\end{minipage}
\hfill
\begin{minipage}{0.47\linewidth}
\centering
\begin{subfigure}{0.98\linewidth}
    \centering
    \begin{lstlisting}[language=QCa, xleftmargin=0.2\linewidth, framexrightmargin=0.2\linewidth]
CNOT[a1, q3];
TOFFOLI[a0, q2, a1]
TOFFOLI[q0, q1, a0];
TOFFOLI[a0, q2, a1];
CNOT[a1, q3];
TOFFOLI[a0, q2, a1];
TOFFOLI[q0, q1, a0];
TOFFOLI[a0, q2, a1];
CNOT[a0, q2];
TOFFOLI[q0, q1, a0];
CNOT[a0, q2];
TOFFOLI[q0, q1, a0];
CNOT[q0, q1];
X[q0]
    \end{lstlisting}
    \caption{Incrementer with dirty ancillas in $\QCa$ with manual uncomputation}
\end{subfigure}
\end{minipage}
\caption{Comparison of incrementer implementations with automatic vs manual uncomputation}
\label{fig:inc-qca}
\end{figure}

\subsection{From Natural Circuits to Structured Templates: The Constant Adder Case Study}

For cases that do not trivially fit into the template, it may be challenging to express them naturally with \templateborrow statements.

Let's examine the circuit in Fig.~\ref{fig:hbadder-circ-full}, which uses dirty ancillas to compute the most significant bit (MSB) $x_3$ of the integer $x_0 x_1 x_2 x_3$ (with $x_0$ being the least significant bit) after adding the constant $(1010)_2$. This circuit corresponds to an instance of HighestBitConstAdder~\cite{factoring2n+2} in our Benchmark~1.  

\begin{figure}[ht]
    \centering
    \begin{quantikz}[row sep = 0.3cm]
        \lstick{$x_0$}   & \qw       & \qw      & \qw     & \qw      & \qw      & \ctrl{1} & \qw      & \qw      & \qw       & \qw        \\
        \lstick{$a_0$}   & \qw       & \qw      & \qw     & \qw      & \ctrl{2} & \targ{}  & \ctrl{2} & \qw      & \qw       & \qw        \\
        \lstick{$x_1$}   & \qw       & \qw      & \qw     & \qw      & \ctrl{1} & \qw      & \ctrl{1} & \qw      & \qw       & \qw        \\
        \lstick{$a_1$}   & \qw       & \qw      & \qw     & \ctrl{2} & \targ{}  & \qw      & \targ{}  & \ctrl{2} & \qw       & \qw        \\
        \lstick{$x_2$}   & \qw       & \ctrl{1} & \targ{} & \ctrl{1} & \qw      & \qw      & \qw      & \ctrl{1} & \targ{}   & \qw        \\
        \lstick{$a_2$}   & \ctrl{1}  & \targ{}  & \qw     & \targ{}  & \qw      & \qw      & \qw      & \targ{}  & \ctrl{1}  & \qw        \\
        \lstick{$x_3$}   & \targ{}   & \qw      & \qw     & \qw      & \qw      & \qw      & \qw      & \qw      & \targ{}   & \qw      
    \end{quantikz}
    \caption{Compute the highest bit of the input integer $x_0 x_1 x_2 x_3$, with $x_0$ as the lowest bit, after adding a constant $(1010)_2$, where $a_0, a_1, a_2$ are dirty ancillas without uncomputation.}
    \label{fig:hbadder-circ-full}
\end{figure}

At first glance, understanding this circuit is challenging: its functionality, construction principles, and correctness are far from obvious. This is because the presented circuit is already an optimized one, where almost all intermediate structure has been removed. Rather than attempting to reverse-engineer it, we directly express an equivalent circuit naturally using $\QCa$.  

Our goal is to compute $x_3$, the highest bit of $x_0 x_1 x_2 x_3 + (1010)_2$. We construct the circuit step by step.

\textbf{Step 1: Using the carry bit (assuming it is available).}  
Suppose we already have a dirty ancilla $a_1$ storing the carry bit from the addition $x_0 x_1 + (10)_2$. Since the third bit of $(1010)_2$ is $1$, $x_3$ must be flipped if either $x_2$ or the carry bit is $1$. This can be written as:

\begin{lstlisting}[language=QCa, xleftmargin=0.2\linewidth, framexrightmargin=0.2\linewidth]
CX[x2, x3];
borrow a1 store {
    
} use {
    qif a1 then {
        qif not x2 then {
            X[x3]
        }
    }
}
\end{lstlisting}

\textbf{Step 2: Computing the carry from $x_0 + 1$.}
Next, we compute the carry bit from $x_0 + 1$ and store it in a temporary dirty ancilla $a_0$:

\begin{lstlisting}[language=QCa, xleftmargin=0.2\linewidth, framexrightmargin=0.2\linewidth]
CX[x2, x3];
borrow a1 store {
    borrow a0 store {
        CX[x0, a0]
    } use {
        
    }
} use {
    qif a1 then {
        qif not x2 then {
            X[x3]
        }
    }
}
\end{lstlisting}

\textbf{Step 3: Computing the carry from $x_0 x_1 + (10)_2$.}  
We now compute the carry bit from $x_0 x_1 + (10)_2$ and store it in $a_1$.  
Since the second bit of $(10)_2$ is $0$, this carry is $1$ if and only if both $x_1 = 1$ and $a_0 = 1$. Therefore, we write:

\begin{lstlisting}[language=QCa, xleftmargin=0.2\linewidth, framexrightmargin=0.2\linewidth]
CX[x2, x3];
borrow a1 store {
    borrow a0 store {
        CX[x0, a0]
    } use {
        TOFFOLI[a0, x1, a1]
    }
} use {
    qif a1 then {
        qif not x2 then {
            X[x3]
        }
    }
}
\end{lstlisting}

This completes the construction.
In Fig.~\ref{fig:hbadder-circ-full-compile}, we present a possible compilation of the above program without uncomputation.
The resulting circuit closely resembles the target circuit in Fig.~\ref{fig:hbadder-circ-full}, although the two are not identical—both nonetheless implement the same intended functionality.
For completeness, Fig.~\ref{fig:hbadder-circ-full-compile-uncomp} shows one possible compilation with uncomputation.
As discussed earlier, such a compilation is not unique, and multiple equivalent implementations may exist.

\begin{figure}[ht]
    \centering
    \begin{minipage}{0.45\linewidth}
        \begin{subfigure}{\linewidth}
            \begin{quantikz}[row sep=0.3cm]
                \lstick{$x_0$} & \qw & \qw & \qw & \ctrl{1} & \qw & \qw & \qw \\
                \lstick{$a_0$} & \qw & \qw & \ctrl{2} & \targ{} & \ctrl{2} & \qw & \qw \\
                \lstick{$x_1$} & \qw & \qw & \ctrl{1} & \qw & \ctrl{1} & \qw & \qw \\
                \lstick{$a_1$} & \qw & \ctrl{2} & \targ{} & \qw & \targ{} & \ctrl{2} & \qw \\
                \lstick{$x_2$} & \ctrl{1} & \octrl{1} & \qw & \qw & \qw & \octrl{1} & \qw \\
                \lstick{$x_3$} & \targ{} & \targ{} & \qw & \qw & \qw & \targ{} & \qw
            \end{quantikz}
            \caption{Compiled circuit without uncomputation.}
        \end{subfigure}
    \end{minipage}
    \hfill
    \begin{minipage}{0.5\linewidth}
        \begin{subfigure}{\linewidth}
            \resizebox{\linewidth}{!}{%
            \begin{quantikz}[row sep=0.3cm]
                \lstick{$x_0$} & \qw & \qw & \qw & \qw & \ctrl{1} & \qw & \qw & \qw & \qw \\
                \lstick{$a_0$} & \qw & \qw & \qw & \ctrl{2} & \targ{} & \ctrl{2} & \qw & \qw & \qw \\
                \lstick{$x_1$} & \qw & \qw & \qw & \ctrl{1} & \qw & \ctrl{1} & \qw & \qw & \qw \\
                \lstick{$a_1$} & \qw & \qw & \ctrl{2} & \targ{} & \qw & \targ{} & \ctrl{2} & \qw & \qw \\
                \lstick{$x_2$} & \ctrl{1} & \targ{} & \ctrl{1} & \qw & \qw & \qw & \ctrl{1} & \targ{} & \qw \\
                \lstick{$x_3$} & \targ{} & \qw & \targ{} & \qw & \qw & \qw & \targ{} & \qw & \qw
            \end{quantikz}
            }%
            \caption{Further unfolded representation of the left circuit.}
        \end{subfigure}
    \end{minipage}
    \caption{A compilation result of the HighestBitConstAdder into a circuit without uncomputation (left), and its further unfolded representation (right).}
    \label{fig:hbadder-circ-full-compile}
\end{figure}

\begin{figure}[ht]
    \centering
    \begin{quantikz}[row sep = 0.3cm]
        \lstick{$x_0$}  & \qw       & \qw      & \qw       & \qw      & \ctrl{1} & \qw      & \ctrl{1} & \qw       & \qw     & \qw      & \ctrl{1} & \qw      & \ctrl{1}  & \qw        \\
        \lstick{$a_0$}  & \qw       & \qw      & \qw       & \ctrl{2} & \targ{}  & \ctrl{2} & \targ{}  & \qw       & \qw     & \ctrl{2} & \targ{}  & \ctrl{2} & \targ{}   & \qw        \\
        \lstick{$x_1$}  & \qw       & \qw      & \qw       & \ctrl{1} & \qw      & \ctrl{1} & \qw      & \qw       & \qw     & \ctrl{1} & \qw      & \ctrl{1} & \qw       & \qw        \\
        \lstick{$a_1$}  & \qw       & \qw      & \ctrl{2}  & \targ{}  & \qw      & \targ{}  & \qw      & \ctrl{2}  & \qw     & \targ{}  & \qw      & \targ{}  & \qw       & \qw        \\
        \lstick{$x_2$}  & \ctrl{1}  & \targ{}  & \ctrl{1}  & \qw      & \qw      & \qw      & \qw      & \ctrl{1}  & \targ{} & \qw      & \qw      & \qw      & \qw       & \qw        \\
        \lstick{$x_3$}  & \targ{}   & \qw      & \targ{}   & \qw      & \qw      & \qw      & \qw      & \targ{}   & \qw     & \qw      & \qw      & \qw      & \qw       & \qw      
    \end{quantikz}
    \caption{A possible compilation with uncomputation result of the HighestBitConstAdder into a circuit.}
    \label{fig:hbadder-circ-full-compile-uncomp}
\end{figure}

\section{Deferred Proofs}\label{app:proofs}

\subsection{Proof of Lemma~\ref{lem:jordancorollary}}\label{appendix:proof-lem:jordancorollary}

Lemma~\ref{lem:jordancorollary} follows as a corollary in Jordan's proof in~\cite{strongequivalence}. For completeness, we sketch the argument, following exactly the same construction in~\cite{strongequivalence}.

\begin{definition}
    A length-$l$ width-5 branching program taking $n$ bits of input is a sequence of $l$ triples, each of the form ($i, \alpha, \beta$), where $i \in \{1, \dots, n\}$, and $\alpha, \beta$ are permutations from $S_5$. The triple is interpreted as an instruction to apply permutation $\alpha$ if the $i^{th}$ input bit is zero, or apply $\beta$ if the $i^{th}$ input bit is one. The final permutation obtained by the composition of the $l$ permutations is the output of the branching program.
\end{definition}

\begin{proposition}[Barrington's theorem \cite{boundedwidthpolysizebranching}]\label{prop:Barrington}
    Given any depth $d$, fan-in 2 boolean formula $f$, and any 5-cycle $\alpha \in S_5$, one can in poly($4^d$) time construct a width-5 branching program of length at most $4^d$ such that the branching program evaluates to $\alpha$ if $f$ is $\true$ and to $I$ otherwise.
\end{proposition}

\begin{lemma}[\cite{strongequivalence}]\label{lem:jordan}
    Let $P$ be a length-$l$ width-5 branching program on $n$ input bits. Given $P$, one can construct a circuit of $O(l)$ Fredkin gates acting on $n+6$ bits that permutes five ancilla bits according to the output of $P$, provided that the sixth ancilla bit is initialized to $1$. Furthermore, the sixth ancilla bit is always left unmodified by this circuit.
\end{lemma}

We now observe that the additional sixth ancilla bit in Lemma~\ref{lem:jordan} becomes unnecessary if we allow the use of NOT gates, yielding the following corollary:

\begin{corollary}\label{corollary:A1}
    Let $P$ be a length-$l$ width-5 branching program on $n$ input bits. Given $P$, one can construct a circuit of $O(l)$ gates over $\{Fredkin, NOT\}$ gates acting on n+5 bits, such that the five designated bits are permuted according to the output of $P$.
\end{corollary}

\begin{proof}
    Our proof follows the proof of Lemma~\ref{lem:jordan} given in~\cite{strongequivalence}.

    Since SWAP gates generate $S_5$, for any triple $(i, \mathbf{1}, \beta)$ we can construct a sequence of $O(1)$ Fredkin gates that applies permutation $\beta$ to the five bits if the $i^{th}$ input bit is $1$, and does nothing otherwise. 
    
    Similarly, we can implement a ``negatively controlled'' Fredkin gate---one that swaps if the control bit is $0$ and does nothing otherwise---by surrounding the control line with two NOT gates. Therefore, for any triple $(i, \alpha, \mathbf{1})$ we obtain an $O(1)$-gate construction over $\{Fredkin, NOT\}$ that applies $\alpha$ conditioned on the $i^{\text{th}}$ bit being $0$.  
    
    By composing the constructions for $(i, \mathbf{1}, \beta)$ and $(i, \alpha, \mathbf{1})$, we can realize an arbitrary triple $(i, \alpha, \beta)$.  
    Simulating all $l$ triples in sequence reproduces the action of $P$.
\end{proof}

We now prove Lemma~\ref{lem:jordancorollary} by adapting the argument in~\cite{strongequivalence}.

\begin{proof}[Proof of \cref{lem:jordancorollary}]
    Any polynomial-size CNF formula $f$ can be expressed as a logarithmic-depth, fan-in-2 boolean circuit.  
    By Proposition~\ref{prop:Barrington}, we can construct, in polynomial time, a width-5 branching program $P$ of polynomial length that evaluates to a fixed 5-cycle if $f$ is satisfied and to the identity permutation otherwise.  
    Then, by Corollary~\ref{corollary:A1}, we can implement $P$ with a polynomial-size $\{Fredkin, NOT\}$ circuit acting on $n+5$ bits. Since each Fredkin gate can be decomposed into three $\tToffoli$ gates, the same result holds for the $\{\textsc{Toffoli}, \textsc{NOT}\}$ gate set.
\end{proof}

\subsection{Proof of \cref{prop:hardercleansafety}}

\begin{proof}[Proof of \cref{prop:hardercleansafety}]
Let $G$ be a quantum circuit acting on a clean ancilla register $\ola$ and a quantum register $\olq$. 
Construct a new circuit $G'$, where we use an open multi-controlled \textsc{SWAP} gate in the middle that fires when the controls are $\ket{\overline{0}}_{\ola}$, as follows:
\[
\begin{quantikz}
\lstick{$p$}                        & \qw             & \qw            & \swap{1}   & \swap{1}    & \qw  \\
\lstick{$anc:=\ket{0}$}             & \qw             & \qw            & \targX{}   & \targX{}    & \qw  \\
\lstick{$\ola=\ket{\overline{0}}$}  & \qwbundle[3]{h} & \gate[2][2]{G} & \octrl{-1} & \qw         & \qw  \\
\lstick{$\olq$}                     & \qwbundle[3]{m} & \qw            & \qw        & \qw         & \qw
\end{quantikz}
\]
We claim that $G$ is safe on $\ola$ if and only if $\mathcal{G'}_{anc, \ola}$ exists.

\paragraph{Case 1: $G$ Is Safe on $\ola$.}
If $G$ is safe, then for any input $\ket{\varphi}_{\olq}$,
\[
    G \bigl(\ket{\overline{0}}_{\ola} \ket{\varphi}_{\olq} \bigr)
    = \ket{\overline{0}}_{\ola} \ket{\varphi'}_{\olq}.
\]
Hence the control line $\ola$ is always $\ket{\overline{0}}_{\ola}$, the open-controlled \textsc{SWAP} acts non-trivially and cancels the final \textsc{SWAP} gate, which implies the joint state of $(anc,\ola)$ is always $\ket{0}_{anc}\ket{\overline{0}}_{\ola}$. Therefore $\mathcal{G'}_{anc, \ola}$ exists.

\paragraph{Case 2: $G$ Is \emph{Not} Safe on $\ola$.}
Let $\{\ket{e_i}_{\ola} \}$ denote the computational basis of $\ola$, with $\ket{e_0}_{\ola} = \ket{\overline{0}}_{\ola}$. Then there exists a state $\ket{\varphi}_{\olq}$ such that
\[
    G \bigl(\ket{\overline{0}}_{\ola} \ket{\varphi}_{\olq} \bigr)
    = \ket{\overline{0}}_{\ola} \ket{\varphi_0}_{\olq} + 
    \sum_{i \neq 0} \ket{e_i}_{\ola} \ket{\varphi_i}_{\olq},
    \qquad
    \sum_{i \neq 0} \ket{e_i}_{\ola} \ket{\varphi_i}_{\olq} \neq 0.
\]
After the controlled-\textsc{SWAP}, we have
\[
\ket{p}_p\ket{0}_{anc}
\Bigl(
    \ket{\overline{0}}_{\ola} \ket{\varphi_0}_{\olq} + 
    \sum_{i \neq 0} \ket{e_i}_{\ola} \ket{\varphi_i}_{\olq}
\Bigr)
\longmapsto
\ket{0}_p\ket{p}_{anc}\ket{\overline{0}}_{\ola} \ket{\varphi_0}_{\olq} +
\ket{p}_p\ket{0}_{anc}\sum_{i \neq 0} \ket{e_i}_{\ola} \ket{\varphi_i}_{\olq}.
\]
After the trailing \textsc{SWAP} gate, we obtain
\[
\ket{p}_p \ket{0}_{anc} \ket{\overline{0}}_{\ola} \ket{\varphi_0}_{\olq}
+ \ket{0}_p \ket{p}_{anc} \sum_{i \neq 0} \ket{e_i}_{\ola} \ket{\varphi_i}_{\olq}.
\]

If $(anc,\ola)$ \emph{admitted} uncomputation under $G'$, then the reduced states
\[
\ket{\phi_p}
:= \ket{p}_p\ket{\varphi_0}_{\olq} + 
\ket{0}_p \sum_{i \neq 0} \ket{\varphi_i}_{\olq},
\qquad p \in \{0,1\},
\]
would both be valid unit vectors. In particular,
\begin{align}
    &\braket{\varphi_0}{\varphi_0} 
    + \braket{\varphi_0}{\sum_{i\neq0} \varphi_i} 
    + \braket{\sum_{i\neq0}\varphi_i}{\varphi_0} 
    + \braket{\sum_{i\neq0}\varphi_i}{\sum_{i\neq0}\varphi_i} = 1, \label{eq:norm1} \\
    &\braket{\varphi_0}{\varphi_0} 
    + \braket{\sum_{i\neq0}\varphi_i}{\sum_{i\neq0}\varphi_i} = 1. \label{eq:norm2}
\end{align}
Subtracting Eq.~\eqref{eq:norm2} from Eq.~\eqref{eq:norm1} yields
\begin{align}
    \braket{\varphi_0}{\sum_{i\neq0}\varphi_i} 
    + \braket{\sum_{i\neq0}\varphi_i}{\varphi_0} = 0. \label{eq:cross}
\end{align}
Moreover, since $\ket{\phi_0}$ and $\ket{\phi_1}$ are orthogonal, we have
\begin{align}
  \bigl(\bra{0}_p\bra{\varphi_0}_{\olq} + \bra{0}_p\sum_{i\ne0}\bra{\varphi_i}_{\olq}\bigr)
  \bigl(\ket{1}_p\ket{\varphi_0}_{\olq} + \ket{0}_p\sum_{i\ne0}\ket{\varphi_i}_{\olq}\bigr) = 0,
\end{align}
which simplifies to
\begin{align}
  \bigl(\bra{\varphi_0} + \sum_{i\ne0}\bra{\varphi_i}\bigr)\sum_{i\ne0}\ket{\varphi_i} = 0. \label{eq:orth}
\end{align}
Combining Eq.~\eqref{eq:cross} and Eq.~\eqref{eq:orth} gives
\begin{align}
    \Bigl(\sum_{i\ne0}\bra{\varphi_i}\Bigr)\Bigl(\sum_{i\ne0}\ket{\varphi_i}\Bigr) = 0.
\end{align}

Hence $\sum_{i\ne0}\ket{\varphi_i}=0$.
By Eq.~\eqref{eq:norm2}, we have
$\|\varphi_0\|^2=1$.
On the other hand, since $G$ is unitary,
$\|\varphi_0\|^2+
\sum_{i\ne0}\|\varphi_i\|^2=1.$
Therefore, $\sum_{i\ne0}\|\varphi_i\|^2=0$,
which implies $\ket{\varphi_i}=0$ for all $i\ne0$,
contradicting the assumption that $G$ is not safe
on $\ola$. Hence $(anc,\ola)$ cannot be
uncomputed under $G'$.

Combining both cases, we conclude that
$G$ is safe on $\ola$ if and only if $\mathcal{G'}_{anc, \ola}$ exists.
\end{proof}

\subsection{Proof of \cref{prop:cleantempuncomp}}\label{appendix:proof-propcleantempuncomp}

\begin{proof}[Proof of \cref{prop:cleantempuncomp}]
Consider arbitrary computational basis states $|i\rangle_{\overline{q}}$ and $|j\rangle_{\overline{t}}$. We have
\begin{equation}\label{eq:map VU}
    \sem{C_u} \sem{C_s} \, |i\rangle_{\overline{q}} \, |0\rangle_{\overline{a}} \, |j\rangle_{\overline{t}}
    = |i\rangle_{\overline{q}} \, |\sem{C_s}_i(0)\rangle_{\overline{a}} \, \sem{C_u}_{i, \sem{C_s}_i(0)} \, |j\rangle_{\overline{t}},
\end{equation}

where $\sem{C_s}_i(0)$ denotes the ancilla output of $\sem{C_s}$ on input $|i\rangle_{\overline{q}}$ with ancilla initialized to $\ket{0}$, and $\sem{C_u}_{i,\sem{C_s}_i(0)}$ denotes the action of $\sem{C_u}$ on the target register $\overline{t}$ conditioned on $\ket{i}_{\olq}$ and $\ket{\sem{C_s}_i(0)\rangle_{\overline{a}}}$.

To show that $\sem{C_s; C_u}$ is uncomputable, it suffices to show that
\[
    \{ |i\rangle_{\overline{q}} \, \sem{C_u}_{i,\sem{C_s}_i(0)} |j\rangle_{\overline{t}} \}_{i,j}
\]
forms a basis of $\mathcal{H}_{\overline{q}} \otimes \mathcal{H}_{\overline{t}}$.

Since $|i\rangle_{\overline{q}}$ are basis states, it further suffices to show that for each fixed $|i\rangle_{\overline{q}}$,
\[
    \{ \sem{C_u}_{i,\sem{C_s}_i(0)} |j\rangle_{\overline{t}} \}_j
\]
forms a basis of $\mathcal{H}_{\overline{t}}$.

This is immediate because $\sem{C_u}_{i,\sem{C_s}_i(0)}$ is unitary. Hence, $\sem{C_s; C_u}$ is uncomputable.

Applying $\sem{C_s}^\dagger$ to uncompute the ancilla in \cref{eq:map VU} yields
\begin{equation}\label{eq:map UVU}
\sem{C_s}^\dagger \sem{C_u} \sem{C_s} \, |i\rangle_{\overline{q}} \, |0\rangle_{\overline{a}} \, |j\rangle_{\overline{t}}
= |i\rangle_{\overline{q}} \, |0\rangle_{\overline{a}} \, \sem{C_u}_{i,\sem{C_s}_i(0)} \, |j\rangle_{\overline{t}}.
\end{equation}

Hence, $\sem{C_s}^\dagger$ restores all ancilla bits to $\ket{0}_{\overline{a}}$ while leaving the working qubits unchanged, and therefore it satisfies Definition~\ref{def:cleanuncomp}.
\end{proof}

\subsection{Proof of \cref{prop:cleantempspform}}

\begin{proof}[Proof of \cref{prop:cleantempspform}]
    By Definition~\ref{def:spform} and \cref{eq:map UVU}.
\end{proof}

\subsection{Proof of \cref{prop:dirtytempuncomp}}

\begin{proof}[Proof of \cref{prop:dirtytempuncomp}]
Similar to Proof~\ref{appendix:proof-propcleantempuncomp},  $\sem{C_s; C_u}$ is uncomputable.

Consider the following equation:

\begin{equation}
    \sem{C_u} \sem{C_s} \, \ket{i}_{\overline{q}} \, \ket{0}_{a} \, \ket{j}_{\overline{t}}
    = \ket{i}_{\overline{q}} \, \ket{\sem{C_s}_i(0)}_{a} \, \sem{C_u}_{i,\sem{C_s}_i(0)} \, \ket{j}_{\overline{t}},
\end{equation}

Since $\sem{C_u} : \cqbf{\overline{q}}{a}{\overline{t}}$, we have:

1. $\sem{C_s}_i = I$: then $\ket{i}_{\overline{q}} \, \ket{\sem{C_s}_i(0)}_{a} \, \sem{C_u}_{i,\sem{C_s}_i(0)} \, \ket{j}_{\overline{t}} = \ket{i}_{\overline{q}} \, \ket{\sem{C_s}_i(0)}_{a} \, \ket{j}_{\overline{t}}$.  

2. $\sem{C_s}_i = X$: then $\ket{i}_{\overline{q}} \, \ket{\sem{C_s}_i(0)}_{a} \, \sem{C_u}_{i,\sem{C_s}_i(0)} \, \ket{j}_{\overline{t}} = \ket{i}_{\overline{q}} \, \ket{\sem{C_s}_i(0)}_{a} \, \sem{C_u}_{i,1} \ket{j}_{\overline{t}}$. 

\begin{equation}\label{eq:map UVUV}
    \sem{C_s}^\dagger \sem{C_u} \sem{C_s} \sem{C_u} \, \ket{i}_{\overline{q}} \, \ket{a}_{a} \, \ket{j}_{\overline{t}}
    = \ket{i}_{\overline{q}} \, \ket{a}_{a} \, \sem{C_u}_{i,\sem{C_s}_i(a)} \sem{C_u}_{i,a} \, \ket{j}_{\overline{t}},
\end{equation}

Since $\sem{C_u} : \cqbf{\overline{q}}{a}{\overline{t}}$, we have:

1. $\sem{C_s}_i = I$: then $\ket{i}_{\overline{q}} \, \ket{a} \sem{C_u}_{i,\sem{C_s}_i(a)} \sem{C_u}_{i,a} \ket{j}_{\overline{t}} = \ket{i}_{\overline{q}} \, \ket{a}_{a} \, \sem{C_u}_{i,a}^2 \ket{j}_{\overline{t}} = \ket{i}_{\overline{q}} \, \ket{a}_{a} \, \ket{j}_{\overline{t}}$.  

2. $\sem{C_s}_i = X$: then $\ket{i}_{\overline{q}} \, \ket{a} \sem{C_u}_{i,\sem{C_s}_i(a)} \sem{C_u}_{i,a} \ket{j}_{\overline{t}} = \ket{i}_{\overline{q}} \, \ket{a}_{a} \, \sem{C_u}_{i, 1} \ket{j}_{\overline{t}}$.

By Definition~\ref{def:dirtyuncomp}, we have
\[
    \sem{C_u; C_s; C_u; C_s^\dagger} \in \Uncomp^*(\sem{C_s; C_u}, \overline{a}).
\]

\end{proof}

\subsection{Proof of \cref{prop:dirtytempspform}}

\begin{proof}[Proof of \cref{prop:dirtytempspform}]
    By Definition~\ref{def:spform} and \cref{eq:map UVUV}.
\end{proof}

\subsection{Proof of \cref{thm:reasoning-soundness}}\label{app:proof-soundreasoning}

\subsubsection{Proof of Soundness of \const Rules}

\begin{proposition}[Soundness of \const rules]
    If $\Gamma \vdash C : \spform{\olq}{\ola}$, then $\sem{C}: \spform{\olq}{\ola}$.
\end{proposition}

\begin{proof}
The proof proceeds by structural induction.

\paragraph{Base Cases.}
\textbf{CONST-SKIP}, \textbf{CONST-UNITARY}, \textbf{CONST-IDENTITY} are directly from definition.

\paragraph{\textbf{CONST-SEQ}}

Assume the induction hypothesis for $C_1, C_2$:
\[
\Gamma \vdash C_1 : \spform{\overline{q_1}}{\overline{t_1}}, \qquad
\Gamma \vdash C_2 : \spform{\overline{q_2}}{\overline{t_2}}.
\]

Let $\olq = \overline{q_1} \cup \overline{q_2} \setminus
(\overline{t_1} \cup \overline{t_2})$ and
$\olt = (\overline{q_1} \cup \overline{q_2} \cup
\overline{t_1} \cup \overline{t_2}) \setminus \olq$. Since $\sem{C_1;C_2} = \sem{C_2}  \sem{C_1}$ and the qubits constant in both $C_1$ and $C_2$
remain constant after composition, the conclusion $\Gamma \vdash C_1;C_2 : \spform{\olq}{\olt}$ holds.

\paragraph{\textbf{CONST-QIF}}

Assume IH for $C_1,C_0$ and let $x \notin
\overline{q_1} \cup \overline{q_2} \cup
\overline{t_1} \cup \overline{t_2}$.
The guardian qubit $a$ is not modified by either branch. Only qubits constant for both branches are guaranteed to be constant. Thus $\Gamma \vdash \qif{x}{C_1}{C_0} : \spform{x,\olq}{\olt}$ holds.

\paragraph{\textbf{CONST-BORROW CLEAN}}

Given $\Gamma \vdash C_s : \spform{\olq}{a}$,
$\Gamma \vdash C_u : \spform{\olq,a}{\olt}$,
and $\Gamma \vdash C_s : \qfree$.
By Proposition~\ref{prop:cleantempuncomp} and Proposition~\ref{prop:cleantempspform}, $\sem{C_s}^\dagger  \sem{C_u}  \sem{C_s} = \ket{0}_a\bra{0} \otimes W_{qv(\bborrow) \setminus a} + \ket{1}_a\bra{1} \otimes V_{qv(\bborrow) \setminus a}$ has \const property $\spform{\olq, a}{\olt}$.
Since $\sem{\borrowc{a}{C_s}{C_u}} = W$, then $W : \spform{\olq}{\olt}$.

\paragraph{\textbf{CONST-BORROW DIRTY}}

The same with above.

\paragraph{Other Rules}
\textbf{CONST-CQBF}, \textbf{CONST-NEW SQUARE}, \textbf{CONST-MOVE SQUARE} are directly from the definition.

\end{proof}

\subsubsection{Proof of Soundness of \qfree Rules}

\begin{proposition}[Soundness of \qfree rules]
    If $\Gamma \vdash C : \qfree$, then $\sem{C}$ is \qfree.
\end{proposition}

\begin{proof}
We prove soundness by structural induction.

\paragraph{Base Cases.}
\textbf{QFree-SKIP}, \textbf{QFree-BASIC} and \textbf{QFree-SEQ} are directly from definition.

\paragraph{\textbf{QFree-QIF}}
Assume $\Gamma\vdash C_1:\qfree$ and $\Gamma\vdash C_0:\qfree$.
Both branches are \qfree and the guardian qubit is constant in $\bqif$ together guarantee that the whole $\bqif$ statement is exactly \qfree.

\paragraph{\textbf{QFree-BORROW CLEAN}}
Assume
\[
\Gamma \vdash C_s : \spform{\olq}{a},\quad
\Gamma \vdash C_u : \spform{\olq,a}{\olt},\quad
\Gamma \vdash C_s : \qfree,\quad
\Gamma \vdash C_u : \qfree.
\]
The first three assumptions ensure that $\sem{C_s}$ and $\sem{C_u}$ conform to the clean-ancilla template. Combining the \qfree property of $\sem{C_u}$ with \cref{eq:map UVU} and the denotational semantics of $\bborrow$, we conclude that $\sem{\borrowc{a}{C_s}{C_u}}$ is \qfree.

\paragraph{\textbf{QFree-BORROW DIRTY}}
The same with \textbf{QFree-BORROW CLEAN}.

\end{proof}

\subsubsection{Proof of Soundness of \qbf Rules}

\begin{proposition}[Soundness of \qbf rules]
    If $\Gamma \vdash C : \qbf$, then $\sem{C}$ is a quantum boolean function (\qbf).
\end{proposition}

\begin{proof}
The proof is by structural induction.

\paragraph{Base Cases.}
\textbf{QBF-SKIP} and \textbf{QBF-BASIC} are directly from definition.

\paragraph{Inductive Steps for Composition / Structural Rules}

\textbf{QBF-CONJUGATE.}
Let $U=\sem{C_1}$ and $V=\sem{C_2}$.
By the premises, $U^2=I$ and $V$ is unitary.
Since $C_2$ is borrow-free, its syntactic adjoint
satisfies $\sem{C_2^\dagger}=V^\dagger$.
Therefore,
\[
(V^\dagger U V)^2=V^\dagger U^2 V=I.
\]

\textbf{QBF-TENSOR.}  
Assume $\Gamma\vdash C_1:\qbf$ and $\Gamma\vdash C_2:\qbf$ and $qv(C_1)\cap qv(C_2)=\emptyset$.
Because the two circuits act on disjoint sets of qubits, their parallel composition is given by the tensor product $\sem{C_1; C_2}=\sem{C_1}\otimes\sem{C_2}$. Then it is obvious that $(\sem{C_1} \otimes \sem{C_2})^2 = I$.

\textbf{QBF-SEQ.}
Assume $\Gamma\vdash C_1:\qbf$, $\Gamma\vdash C_2:\qbf$ and $\Gamma\vdash [C_1,C_2]$ ($\sem{C_1}$ and $\sem{C_2}$ commute). Then this rule holds by the fact that the composition of two QBFs remains to be a \qbf if and only if they commute.

\textbf{QBF-QIF.}  
Assume $\Gamma \vdash C_1 : \qbf$ and $\Gamma \vdash C_0 : \qbf$.
Executing $\bqif$ twice applies each branch twice, conditioned on the control qubit $x$.
Since both branches are QBF, their double execution is semantically the identity, hence
$\bqif$ itself is a \qbf.

\textbf{QBF-BORROW CLEAN}

Assume
\[
\Gamma \vdash C_s : \spform{\olq}{a},\quad
\Gamma \vdash C_u : \spform{\olq,a}{\olt},\quad
\Gamma \vdash C_s : \qfree,\quad
\Gamma \vdash C_u : \qbf.
\]
The first three assumptions guarantee that $\sem{C_s}$ and $\sem{C_u}$ adhere to the clean-ancilla template.
By combining the QBF property of $\sem{C_u}$ with \cref{eq:map UVU} and the denotational semantics of $\bborrow$,
we see that executing this statement twice yields the identity.
Hence, $\sem{\borrowc{a}{C_s}{C_u}}$ is a \qbf.

\textbf{QBF-BORROW DIRTY}

The same with \textbf{QBF-BORROW CLEAN}.

\textbf{QBF-CQBF}.

Follows the definition of CQBF.

\end{proof}

\subsubsection{Proof of Soundness of \bcqbf Rules}

\begin{proposition}[Soundness of \bcqbf rules]
    If $\Gamma \vdash C : \cqbf{\olp}{\olq}{\olt}$, then $\sem{C}$ is a $\cqbf{\olp}{\olq}{\olt}$.
\end{proposition}

\begin{proof}
Proof is by structural induction.

\paragraph{Base Cases.}

\textbf{CQBF-SKIP}, \textbf{CQBF-BASIC} and \textbf{CQBF-IDENTITY} are directly from the definition.

\textbf{CQBF-QIF}

Assume $\Gamma \vdash C_1 : \cqbf{\olp}{\olq}{\olt}$ and $C_0 = \bskip$. By rule \textbf{QBF-QIF} we obtain $\Gamma \vdash \qif{x}{C_1}{C_0} : \qbf$. By definition of CQBF, $\sem{C_1}$ acts non-trivially on $\olt$ only when $\olq = \ket{\overline{1}}_{\olq}$. Then, $\sem{\bqif}$ acts non-trivially on $\olt$ only when $x = \ket{1}_x$ and $\olq = \ket{\overline{1}}_{\olq}$. Thus $\sem{\bqif}$ is a $\cqbf{\olp}{x, \olq}{\olt}$.

\paragraph{\textbf{CQBF-QIF0}}
Let $U=\sem{C_0}$.
The conclusion denotes
\[
|0\rangle\langle0|_x\otimes U
+
|1\rangle\langle1|_x\otimes I.
\]
Its square is the identity because $U^2=I$,
and it preserves the computational-basis values
of $x,\olp,\olq$.
Whenever the control register $\olq$ is not all ones,
$U$ is the identity, so the whole conditional
is the identity as well.
Thus the conclusion has type
$\cqbf{x,\olp}{\olq}{\olt}$.

\paragraph{\textbf{CQBF-BORROW CLEAN}}

Assume
\begin{align*}
& \Gamma \vdash C_s : \spform{\olq}{a}, &
& \Gamma \vdash C_u : \spform{\olq,a}{\olt}, &
& \Gamma \vdash C_s : \qfree, \\
& \Gamma \vdash C_s : \cqbf{A, B, C}{D, E, F}{a}, &
& \Gamma \vdash C_u : \cqbf{a, A, D, G}{B, E, H}{\olt} \\
& A, B, C, D, E, F, G, H \text{ are disjoint} &
& \olq = A \cup B \cup C \cup D \cup E \cup F \cup G \cup H.
\end{align*}

We have

\begin{align*}
& \sem{C_s}^\dagger  \sem{C_u}  \sem{C_s} 
    \ket{\olq} \ket{0} \ket{\olt} \\
&= \sem{C_s}^\dagger  \sem{C_u} 
    \ket{\olq} 
    \bigl(\sem{C_s}_{\ket{\olq}} \ket{0}\bigr) \ket{\olt} \\
&= \sem{C_s}^\dagger 
    \ket{\olq} 
    \bigl(\sem{C_s}_{\ket{\olq}} \ket{0}\bigr)
    \sem{C_u}_{\ket{\olq} \ket{\sem{C_s}_{\ket{\olq}}(0) }} \ket{\olt} \\
&= \ket{\olq} \ket{0}
    \sem{C_u}_{\ket{\olq} \ket{\sem{C_s}_{\ket{\olq}}(0) }} \ket{\olt} \\
&= \ket{\olq} \ket{0}
    \sem{C_u}_{\ket{AG} \ket{D} \ket{E} \ket{BH} \ket{\sem{C_s}_{\ket{AC} \ket{DF} \ket{E} \ket{B}}(0)}} \ket{\olt}.
\end{align*}

By assumption that $\sem{C_u}$ is a \bcqbf, then $\sem{C_u}_{\ket{\olq} \ket{\sem{C_s}_{\ket{\olq}}(0) }}$ is a \qbf and acts non-trivially on $\olt$ only when $EBH = \ket{\overline{1}}$.
Thus $\sem{\bborrow}$ is a $\cqbf{A, C, D, F, G}{E, B, H}{\olt}$ by definition.

\paragraph{\textbf{CQBF-BORROW DIRTY}}

Assume
\begin{align*}
& \Gamma \vdash C_s : \spform{\olq}{a}, &
& \Gamma \vdash C_s : \qfree, \\
& \Gamma \vdash C_s : \cqbf{A, B, C}{D, E, F}{a}, &
& \Gamma \vdash C_u : \cqbf{A, D, G}{B, E, H, a}{\olt} \\
& A, B, C, D, E, F, G, H \text{ are disjoint} &
& \olq = A \cup B \cup C \cup D \cup E \cup F \cup G \cup H.
\end{align*}

We have

\begin{align*}
& \sem{C_s}^\dagger  \sem{C_u}  \sem{C_s}  \sem{C_u} 
    \ket{\olq} \ket{a} \ket{\olt} \\
&= \sem{C_s}^\dagger  \sem{C_u}  \sem{C_s} 
    \ket{\olq} \ket{a}
    \bigl(\sem{C_u}_{\ket{\olq} \ket{a}} \ket{\olt}\bigr) \\
&= \sem{C_s}^\dagger  \sem{C_u} 
    \ket{\olq} 
    \bigl(\sem{C_s}_{\ket{\olq}} \ket{a}\bigr) 
    \bigl(\sem{C_u}_{\ket{\olq} \ket{a}} \ket{\olt}\bigr) \\
&= \sem{C_s}^\dagger
    \ket{\olq} 
    \bigl(\sem{C_s}_{\ket{\olq}} \ket{a}\bigr) 
    \bigl(\sem{C_u}_{\ket{\olq} \ket{\sem{C_s}_{\ket{\olq}}(a) }} \sem{C_u}_{\ket{\olq} \ket{a}} \ket{\olt}\bigr) \\
&= \ket{\olq} \ket{a}
    \bigl(\sem{C_u}_{\ket{\olq} \ket{\sem{C_s}_{\ket{\olq}}(a) }} \sem{C_u}_{\ket{\olq} \ket{a}} \ket{\olt}\bigr) \\
&= \ket{\olq} \ket{a}
    \bigl(\sem{C_u}_{\ket{AG} \ket{D} \ket{E} \ket{BH} \ket{\sem{C_s}_{\ket{AC} \ket{DF} \ket{E} \ket{B} }(a) }} \sem{C_u}_{\ket{AG} \ket{D} \ket{E} \ket{BH} \ket{a}} \ket{\olt}\bigr).
\end{align*}

By assumption that $\sem{C_u}$ is a \bcqbf, then $\sem{\bborrow}$ acts non-trivially on $\olt$ only when $EBH = \ket{\overline{1}}$. 

By assumption that $\sem{C_s}$ is a \bcqbf, When $DFE \ne \ket{\overline{1}}$, since $\sem{C_u}_{\ket{\olq} \ket{a}}$ is a \qbf, we have

\begin{align*}
& \sem{C_u}_{\ket{\olq} \ket{\sem{C_s}_{\ket{\olq}}(a) }} \sem{C_u}_{\ket{\olq} \ket{a}} \\
&=  \sem{C_u}_{\ket{\olq} \ket{a}} \sem{C_u}_{\ket{\olq} \ket{a}} \\
&= I.
\end{align*}

Thus $\sem{\bborrow}$ acts non-trivially on $\olt$ only when $DFE = \ket{\overline{1}}$.

When $\sem{C_s}_{\ket{\olq}} = I$, we have $\sem{\bborrow}_{\ket{\olq} \ket{a}} = I$ so that $\sem{\bborrow}_{\ket{\olq} \ket{a}}^2 = I$.

When $\sem{C_s}_{\ket{\olq}} = X$, by assumption that $\sem{C_u} : \cqbf{\olq}{a}{\olt}$, we have 

\begin{align*}
& \sem{C_u}_{\ket{\olq} \ket{\sem{C_s}_{\ket{\olq}}(a) }} \sem{C_u}_{\ket{\olq} \ket{a}} \\
&=  \sem{C_u}_{\ket{\olq} \ket{a \oplus 1}} \sem{C_u}_{\ket{\olq} \ket{a}} \\
&=  \sem{C_u}_{\ket{\olq} \ket{a}} \sem{C_u}_{\ket{\olq} \ket{a \oplus 1}} \\
\end{align*}

Then $\sem{\bborrow}_{\ket{\olq} \ket{a}} = \sem{C_u}_{\ket{\olq} \ket{a \oplus 1}} \sem{C_u}_{\ket{\olq} \ket{a}}$ and $\sem{\bborrow}_{\ket{\olq} \ket{a}}^2 = I$.

Therefore, $\sem{\bborrow}$ is a \qbf and acts non-trivially on $\olt$ only when $EBH = \ket{\overline{1}}$ and $DFE = \ket{\overline{1}}$. And $\sem{\bborrow}$ is a $\cqbf{A, C, G}{B, D, E, F, H}{\olt}$ by definition.

\paragraph{Other Rules.} 
\textbf{CQBF-NEW SQUARE}, \textbf{CQBF-MOVE SQUARE} and \textbf{CQBF-MOVE CURLY} are directly from the definition.

\end{proof}

\begin{proof}[Proof of \cref{thm:reasoning-soundness}]
For every rule in \cref{fig:sp-rule,fig:qfree-rule,fig:qbf-rule,fig:cqbf-rule,fig:comm-rule} whose conclusion is a borrow statement of the form $\borrowcd{a}{C_s}{C_u}$, the premises explicitly require the semantic constraints in \cref{def:cleantemplate} (for the clean-ancilla case) or \cref{def:dirtytemplate} (for the dirty-ancilla case), respectively. 
By \cref{prop:cleantempuncomp,prop:dirtytempuncomp}, any such derivation entails that the corresponding borrow statement is semantically valid (i.e., its denotation is not $\bot$).

Moreover, the sequential \const rules propagate \const from subcircuits to larger circuits: any derivation of $\Gamma \vdash G : \spform{\olq}{\olt}$ necessarily contains derivations of \const judgments for all relevant subcircuits of $G$. 
Hence, if the whole circuit $G$ is derivable as \const, then every borrow substatement appearing in $G$ is semantically valid. 
It follows from the denotational semantics in \cref{def:semantics} that $\sem{G}\neq \bot$.

Finally, it remains to justify the above argument by establishing soundness of each rule; this is proved above.
\end{proof}

\subsection{Proof of \cref{thm:tpun-soundness-full}}

\begin{proof}[Proof of \cref{thm:tpun-soundness-full}]
\TpUnx outputs the synthesized circuit $\Tpsynth(G)$ when the circuit $G$ is checked by our system, under which the circuit $G$ and all its subcircuits are valid (\cref{thm:reasoning-soundness}). Then the remaining goal is to prove $\sem{\Tpsynth(G)} \ket{x}_{\old} \ket{0}_{\ola} \ket{\varphi}_{\olq} = \ket{x}_{\old} \ket{0}_{\ola} \sem{G}\ket{\varphi}_{\olq}$.
We prove the theorem by structural induction on the program $G$, following the recursive definition of \Tpsynth.

\textbf{Base cases.}
For the trivial programs $\bskip$ and a single gate $U[\olq]$ the statement is immediate.

\textbf{Inductive step for composition.}
If $G = C_1; C_2$, then by the non-overlapping structure of \templateborrow statements, the sets of ancillas used in $C_1$ and $C_2$ are disjoint.  
Let $\overline{d}$ and $\ola$ denote the $m_1$-qubit dirty ancilla register and $n_1$-qubit clean ancilla register used in $C_1$, and let $\overline{e}$ and $\overline{b}$ denote the $m_2$-qubit dirty and $n_2$-qubit clean ancilla registers used in $C_2$.  
Let $\olq$ be the set of working qubits that are shared by both subcircuits. 
By inductive hypothesis, let $\Tpsynth(C_1; C_2) = \cC_1; \cC_2$, where $\cC_1$ denotes $\Tpsynth(C_1)$ and $\cC_2$ denotes $\Tpsynth(C_2)$. Then we have: 
\begin{align*}
\sem{\cC_1; \cC_2} \ket{x}_{\overline{d}} \ket{0}_{\ola} \ket{y}_{\overline{e}} \ket{0}_{\overline{b}} \ket{\varphi}_{\olq} 
&= \sem{\cC_2} \sem{\cC_1} \ket{x}_{\overline{d}} \ket{0}_{\overline{a}} \ket{y}_{\overline{e}} \ket{0}_{\overline{b}} \ket{\varphi}_{\olq}    \\
&=  \sem{\cC_2} \ket{x}_{\overline{d}} \ket{0}_{\overline{a}} \ket{y}_{\overline{e}} \ket{0}_{\overline{b}} \sem{C_1} \ket{\varphi}_{\olq} \\
&=  \ket{x}_{\overline{d}} \ket{0}_{\overline{a}} \ket{y}_{\overline{e}} \ket{0}_{\overline{b}} \sem{C_2} \sem{C_1} \ket{\varphi}_{\olq} \\
&=  \ket{x}_{\overline{d}} \ket{0}_{\overline{a}} \ket{y}_{\overline{e}} \ket{0}_{\overline{b}} \sem{C_1; C_2} \ket{\varphi}_{\olq}
\end{align*}

\textbf{Inductive step for conditional.}
Simliar with above.

\textbf{Case of a template borrow clean.}
By induction hypothesis, let $\Tpsynth(\borrowc{b}{C}{G}) = \cC; \cG; \cC^\dagger$, where $\cC$ denotes $\Tpsynth(C)$ and $\cG$ denotes $\Tpsynth(G)$. Then we have:
\begin{align*}
& \sem{\cC; \cG; \cC^\dagger} \ket{x}_{\overline{d}} \ket{0}_{\ola} \ket{0}_b \ket{\varphi}_{\olq} \\
&= \sem{\cC}^\dagger \sem{\cG} \ket{x}_{\overline{d}} \ket{0}_{\overline{a}} \sem{C} \ket{0}_b \ket{\varphi}_{\olq}\tag*{(IH)}    \\
&= \sem{\cC}^\dagger \ket{x}_{\overline{d}} \ket{0}_{\overline{a}} \sem{G} \sem{C} \ket{0}_b \ket{\varphi}_{\olq}\tag*{(IH)}    \\
&= \ket{x}_{\overline{d}} \ket{0}_{\overline{a}} \sem{C}^\dagger \sem{G} \sem{C} \ket{0}_b \ket{\varphi}_{\olq}\tag*{(IH)}    \\
&= \ket{x}_{\overline{d}} \ket{0}_{\overline{a}} \sem{C; G; C^\dagger} \ket{0}_b \ket{\varphi}_{\olq}\tag*{(IH)}    \\
&= \ket{x}_{\overline{d}} \ket{0}_{\overline{a}} \ket{0}_b \sem{\borrowc{b}{C}{G}} \ket{\varphi}_{\olq}\tag*{(\cref{thm:reasoning-soundness,prop:cleantempuncomp})}
\end{align*}

\textbf{Case of a template borrow dirty.}
Similar with above.
\end{proof}

\subsection{Proof of \cref{prop:complete-qfreenormal}}\label{appendix:proof-prop-completeqfreenormal}

We prove the following key Proposition first and then prove \cref{prop:complete-qfreenormal} based on it.

\begin{proposition}\label{prop:complete-safe-qfreenormal}
A circuit $G$ uses clean ancillas $\ola$ safely if and only if there exists a \qfree circuit $G'$ over the same qubits such that:
\begin{enumerate}
  \item $\sem{G'} = \sem{G}$; and
  \item for every $\tMCX{i}$ gate in $G'$ whose target is an ancilla $b \in \ola$, $i \ge 1$ and the gate uses at least one other ancilla $c \in \ola$ as a control qubit.
\end{enumerate}
\end{proposition}

\paragraph{Background.}
To prove the above proposition, we recall several concepts and results from~\cite{completeReversibleCircuitRules}, mainly from Section~4. For completeness, we briefly restate the necessary definitions.

\paragraph{Hypercube and Hamiltonian Paths.}
The $n$-hypercube graph is an undirected graph whose vertices are $\{0,1\}^n$, with an edge between two vertices if and only if their bit strings differ in exactly one position. All hypercube graphs are Hamiltonian. Let
\[
  \mathbb{H} = (a_0, a_1, \dots, a_{2^n-1})
\]
be a Hamiltonian path of an $n$-hypercube, where $a_i$ and $a_{i+1}$ differ in exactly one bit. Then any $n$-ary reversible function induces a permutation
\[
  \binom{a_0, a_1, \dots, a_{2^n-1}}{a'_0, a'_1, \dots, a'_{2^n-1}},
\]
which maps $a_i$ to $a'_i$. We also denote this permutation by
\[
  (a'_0, a'_1, \dots, a'_{2^n-1})_{\mathbb{H}}.
\]

\paragraph{MPMCX Gates.}
Mixed polarity multiple-controlled Toffoli (MPMCX) gates extend MCX gates by allowing negative control bits. An MPMCX gate is specified by a set $P$ of positive controls, a set $N$ of negative controls, and a target bit. For example,
\[
\begin{quantikz}[row sep = 0.2cm]
  \lstick{$q_0$} & \octrl{2} &   \\
  \lstick{$q_1$} & \ctrl{1}  &   \\
  \lstick{$t$}   & \targ{}   &  
\end{quantikz}
\]
has one negative and one positive control. An MCX gate is precisely an MPMCX gate with $N = \emptyset$.

Given a Hamiltonian path $\mathbb{H}$, define
\[
  \Delta_{\mathbb{H}} = \{M_0, M_1, \dots, M_{2^n-2}\}
\]
to be the set of $n$-bit MPMCX gates such that, for each $0 \le k \le 2^n-2$, the gate $M_k$ exchanges $a_k$ and $a_{k+1}$. Concretely, $M_k$ targets the unique bit position where $a_k$ and $a_{k+1}$ differ, uses positive controls for positions where both bits are $1$, and negative controls for positions where both bits are $0$.

Note that $\tilde{C} = M_i M_{i+1} \dots M_{i+k}$ defines the permutation of cycle shift by 1 position
\[
  (a_0, \dots, a_{i-1}, \red{a_{i+k+1}}, \blue{a_i, a_{i+1}, \dots, a_{i+k}}, a_{i+k+2}, \dots, a_{2^n-1})_{\mathbb{H}}.
\]

For any $C$ that defines the permutation $(b_1, \dots, b_{2^n-1})_{\mathbb{H}}$, $\tilde{C} C$ defines the permutation obtained by cycle shifting 1 position from the permutation defined by $C$

\[
  (b_0, \dots, b_{i-1}, \red{b_{i+k+1}}, \blue{b_i, b_{i+1}, \dots, b_{i+k}}, b_{i+k+2}, \dots, b_{2^n-1})_{\mathbb{H}}.
\]

\begin{definition}[Canonical form~\cite{completeReversibleCircuitRules}]
An $n$-bit reversible circuit is in canonical form based on $\mathbb{H}$ if it has the form
\[
  \mathbf{C}_m \mathbf{C}_{m-1} \dots \mathbf{C}_1 \mathbf{C}_0
\]
such that:
\begin{enumerate}
  \item each subcircuit
  \[
    \mathbf{C}_i = M_x M_{x+1} \dots M_{x+k}
  \]
  is a sequence of gates from $\Delta_{\mathbb{H}}$ with $0 \le x \le x+k \le 2^n-2$; and
  \item for $\mathbf{C}_i = M_x \dots M_{x+k}$ and $\mathbf{C}_j = M_y \dots M_{y+\ell}$, if $i < j$ then $x < y$.
\end{enumerate}
\end{definition}

Intuitively, the canonical form above states that any permutation $(b_0,\dots,b_{2^n-1})_{\mathbb{H}}$ can be realized by first applying $C_0$ to cyclically shift $b_0$ into the first position, then applying $C_1$ to cyclically shift $b_1$ into the second position, and so on.
Each cyclic shift satisfies the first condition for constructing a cyclic shift.
Moreover, the sequence of cyclic shifts is uniquely determined by the permutation, as specified by the second condition above, which implies the uniqueness of the canonical form.

\begin{proposition}[\cite{completeReversibleCircuitRules}]
Every $n$-ary reversible function can be implemented by a unique $n$-bit reversible circuit in canonical form based on $\mathbb{H}$.
\end{proposition}

\begin{proof}[Proof of~\cref{prop:complete-safe-qfreenormal}]
Recall that a \qfree circuit $G$ uses clean ancillas $\ola$ safely if
\[
  \sem{G} \ket{e_i}_{\olq} \ket{0}_{\ola}
  = \ket{e_j}_{\olq} \ket{0}_{\ola}
\]
for every computational basis state $\ket{e_i}_{\olq}$.

\paragraph{($\Rightarrow$).}
If here exists a circuit $G'$ satisfying conditions (1) and (2), then condition~(2) immediately implies that $G'$ uses $\ola$ safely. By condition~(1), $\sem{G'} = \sem{G}$, hence $G$ also uses $\ola$ safely.

\paragraph{($\Leftarrow$).}
Assume that $G$ uses clean ancillas $\ola$ safely. Suppose $G$ acts on $m \ge 1$ working qubits $\olq$ and $n \ge 1$ clean ancillas $\ola$. Let $\ket{e_i}_{\olq}$, $i = 0, \dots, 2^m-1$, and $\ket{f_j}_{\ola}$, $j = 0, \dots, 2^n-1$, denote the computational basis states. We also $e_i$ or $f_j$ to denote the bit string of $i$ or $j$ and write $e_i f_j$ for the concatenation of the corresponding bit strings.

Let
\[
  \mathbb{H}_{\olq} = (a_0, a_1, \dots, a_{2^m-1})
\]
and
\[
  \mathbb{H}_{\ola} = (b_0 = f_0, b_1, \dots, b_{2^n-1})
\]
be two Hamiltonian paths of the $m$- and $n$-hypercube graphs, respectively. Then the interleaved sequence
\begin{align*}
  \mathbb{H}_{\olq,\ola} = (& a_0 f_0, \dots, a_{2^m-1} f_0, \\
  & a_{2^m-1} b_1, \dots, a_0 b_1, \\
  & a_0 b_2, \dots, a_{2^m-1} b_2, \\
  & \vdots \\
  & a_{2^m-1} b_{2^n-1}, \dots, a_0 b_{2^n-1})
\end{align*}
is a Hamiltonian path of the $(m+n)$-hypercube.

Let $\Delta_{\mathbb{H}_{\olq,\ola}} = \{M_0, M_1, \dots, M_{2^{m+n}-2}\}$. Since $G$ uses $\ola$ safely, the permutation induced by $G$ along $\mathbb{H}_{\olq,\ola}$ has the form
\[
  (a'_0 f_0, \dots, a'_{2^m-1} f_0, \dots)_{\mathbb{H}_{\olq,\ola}},
\]
where all remaining basis states have ancilla components different from $f_0$.

By~\cite{completeReversibleCircuitRules}, there exists a unique circuit $G'$ in canonical form based on $\mathbb{H}_{\olq,\ola}$ such that $\sem{G'} = \sem{G}$.

It remains to show that $G'$ satisfies condition~(2). Write $G'$ as
\[
  C_r C_{r-1} \dots C_{s+1} C_s \dots C_1 C_0,
\]
where the \emph{tail part} $C_s \dots C_0$ consists of gates exchanging only states in the first block (with ancilla state $f_0$), and the \emph{head part} $C_r \dots C_{s+1}$ contains no such gates and no gate exchanging the last element of the first block with the first element of the second block.

All gates in the tail part are MPMCX gates targeting working qubits with $\ola \subseteq N$. These gates can be implemented only using MCX gates that target working qubits, and hence satisfy condition~(2).

For the head part, MPMCX gates targeting working qubits are handled similarly. Consider an MPMCX gate in the head part whose target is an ancilla. We claim that at least one other ancilla must appear as a positive control, which enables the gate to be implemented only using MCX gates controlled by this other ancilla. Otherwise, if all other ancillas appear only as negative controls, the gate would exchange the last element of the first line with the first element of the second line, contradicting the condition of the head part.
\end{proof}

Now it suffices to prove \cref{prop:complete-qfreenormal} based on \cref{prop:complete-safe-qfreenormal}.

\begin{proposition}
A circuit $G$ is uncomputable if and only if there exists a \qfree circuit $G'$ over the same set of quantum variables such that:
\begin{enumerate}
\item $\sem{G} = \sem{G'}$.
\item For any $\tMCX{i}$ gate in $G'$ whose target is an ancilla in $\ola$ and uses no ancilla as a control qubit, every gate appearing to its right is also a $\tMCX{i}$ gate targeting (possibly different) ancillas in $\ola$.
\end{enumerate}
\end{proposition}

\begin{proof}
  ($\Rightarrow$).
  Similar to~\cref{appendix:proof-prop-qfreenormal}, condition~(2) implies that $G'$ is uncomputable.
  By condition~(1), we have $\sem{G'} = \sem{G}$, and hence $G$ is also uncomputable.

  ($\Leftarrow$).
  Assume that $G$ is uncomputable.
  One possible uncomputation procedure is to append, to the right of $G$, a circuit $C$
  consisting solely of MCX gates targeting ancilla qubits, yielding a circuit $G'$
  that uses the clean ancillas $\ola$ safely.

  By~\cref{prop:complete-safe-qfreenormal}, there exists a circuit $G''$ such that
  $\sem{G''} = \sem{G'}$ and $G''$ satisfies condition~(2) in~\cref{prop:complete-safe-qfreenormal}.
  Appending the inverse circuit $C^\dagger$ to the right of $G''$, we obtain
  the circuit $G''; C^\dagger$, which satisfies $\sem{G''; C^\dagger} = \sem{G}$. Moreover, $G''; C^\dagger$ satisfies condition~(2) here.
\end{proof}

\subsection{Proof of \cref{prop:qfreenormal}}\label{appendix:proof-prop-qfreenormal}

\begin{proof}[Proof of \cref{prop:qfreenormal}]
Note that the gate family ${\tMCX{i}}_{i \ge 0}$ (multi-controlled $X$ gates with arbitrary numbers of controls) forms a universal gate set without the need for ancilla bits for \qfree circuits. We prove the two directions separately.

\medskip\noindent\textbf{($\Leftarrow$).} 
Assume a circuit $G'$ satisfying (1) and (2) exists. By (2), we can decompose $G'$ as a composition
\[
    G' \equiv G'_1 \; ; \; G'_2,
\]
where $G'_1$ contains no $\tMCX{i}$ gates whose \emph{targets} lie in $\overline{a}$, and $G'_2$ consists precisely of the $\tMCX{i}$ gates that \emph{do} target bits in $\overline{a}$. (Gates in $G'_1$ may still use ancilla bits in $\overline{a}$ as control lines, but they do not flip those ancilla bits.)

Now consider the action of $G'$ on a basis state $\ket{0}_{\overline{a}}\ket{e_i}_{\overline{q}}$. Since all ancilla bits in $\overline{a}$ are initialized to $0$, any multi-controlled-$X$ gate in $G'_1$ that uses an ancilla bit as a \emph{control} acts trivially (because that control is $0$). Thus we may remove (i.e. cancel) those gates from $G'_1$ when evaluating on inputs with $\overline{a}=\ket{0}$. Denote by $F$ the resulting circuit acting only on the data register $\overline{q}$ (obtained from $G'_1$ after deleting the now-trivial controlled gates). Hence for every basis state $\ket{e_i}_{\overline{q}}$ we have
\[
    G'\ket{0}_{\overline{a}}\ket{e_i}_{\overline{q}}
    \;=\; G'_2 \bigl(\ket{0}_{\overline{a}}\otimes F\ket{e_i}_{\overline{q}}\bigr).
\]
Because $G'_2$ only targets bits in $\overline{a}$, it does not change the content of $\overline{q}$; therefore the final state has the form
\[
    \ket{a'}_{\ola} \otimes F\ket{e_i}_{\overline{q}},
\]
which shows that the action of $G'$ on $\overline{q}$ is given by the (reversible) permutation $F$. In particular, $G'$ is uncomputable and the ancillas $\ola$ can be uncomputed by reversing $G'_2$. By condition (1) the original circuit $G$ has the same input–output behaviour on states with ancillas $\ket{0}_{\ola}$, hence $G$ is uncomputable as well.

\medskip\noindent\textbf{($\Rightarrow$).} 
Now assume $G$ is uncomputable. Work in the computational-basis picture: the reversible boolean function implemented by $G$ maps classical input $(a,x)\in\{0,1\}^m\times\{0,1\}^n$ to some output $(f(a,x),\,h(a,x))$, where $f:\{0,1\}^{m+n}\to\{0,1\}^m$ describes the final ancilla contents and $h:\{0,1\}^{m+n}\to\{0,1\}^n$ describes the final data-register contents. In particular, when the ancilla are initialized to $\ket{0}_{\ola}$ we have
\[
    G\ket{0}_{\overline{a}}\ket{e_i}_{\overline{q}}
    \;=\; \ket{f(0,e_i)}_{\overline{a}}\; \ket{h(0,e_i)}_{\overline{q}}.
\]

Uncomputability of $G$ means precisely that the map $e_i \mapsto h(0,e_i)$ on $\{0,1\}^n$ is a permutation (i.e. a reversible function on the data register) — otherwise one could not restore the ancilla to $\ket{0}$ without disturbing $\overline{q}$. Hence the fixed function $g|_{\ola} := h(\overline{0},\cdot)$ is a bijection on $n$-bit strings and therefore can be synthesized by some \qfree circuit $H$ that acts only on $\overline{q}$.

Let $h_0(x):=h(0,x)$ and $f_0(x):=f(0,x)$.
Since $h_0$ is a bijection, define
\[
k(y):=f_0(h_0^{-1}(y)).
\]
Each output bit of $k$ can be expressed in algebraic
normal form (ANF). Thus, we can construct a \qfree
circuit $K$ using only $\tMCX{i}$ gates targeting
ancillas, such that
\[
K:(a,y)\mapsto(a\oplus k(y),y).
\]

Finally, set $G':=H;K$.
By construction, condition (2) is satisfied.
For every working basis input $x$,
\[
(0,x)\xmapsto{H}(0,h_0(x))
\xmapsto{K}(f_0(x),h_0(x))=G(0,x).
\]
Hence condition (1) is satisfied.

\end{proof}

\subsection{Proof of \cref{prop:quantumnormal}}

\begin{proof}[Proof of \cref{prop:quantumnormal}]\label{appendix:proof-prop-quantumnormal}
Assume there exists a circuit $G'$ satisfying conditions (1) and (2). It suffices to prove $G'$ is uncomputable.
Decompose $G'$ into two sequential components by (2):
\[
    G' = G'_1 \; ; \; G'_2,
\]
where
\begin{itemize}
    \item $G'_1$ contains all gates that \emph{do not target ancilla bits} in $\overline{a}$ (they may still act on $\overline{q}$ and use ancillas as controls),
    \item $G'_2$ contains all gates targeting ancilla bits in $\overline{a}$.
\end{itemize}

Consider the action of $G'_1$ on an arbitrary initial state $\ket{0}_{\overline{a}} \ket{\varphi}_{\overline{q}}$.
Any multi-controlled gate in $G'_1$ that uses ancilla qubits as controls acts trivially, since all ancillas are initialized to $\ket{0}$.  
Hence,
\[
G'_1 \ket{0}_{\overline{a}} \ket{\varphi}_{\overline{q}} = \ket{0}_{\overline{a}} \sum_j \lambda_j \ket{e_j}_{\overline{q}}, \sum_{j} \lambda_j^* \lambda_j = 1.
\]

Next, $G'_2$ targets only ancilla qubits, so it does not affect the data register $\overline{q}$.  
By condition (2), no Hadamard gate acts on ancillas, so $G'_2$ consists solely of classical or diagonal operations on $\overline{a}$.  
Therefore, we can write
\[
G'_2 \left( \ket{0}_{\overline{a}} \sum_j \lambda_j \ket{e_j}_{\overline{q}} \right) 
= \sum_j \lambda_j \, e^{i \theta_j} \ket{a_j}_{\overline{a}} \ket{e_j}_{\overline{q}},
\]
where each $\ket{a_j}_{\overline{a}}$ is a computational basis state of $\overline{a}$ and $\theta_j \in \mathbb{R}$.

By Definition~\ref{def:cleanuncomp}:
\[
\sum_j \lambda_j \, e^{i \theta_j} \ket{a_j}_{\overline{a}} \ket{e_j}_{\overline{q}} \xmapsto{\cG'} \ket{0}_{\overline{a}} \sum_j \lambda_j \, e^{i \theta_j} \ket{e_j}_{\overline{q}}.
\]

Since 
\[
\sum_i | \lambda_i e^{i \theta_i} |^2 = \sum_i |\lambda_i|^2 = 1,
\]
the resulting state is always valid, implying the existence of $\cG'$.
\end{proof}

\subsection{Proof of \cref{rule:R-1,rule:R-2,rule:R-3}}

\begin{proof}[Proof of Rewrite Rules]
\cref{rule:R-1} is obvious and \cref{rule:R-3} is inverse of \cref{rule:R-2}. It suffices to prove \cref{rule:R-2}.

Partition the control sets as
\[
\olq = \olq_1 \cup \olq_2, \quad 
\olp = \olq_2 \cup \{a\} \cup \olq_3,
\]
where $\olq_2 = \olq \cap \olp$, $\olq_1 = \olq \setminus \olp$, $\olq_3 = \olp \setminus (\olq \cup \{a\})$.  

Define
\[
x := \bigwedge_{i \in \olq_1} x_i, \quad 
y := \bigwedge_{i \in \olq_2} y_i, \quad 
u := \bigwedge_{i \in \olq_3} u_i
\]
for $x, y, z$ appears together with $\land$ for simplicity.

Let the initial computational basis state be
\[
\ket{\overline{x}}_{\olq_1} \ket{\overline{y}}_{\olq_2} \ket{z}_a \ket{\overline{u}}_{\olq_3} \ket{v}_t
\]

1. Apply $\tMCX{m}[\olq, a]$: flips $a$ if all qubits in $\olq_1 \cup \olq_2$ are 1:
\[
a' = a \oplus (x \land y)
\]

2. Apply $\tMCX{n}[\olp, t]$: flips $t$ if all controls in $\olq_2 \cup \{a\} \cup \olq_3$ are 1:
\[
t_\text{LHS} = t \oplus (y \land a' \land u) = t \oplus (y \land (a \oplus x \land y) \land u)
\]

1. Apply $\tMCX{n}[\olp, t]$ first: flips $t$ if all controls in $\olq_2 \cup \{a\} \cup \olq_3$ are 1:
\[
t_1 = t \oplus (y \land a \land u)
\]

2. Apply $\mathtt{C^h NOT}[(\olq \cup \olp)/a, t]$, controlled on $\olq_1 \cup \olq_2 \cup \olq_3$:
\[
t_{RHS} = t_2 = t_1 \oplus (x \land y \land u) = t \oplus (y \land a \land u) \oplus (x \land y \land u)
\]

3. Apply $\tMCX{m}[\olq, a]$: flips $a$ if $\olq_1 \cup \olq_2$ are all 1, $a_\text{final} = a \oplus x \land y$.

Then the final $t$ value is
\[
t_\text{LHS} = t \oplus (y \land a \land u) \oplus (x \land y \land u) = t_\text{RHS}.
\]

Hence LHS and RHS produce identical outputs for all computational basis states, proving \cref{rule:R-2}.

\end{proof}

\subsection{Proof of \cref{thm:normalize-soundness}}

The termination argument for \cref{alg:rewrite-normalization} is based on three observations.
First, the algorithm always processes the first violating pair, so the global execution can be decomposed into a sequence of local subproblems.
Second, newly introduced gate in \cref{rule:R-2} targets the same working qubit, and the newly introduced gate in \cref{rule:R-3} targets the same ancilla.
This is the key in the inductive proof of \cref{lem:termination-one-pair}.
Third, each local computation introduces only finitely many gates, because the compensating gates generated by R-2 and R-3 have restricted targets and control dependencies, as established below.

\begin{lemma}\label{lem:termination-one-pair}
Given a quantum circuit $G$ acting on ancillas $\ola$ and working qubits $\olq$,
if $G$ contains exactly one violating pair and that pair is located at the end of $G$,
then \cref{alg:rewrite-normalization} always terminates.
\end{lemma}

\begin{proof}
If some $\tH$ gate acts on an ancilla, then the algorithm immediately reports failure,
so termination is trivial.

Assume therefore that no $\tH$ gate acts on any ancilla initially.
Since $G$ contains exactly one violating pair and that pair is at the end of the circuit,
we may write
\[
G = \overline{W}; A_0; A_1; \dots; A_m; W,
\]
where:
\begin{itemize}
    \item $\overline{W}$ consists only of working-targeting gates,
    \item each $A_i$ is an ancilla-targeting gate, and
    \item $A_m;W$ is the unique violating pair.
\end{itemize}
By definition of \cref{alg:rewrite-normalization}, the prefix $\overline{W}$
is never touched during the procedure. Hence it suffices to consider the suffix
\[
A_0; A_1; \dots; A_m; W.
\]

We proceed by cases on $|\ola|$.

\paragraph{Case 1: $|\ola|=0$.}
This is trivial.

\paragraph{Case 2: $|\ola|=1$.}
Let $a$ denote the unique ancilla.
We prove termination by induction on $m$.

\smallskip
\noindent

\emph{Base case: $m=0$.}
The local block consists only of $A_0;W$.
If a rewriting rule applies, the resulting circuit
contains no violating pair; otherwise, the algorithm
immediately reports failure. Thus, it terminates.

\smallskip
\noindent
\emph{Inductive step.}
Assume the claim holds for $m$, and consider
\[
A_0;\dots;A_{m-1};A_m;W.
\]

\begin{itemize}
    \item If one of \cref{rule:R-1,rule:R-4,rule:R-5,rule:R-6} applies, then
    \[
    A_0;\dots;A_{m-1};A_m;W
    \;\rightsquigarrow\;
    A_0;\dots;A_{m-1};W;A_m.
    \]
    The gate $A_m$ becomes a suffix that will never again be touched by the procedure.
    Thus it remains to process
    \[
    A_0;\dots;A_{m-1};W,
    \]
    which terminates by the induction hypothesis.

    \item If \cref{rule:R-3} applies, then
    \[
    A_0;\dots;A_{m-1};A_m;W
    \;\rightsquigarrow\;
    A_0;\dots;A_{m-1};W;A';A_m,
    \]
    where both $A'$ and $A_m$ target the unique ancilla $a$.
    Hence $A';A_m$ forms a suffix that will never again be touched,
    and it remains to process
    \[
    A_0;\dots;A_{m-1};W,
    \]
    which terminates by the induction hypothesis.

    \item If \cref{rule:R-2} applies, then
    \[
    A_0;\dots;A_{m-1};A_m;W
    \;\rightsquigarrow\;
    A_0;\dots;A_{m-1};W;W';A_m.
    \]
    Again, $A_m$ becomes part of a suffix that will never again be touched.
    So we first process
    \[
    A_0;\dots;A_{m-1};W;W'.
    \]

    Let $t$ be the target of $W$. Then $W'$ is also a working-targeting gate on $t$.
    The gate $W'$ will not be touched until the prefix
    \[
    A_0;\dots;A_{m-1};W
    \]
    has already been transformed into normal form, say $\overline{W};\overline{A'}$,
    which terminates by the induction hypothesis.

    If that process terminates without failure, then $W'$ is processed next.
    By \cref{rule:R-2}, the newly introduced gate $W'$ does not use the ancilla $a$.
    Hence \cref{rule:R-2} will never be applied to $W'$.
    Moreover, \cref{rule:R-3} cannot be applied to $W'$ either;
    otherwise, the mutual-control pattern would already have arisen while processing
    \[
    A_0;\dots;A_{m-1};W.
    \]
    Therefore, only \cref{rule:R-1,rule:R-4,rule:R-5,rule:R-6}
    can apply during processing $\overline{A'};W'$.
    Since $\overline{A'}$ is finite, this process terminates.
\end{itemize}

Hence the claim holds when $|\ola|=1$.

\paragraph{Case 3: $|\ola|>1$.}
Let $t$ be the target of $W$.
By \cref{rule:R-2,rule:R-3}, every newly introduced
working-targeting gate still targets $t$.

Observe that an ancilla-targeting gate $A'$
introduced by \cref{rule:R-3} cannot introduce
further compensating gates.
Indeed, $A'$ has the same target $a$ as its
parent gate $A$, but does not use $t$ as a control.
Any working-targeting gate approaching $A'$
from the right must first cross $A$.
If this crossing succeeds, the working gate
cannot use $a$ as a control, and hence
commutes with $A'$.

Therefore, only the original ancilla-targeting
gates can introduce compensating gates.
Each working-targeting gate moves leftward
and crosses each original ancilla-targeting gate
at most once.
Since each such crossing introduces at most
one additional working-targeting gate,
the resulting branching process has finite
depth and branching factor.
Consequently, only finitely many gates are
introduced, and the local rewriting terminates.
\end{proof}

\begin{proof}[Proof of \cref{thm:normalize-soundness}]
  If some $\tH$ gate acts on an ancilla in the input circuit, then the algorithm immediately terminates and reports failure.

  Assume now that no $\tH$ gate acts on any ancilla initially.
  Then no rewriting rule introduces an $\tH$ gate acting on an ancilla, so this property remains invariant throughout the process.

  Since the success condition, absence of violating pairs, is exactly condition~(2) of \cref{prop:quantumnormal}, we only need to prove termination here.

  The rewriting procedure processes the first violating pair without affecting the unprocessed suffix. By \cref{lem:termination-one-pair}, each such local computation terminates, either with failure or with a normalized prefix.
  In the latter case, the unprocessed suffix of the original circuit becomes strictly shorter.
  Since the original circuit is finite, the entire procedure terminates.
\end{proof}

\subsection{Proof of \cref{thm:rwun-soundness}}

We can thus establish the correctness of our synthesis method from the normal form for both clean and dirty ancillas in quantum circuits.

\begin{proposition}\label{prop:normal-form-synthesis}
Given a circuit $G'$ satisfying 
the conditions in Proposition~\ref{prop:quantumnormal}:

\textit{Clean uncomputation synthesis}: A circuit $\cG$ satisfying $\sem{\cG} \in \Uncomp(\sem{G}, \ola)$ can be directly synthesized by reversing all gates that target ancillas, in reverse order:
\begin{itemize}
\item For $\tX$ and $\tMCX{i}$ gates targeting an ancilla, append the same gate;
\item For $\tZ$, $\tS$, and $\tT$ gates targeting an ancilla, append nothing.
\end{itemize}

\textit{Dirty uncomputation synthesis}: If Algorithm~\ref{alg:rewrite-phasereorderring} succeeds on $G'$, then removing all gates acting on ancillas yields a circuit $\cG^*$ such that $\sem{\cG^*} \in \Uncomp^*(\sem{G}, \ola)$.
\end{proposition}

\begin{proof}
1. As shown in Proof~\ref{appendix:proof-prop-quantumnormal}, the circuit $G'$ can be decomposed as 
$G' \equiv G'_1 ; G'_2$, 
where $G'_2$ consists solely of classical or diagonal operations on the ancillas.  
The overall matrix representation of $G'_2$ can therefore be expressed as a product $PD$, 
where $P$ is a permutation matrix that does not affect $\olq$, 
and $D$ is a diagonal phase matrix.  
Since permutation matrices commute with diagonal phase matrices up to a reordering of phases, 
we have $PD = D'P$ with $D'$ still diagonal.  
Hence, the classical permutation component of $G'_2$ is preserved.

Reversing the gates in $G'_2$ in reverse order correctly inverts the permutation while keeping the diagonal phase unchanged, which produces a valid uncomputation circuit $\cG$.

2. If phase reordering succeeds, the ancilla-update
tail has the form $D';P$, where $D'$ consists of
phase-type gates in $\mathcal P$ and $P$ consists
only of classical gates targeting ancillas.
To evaluate $\func(\sem{G},\ola)$, initialize the
ancillas in $\ket{0}$.
The prefix leaves them unchanged; therefore,
all gates in $D'$ involving ancillas act as the
identity, since the supported phase-type gates
have no effect when the involved ancillas are
in $\ket{0}$. The phase gates acting only on
working qubits remain.
The suffix $P$ changes only ancillas, so summing
over their output basis states removes its
contribution to the extracted functionality.
Consequently, deleting every gate involving
ancillas yields exactly $\func(\sem{G},\ola)$
on the working register.
Since the resulting circuit does not act on
the dirty ancillas, it implements
$I_{\ola}\otimes\func(\sem{G},\ola)$,
as required. This deletion need not preserve
the original unitary on arbitrary ancilla states.
\end{proof}

\begin{proof}[Proof of \cref{thm:rwun-soundness}]
By~\cref{thm:normalize-soundness}, if the output circuit $\RwUnx(G)$ contains no $\bborrow$ statements, this indicates that the normalization of the original input circuit in \NWsynth has succeeded.
Consequently, we have $\sem{G} \ne \bot$.

Then our goal is to prove $\sem{\Rwsynth(G)} \ket{x}_{\old} \ket{0}_{\ola} \ket{\varphi}_{\olq} = \ket{x}_{\old} \ket{0}_{\ola} \sem{G}\ket{\varphi}_{\olq}$ by structural induction on the program $G$.

\textbf{Base cases.}
For the trivial programs $\bskip$ and a single gate $U[\olq]$ the statement is immediate.

\textbf{Inductive step for composition.}
If $G = C_1; C_2$, then by the non-overlapping structure of \rewriteborrow statements, the sets of ancillas used in $C_1$ and $C_2$ are disjoint.  
Let $\overline{d}$ and $\ola$ denote the $m_1$-qubit dirty ancilla register and $n_1$-qubit clean ancilla register used in $C_1$, and let $\overline{e}$ and $\overline{b}$ denote the $m_2$-qubit dirty and $n_2$-qubit clean ancilla registers used in $C_2$.  
Let $\olq$ be the set of working qubits that are shared by both subcircuits. 
By inductive hypothesis, let $\Rwsynth(C_1; C_2) = \cC_1; \cC_2$, where $\cC_1$ denotes $\Rwsynth(C_1)$ and $\cC_2$ denotes $\Rwsynth(C_2)$. Then we have: 
\begin{align*}
\sem{\cC_1; \cC_2} \ket{x}_{\overline{d}} \ket{0}_{\ola} \ket{y}_{\overline{e}} \ket{0}_{\overline{b}} \ket{\varphi}_{\olq} 
&= \sem{\cC_2} \sem{\cC_1} \ket{x}_{\overline{d}} \ket{0}_{\overline{a}} \ket{y}_{\overline{e}} \ket{0}_{\overline{b}} \ket{\varphi}_{\olq}    \\
&=  \sem{\cC_2} \ket{x}_{\overline{d}} \ket{0}_{\overline{a}} \ket{y}_{\overline{e}} \ket{0}_{\overline{b}} \sem{C_1} \ket{\varphi}_{\olq} \\
&=  \ket{x}_{\overline{d}} \ket{0}_{\overline{a}} \ket{y}_{\overline{e}} \ket{0}_{\overline{b}} \sem{C_2} \sem{C_1} \ket{\varphi}_{\olq} \\
&=  \ket{x}_{\overline{d}} \ket{0}_{\overline{a}} \ket{y}_{\overline{e}} \ket{0}_{\overline{b}} \sem{C_1; C_2} \ket{\varphi}_{\olq}
\end{align*}

\textbf{Inductive step for quantum conditional.}
Simliar with above.

\textbf{Case of a borrow clean.}
By induction hypothesis, let $\Rwsynth(\borrowcc{a}{C}) = \NWsynth(\cC, a)$, where $\cC$ denotes $\Rwsynth(C)$. Then we have:
\begin{align*}
\sem{\cC} \ket{x}_{\overline{d}} \ket{0}_{\ola} \ket{0}_b \ket{\varphi}_{\olq} 
&= \ket{x}_{\overline{d}} \ket{0}_{\overline{a}} \sem{C} \ket{0}_b \ket{\varphi}_{\olq}
\end{align*}
Then by \cref{prop:normal-form-synthesis},
\begin{align*}
\NWsynth(\cC, a) \ket{x}_{\overline{d}} \ket{0}_{\ola} \ket{0}_b \ket{\varphi}_{\olq} 
&= \cG \ket{x}_{\overline{d}} \ket{0}_{\overline{a}}  \ket{0}_b \ket{\varphi}_{\olq} \\
&= \ket{x}_{\overline{d}} \ket{0}_{\overline{a}}  \ket{0}_b \sem{\borrowcc{a}{C}} \ket{\varphi}_{\olq},
\end{align*}
where $\cG$ is an arbitrary circuit such that $\sem{\cG} \in \Uncomp(\cC, a)$.

\textbf{Case of a borrow dirty.}
Similar with above.
\end{proof}

\subsection{Proof of~\cref{prop:sufficient_rwun_success}}\label{appendix:proof-prop-sufficient_rwun_success}

\begin{lemma}\label{lem:rule-donot-introduce-cycle}
  Given a \qfree circuit $C$ and its corresponding acyclic circuit graph $G$,
  suppose that all target edges of working qubits in $G$ are modified to be bidirected.
  If $G$ contains no aw-cycle, then none of the three rewriting rules introduces any aw-cycle.
\end{lemma}
\begin{proof}[Proof]
    We consider each rewriting rule separately.

    \paragraph{Case~1.}
    $\tMCX{m}[\olq, a]; \tMCX{n}[\olp, t]
    \equiv
    \tMCX{n}[\olp, t]; \tMCX{m}[\olq, a]$,
    where $t \notin \olq$ and $a \notin \olp$.

    Swapping two gates that share only common control qubits does not change the circuit graph. This has already been shown in Unqomp~\cite{unqomp}.

    \paragraph{Case~2.}
    $\tMCX{m}[\olq, a]; \tMCX{n}[\olp, t]
    \equiv
    \tMCX{n}[\olp, t];
    \mathtt{C^hNOT}[(\olq \cup \olp)\setminus\{a\}, t];
    \tMCX{m}[\olq, a]$,
    where $a \in \olp$ and $t \notin \olq$.

    We instantiate the left-hand side of the rule as shown below, where
    $a$ is an ancilla qubit and $t$ is a working qubit.
    The ancilla $a$ is used as a control in the second gate.

    We partition the control qubits as follows:
    $\overline{q_1}$ are controls appearing only in the first gate,
    $\overline{q_2}$ are common controls,
    and $\overline{q_3}$ are controls appearing only in the second gate.

    \[
      \centering
    \begin{minipage}[t]{0.50\textwidth}
    \vspace{0pt}
    \centering
    \begin{quantikz}[row sep = 0.2cm]
      \lstick{$\overline{q_1} = \olq \setminus \olp$} & \qwbundle{}  & \ctrl{2} & \qw       &   \\
      \lstick{$\overline{q_2} = \olq \cap \olp$} & \qwbundle{}  & \ctrl{1} & \ctrl{3}  &   \\
      \lstick{$a$}            & \qw          & \targ{}  & \ctrl{2}  &   \\
      \lstick{$\overline{q_3} = \olp \setminus \olq$} & \qwbundle{}  & \qw      & \ctrl{1}  &   \\
      \lstick{$b$}            & \qw          & \qw      & \targ{}   &   
    \end{quantikz}
    \end{minipage}\hfill
    \begin{minipage}[t]{0.42\textwidth}
    \vspace{0pt}
    \centering
    \includegraphics[height=2.5cm]{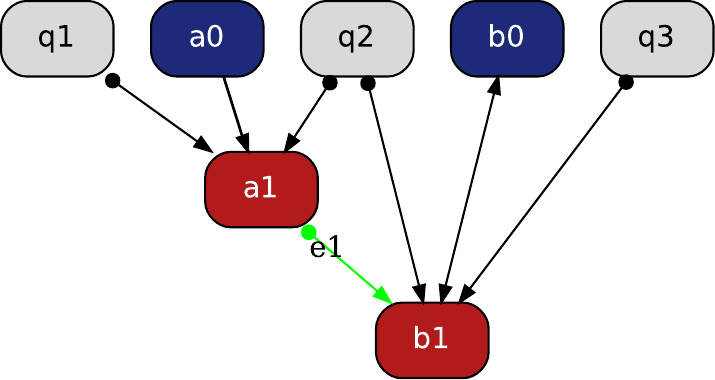}
    \end{minipage}
    \]

    In the circuit graph, we denote by $a_0$ and $a_1$ the states of $a$ before and after the first gate, and similarly for $b_0$ and $b_1$.
    All target edges of working qubits are taken to be bidirected. We call the above circuit graph the \emph{original} one, which will be mentioned in the following proof.

    We first remove the ancilla $a$ from the control set of the second gate,
    and accordingly remove the corresponding control edge in the circuit graph.
    The resulting circuit and circuit graph are shown below.

    \[
    \centering
    \begin{minipage}[t]{0.50\textwidth}
    \vspace{0pt}
    \centering
    \begin{quantikz}[row sep = 0.2cm]
      \lstick{$\overline{q_1} = \olq \setminus \olp$} & \qwbundle{}  & \ctrl{2} & \qw       &   \\
      \lstick{$\overline{q_2} = \olq \cap \olp$} & \qwbundle{}  & \ctrl{1} & \ctrl{3}  &   \\
      \lstick{$a$}            & \qw          & \targ{}  & \qw  &   \\
      \lstick{$\overline{q_3} = \olp \setminus \olq$} & \qwbundle{}  & \qw      & \ctrl{1}  &   \\
      \lstick{$b$}            & \qw          & \qw      & \targ{}   &   
    \end{quantikz}
    \end{minipage}\hfill
    \begin{minipage}[t]{0.50\textwidth}
    \vspace{0pt}
    \centering
    \includegraphics[height=2.0cm]{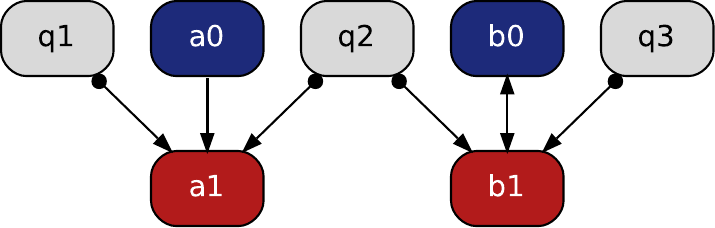}
    \end{minipage}
    \]

    At this point, the two gates share only common control qubits.
Hence, by Case~1, they can be swapped without changing the circuit graph.

    \[
      \centering
    \begin{minipage}[t]{0.50\textwidth}
    \vspace{0pt}
    \centering
    \begin{quantikz}[row sep = 0.2cm]
      \lstick{$\overline{q_1} = \olq \setminus \olp$} & \qwbundle{}   & \qw       & \ctrl{2} &   \\
      \lstick{$\overline{q_2} = \olq \cap \olp$}      & \qwbundle{}   & \ctrl{3}       & \ctrl{1} &   \\
      \lstick{$a$}                                    & \qw           & \qw       & \targ{}  &    \\
      \lstick{$\overline{q_3} = \olp \setminus \olq$} & \qwbundle{}   & \ctrl{1}       & \qw      &   \\
      \lstick{$b$}                                    & \qw           & \targ{}        & \qw      &   
    \end{quantikz}
    \end{minipage}\hfill
    \begin{minipage}[t]{0.48\textwidth}
    \vspace{0pt}
    \centering
    \includegraphics[height=2.0cm]{fig/no_intro_cycle/fig2.pdf}
    \end{minipage}
    \]

    We then reintroduce the control on ancilla $a$ to the first gate,
and add a corresponding control edge $e_1$ to the circuit graph.

    \[
      \centering
    \begin{minipage}[t]{0.50\textwidth}
    \vspace{0pt}
    \centering
    \begin{quantikz}[row sep = 0.2cm]
      \lstick{$\overline{q_1} = \olq \setminus \olp$} & \qwbundle{}   & \qw       & \ctrl{2} &   \\
      \lstick{$\overline{q_2} = \olq \cap \olp$}      & \qwbundle{}   & \ctrl{3}       & \ctrl{1} &   \\
      \lstick{$a$}                                    & \qw           & \ctrl{2}       & \targ{}  &    \\
      \lstick{$\overline{q_3} = \olp \setminus \olq$} & \qwbundle{}   & \ctrl{1}       & \qw      &   \\
      \lstick{$b$}                                    & \qw           & \targ{}        & \qw      &   
    \end{quantikz}
    \end{minipage}\hfill
    \begin{minipage}[t]{0.42\textwidth}
    \vspace{0pt}
    \centering
    \includegraphics[height=3.0cm]{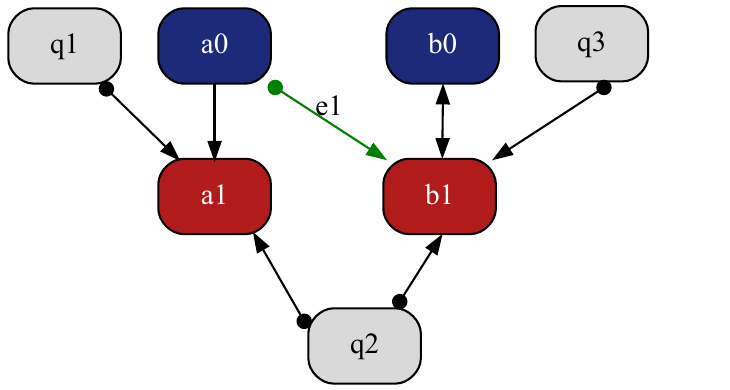}
    \end{minipage}
    \]

    If the addition of $e_1$ introduced an aw-cycle, then $e_1$ must lie on such a cycle. This would imply the existence of a directed path from $b_1$ to $a_0$.
    Since $a_0$ is reachable to $a_1$ by a directed target edge, this yields a directed path from $b_1$ to $a_1$.
    Equivalently, there is a directed path from $b_1$ to $a_1$ in the original circuit graph. Together with the green control edge $e_1$ from $a_1$ to $b_1$ in the original circuit graph, this forms an aw-cycle in the original circuit graph, contradicting the assumption.

    Next, we insert an $\tX$ gate on $b$ in the middle of the two gates.
    In the circuit graph, this corresponds to inserting a new node $b_1'$
    after $b_1$ as shown below (when there is a target node after $b_1$ originally, $b_1'$ is inserted in their middle).

    \[
      \centering
    \begin{minipage}[t]{0.50\textwidth}
    \vspace{0pt}
    \centering
    \begin{quantikz}[row sep = 0.2cm]
      \lstick{$\overline{q_1} = \olq \setminus \olp$} & \qwbundle{} & \qw         &  \qw                   & \ctrl{2} &   \\
      \lstick{$\overline{q_2} = \olq \cap \olp$}      & \qwbundle{} & \ctrl{3}    &  \qw                   & \ctrl{1} &   \\
      \lstick{$a$}                                    & \qw         & \ctrl{2}    &  \qw                   & \targ{}  &    \\
      \lstick{$\overline{q_3} = \olp \setminus \olq$} & \qwbundle{} & \ctrl{1}    &  \qw                   & \qw      &   \\
      \lstick{$b$}                                    & \qw         & \targ{}     &  \targ{}               & \qw      &   
    \end{quantikz}
    \end{minipage}\hfill
    \begin{minipage}[t]{0.42\textwidth}
    \vspace{0pt}
    \centering
    \includegraphics[height=3.0cm]{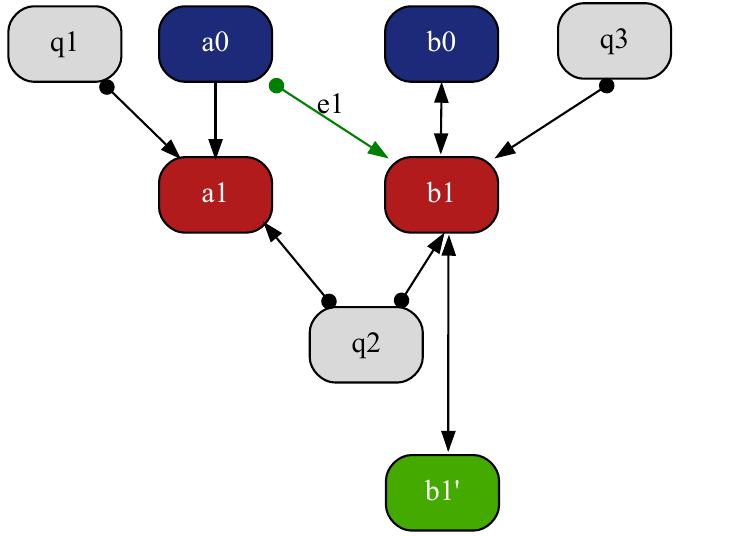}
    \end{minipage}
    \]

    When there is a target node after $b_1$ before insertion, this insertion does not introduce any aw-cycle if there is no aw-cycle before insertion. When there is no target node after $b_1$ before insertion, if this insertion introduce an aw-cycle, then $b_1'$ must be in the aw-cycle. Since $b_1'$ is only connected with $b_1$, this suggests that there exists an aw-cycle involving $b_1$ before insertion, contradicting the assumption.

    We then let $\overline{q_3}$ control the $\tX$ gate and add a control edge $e_2$ from $\overline{q_3}$ to $b_1'$ in the circuit graph.

    \[
      \centering
    \begin{minipage}[t]{0.50\textwidth}
    \vspace{0pt} 
    \centering
    \begin{quantikz}[row sep = 0.2cm]
      \lstick{$\overline{q_1} = \olq \setminus \olp$} & \qwbundle{}  & \qw      & \qw             & \ctrl{2} &   \\
      \lstick{$\overline{q_2} = \olq \cap \olp$}      & \qwbundle{}  & \ctrl{3} & \qw             & \ctrl{1} &   \\
      \lstick{$a$}                                    & \qw          & \ctrl{2} & \qw             & \targ{}  &    \\
      \lstick{$\overline{q_3} = \olp \setminus \olq$} & \qwbundle{}  & \ctrl{1} & \ctrl{1}        & \qw      &   \\
      \lstick{$b$}                                    & \qw          & \targ{}  & \targ{}         & \qw      &   
    \end{quantikz}
    \end{minipage}\hfill
    \begin{minipage}[t]{0.42\textwidth}
    \vspace{0pt}
    \centering
    \includegraphics[height=3.0cm]{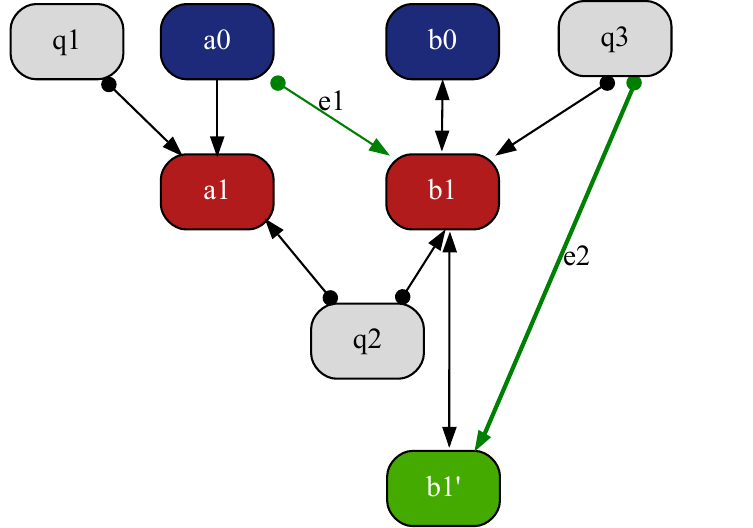}
    \end{minipage}
    \]

    If the addition of $e_2$ introduced an aw-cycle, then $e_2$ must be in that cycle. Replacing $e_2$ by the path from $\overline{q_3}$ to $b_1'$ yields an aw-cycle before addition, again contradicting the assumption.

    Finally, we let $\overline{q_2}$ and $\overline{q_1}$ control the $\tX$ gate, adding edges $e_3$ and $e_4$ to the circuit graph.

    \[
    \centering
    \begin{minipage}[t]{0.50\textwidth}
    \vspace{0pt} 
    \centering
    \begin{quantikz}[row sep = 0.2cm]
      \lstick{$\overline{q_1} = \olq \setminus \olp$} & \qwbundle{}   & \qw        & \ctrl{4}       & \ctrl{2} &   \\
      \lstick{$\overline{q_2} = \olq \cap \olp$}      & \qwbundle{}   & \ctrl{3}   & \ctrl{3}       & \ctrl{1} &   \\
      \lstick{$a$}                                    & \qw           & \ctrl{2}   & \qw            & \targ{}  &    \\
      \lstick{$\overline{q_3} = \olp \setminus \olq$} & \qwbundle{}   & \ctrl{1}   & \ctrl{1}       & \qw      &   \\
      \lstick{$b$}                                    & \qw           & \targ{}    & \targ{}        & \qw      &   
    \end{quantikz}
    \end{minipage}\hfill
    \begin{minipage}[t]{0.42\textwidth}
    \vspace{0pt}
    \centering
    \includegraphics[height=3.0cm]{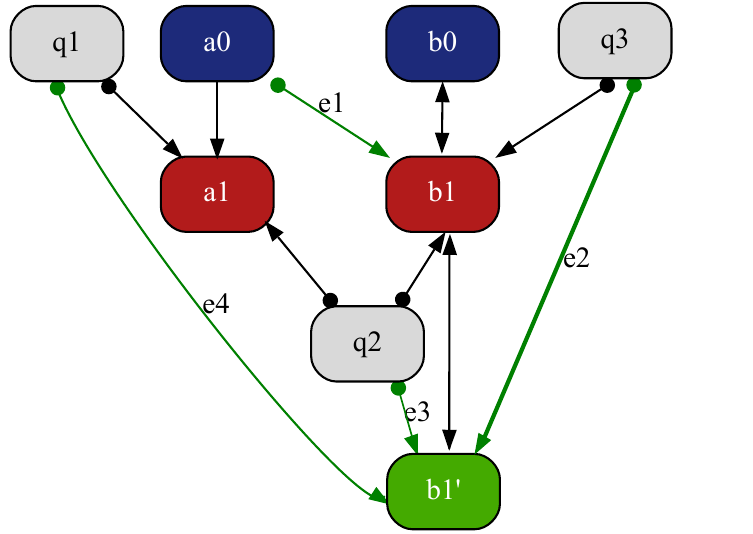}
    \end{minipage}
    \]

    Since both $\overline{q_1}$ and $\overline{q_2}$ are already connected to $b_1'$, these additions cannot introduce any aw-cycle for the same reason.

    All operations---including removing or adding control edges and inserting intermediate nodes---are local. They neither depend on nor affect other parts of the circuit graph not depicted in the figures. Hence, the above reasoning applies to the entire circuit.
    
    \paragraph{Case~3.}
    This case is analogous to Case~2 and can be proved in the same manner.

\end{proof}

\begin{proof}[Proof of \cref{prop:sufficient_rwun_success}]
  At each step, \RwUn selects two consecutive gates $(g_1,g_2)$ such that
  $g_1$ is a $\tMCX{m}$ gate targeting an ancilla qubit and $g_2$ is a $\tMCX{n}$ gate
  targeting a working qubit.
  The only case in which none of the three rewriting rules applies is when the
  subgraph induced by this local configuration forms an aw-cycle in the
  bidirected-target circuit graph.

  By assumption, the initial circuit graph contains no aw-cycle.
  Moreover, by~\cref{lem:rule-donot-introduce-cycle}, applying any of the three rewriting rules
  does not introduce an aw-cycle.
  Therefore, after each rewriting step, the current circuit graph still contains
  no aw-cycle, and hence \RwUn can always apply a rule and never gets stuck.

  Finally, \RwUn terminates and reaches a normal form.
  Hence, \RwUn succeeds on $C$.
\end{proof}

\section{Detailed Evaluation Setting}\label{app:setting-detail}

In this section, we describe  \bAwDep in detail and also provide the parameter of \cref{tab:practical-circuits} in \cref{tab:practical-circuits-parameter} and \cref{tab:practical-circuits-size-depth}.

\paragraph{Base Propagation Score via Alternating Levels.}
We first define a base propagation score that quantifies how modifications on dependency sources further fan out in the circuit graph.
Let $G$ be the circuit graph of a circuit, whose nodes are gate instances and whose edges are of two kinds:
target edges $(\to)$ link consecutive gate nodes on the same qubit wire, and control edges $(\ctlarrow)$ represent control dependencies (from a control wire node to a target wire node).

Fix an ancilla wire $a$ and a gate node $n$ on $a$ (i.e., $n$ is a gate instance targeting $a$).
We define the dependency complexity of $n$ by alternating between two kinds of expansions:
(1) following target edges to enumerate downstream modifications on a wire, and
(2) counting incoming control edges to capture how those modifications are further depended upon by other wires.
Formally, we compute a sequence of node sets $\mathcal{S}_0,\mathcal{S}_1,\dots$ as follows.

\begin{itemize}
\item \textbf{Level 0 (sources).} $\mathcal{S}_0$ contains all source nodes of incoming control edges to $n$ in $G$.
That is, for each control edge $m \ctlarrow n$, we add the source node $m$ to $\mathcal{S}_0$.
\item \textbf{Odd levels (downstream modifications).} For odd $\ell=2k+1$, $\mathcal{S}_\ell$ is obtained by following target edges forward from $\mathcal{S}_{\ell-1}$.
Intuitively, $\mathcal{S}_\ell$ collects the gate nodes that modify the wires we depended on.
\item \textbf{Even levels (fan-out by being depended).} For even $\ell=2k\ge 2$, for each node $u\in \mathcal{S}_{\ell-1}$ we count the number of incoming control edges to $u$,
add this count to the complexity score, and collect the source nodes of those incoming control edges into $\mathcal{S}_{\ell}$.
\end{itemize}

The total dependency complexity of $n$ is the sum of the contributions at all even levels.
In other words, odd levels only \emph{enumerate} modifications on depended wires, while even levels quantify how much those modifications are themselves depended on.
\cref{fig:dep_comp_eg} illustrates the procedure on a concrete example.
For gate node \texttt{a1(CX)}, level~0 collects a single node \texttt{q0(init)}.
At level~1, we follow target edges and collect the forward-reachable gate node \texttt{q1(CX)}.
At level~2, we count the incoming control edges to \texttt{q1(CX)} (one in total), and we collect their sources \texttt{r0(init)}.

For large circuits, this alternation can expand substantially.
We therefore bound the maximum depth by a hyperparameter $L$ (\textbf{L=10} used in \cref{sec:evaluation}), and compute the score up to level $L$.

\begin{remark}[Why odd-level modifications alone should not count.]
A key subtlety is that not every downstream modification on a depended qubit truly increases uncomputation difficulty.
If a modification occurs in a region that is independent of the ancilla, then it can be ignored: uncomputation can be performed by inserting the inverse at an earlier point (before the modification), without having to ``track'' it through the rest of the circuit.

Consider the following circuit fragment:
\[
\begin{quantikz}
    \lstick{$r$}            & \qw       & \qw       & \qw      & \qw \\
    \lstick{$q$}            & \ctrl{1}  & \qw       & \targ{}  & \qw \\
    \lstick{$a:=\ket{0}$}   & \targ{}   & \ctrl{1}  & \qw      & \qw \\
    \lstick{$p$}            & \qw       & \targ{}   & \qw      & \qw
\end{quantikz}
\]
Here the ancilla $a$ depends on $q$, and $q$ is later modified by an $\tX[q]$-like operation (the third column on wire $q$).
However, this modification lies in a tail region that never interacts with $a$ again; hence it does not affect the possibility of uncomputing $a$.
Indeed, we can uncompute $a$ by inverting the ancilla-writing part, ignoring the tail on $q$ entirely.
Therefore, such modifications should contribute \emph{zero} to the dependency complexity.

Now compare with a case where the modification on $q$ becomes controlled:
\[
\begin{quantikz}
    \lstick{$r$}            & \qw       & \qw       & \ctrl{1}  &  \qw    \\
    \lstick{$q$}            & \ctrl{1}  & \qw       & \targ{}   &  \qw    \\
    \lstick{$a:=\ket{0}$}   & \targ{}   & \ctrl{1}  & \qw       &  \qw    \\
    \lstick{$p$}            & \qw       & \targ{}   & \qw       &  \qw   
\end{quantikz}
\]
Even here, if the controlled modification on $q$ is still in a tail region that does not interact with $a$, it remains irrelevant to uncomputing $a$.
What makes a downstream modification \emph{effectively relevant} is the existence of a causal ``return path'' back to the ancilla.
For instance, adding a later $\tCNOT[q,a]$ makes the modification on $q$ no longer ignorable:
\[
\begin{quantikz}
    \lstick{$r$}            & \qw       & \qw       & \ctrl{1}  &  \qw       \\
    \lstick{$q$}            & \ctrl{1}  & \qw       & \targ{}   &  \ctrl{1}  \\
    \lstick{$a:=\ket{0}$}   & \targ{}   & \ctrl{1}  & \qw       &  \targ{}   \\
    \lstick{$p$}            & \qw       & \targ{}   & \qw       &  \qw   
\end{quantikz}
\]
Now the modification on $q$ can influence the later interaction with $a$, and it indeed increases uncomputation difficulty.
\end{remark}

\paragraph{Cycle-Gated Source Collection.}
Motivated by the above, \bAwDep should only count modifications that can propagate back to the ancilla.
We enforce this by refining Level~0: instead of collecting \emph{all} sources $m$ of incoming control edges $m\ctlarrow n$, we only collect those sources whose downstream effects can return to the ancilla.

Concretely, for each ancilla gate node $n$, we build a modified graph $G_n$ from the circuit graph $G$ as follows:
(i) reverse every control edge incident to $n$ (both edges into $n$ and edges out of $n$);
(ii) make every target edge on ancilla wire $a$ that occurs \emph{after} $n$ bidirected.
We then say that a dependency source node $m$ is \emph{effective} for $n$ if $G_n$ contains a directed cycle that includes both $m$ and $n$
(equivalently, $m$ and $n$ lie in the same nontrivial strongly connected component in $G_n$).
Intuitively, this condition witnesses that changes originating from $m$ can, after our local rewiring at $n$, propagate forward and eventually return to affect the ancilla again.
Only such sources are included in $\mathcal{S}_0$ for $n$. Take \cref{fig:dep_comp_eg2} as an example.

Finally, the dependency complexity of an ancilla wire is defined as the sum of the (cycle-gated) dependency complexities of all its gate nodes,
and \bAwDep of a circuit is the maximum complexity among all ancilla wires.

\begin{figure}[t]
  \centering
  \begin{minipage}[t]{0.48\textwidth}
    \centering
    \includegraphics[width=\linewidth]{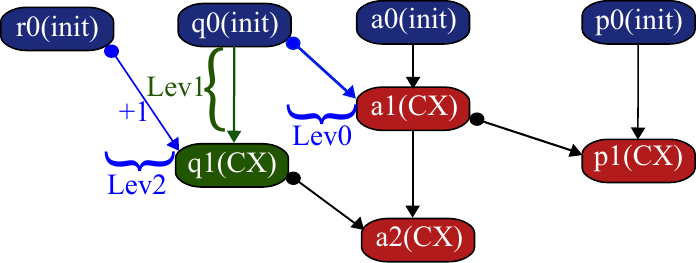}
    \caption{The dependency-complexity of node \texttt{a1(CX)} is $1$, computed by alternating positive \emph{even} (blue) levels from shallow to deep.}
    \label{fig:dep_comp_eg}
  \end{minipage}\hfill
  \begin{minipage}[t]{0.48\textwidth}
    \centering
    \includegraphics[width=\linewidth]{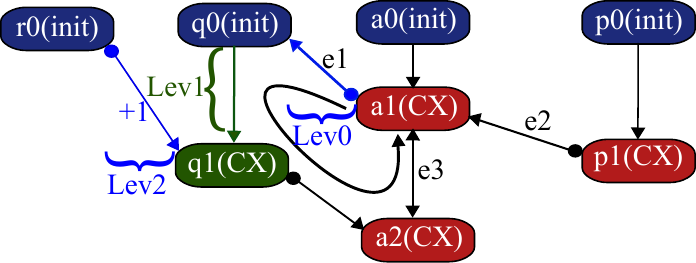}
    \caption{The circuit graph adapted from \cref{fig:dep_comp_eg}, with control edges $e_1, e_2$ obtained by reversing, and target edge $e_3$ being bidirected. The black cycle demonstrates that the modification of \texttt{q1(CX)} propagates back to \texttt{a1(CX)}.}
    \label{fig:dep_comp_eg2}
  \end{minipage}
\end{figure}

\begin{table}[tbp]
    \centering
    \small
    \caption{Scale parameter (n) of \cref{tab:practical-circuits}. `Iter' is the iteration count. }
    \label{tab:practical-circuits-parameter}
    \begin{adjustbox}{max width=\textwidth}
    \begin{tabular}{l p{0.75\linewidth}}
        \toprule
        \textbf{Circuit} & \textbf{Scale Parameter ($n$)} \\
        \midrule
        Clean IntegerComparator &
         $n$ controls $+\,1$ target $+\,(n-1)$ ancillas; comparing to the $n$-bit string $(010101\dots)_2$ \\
        Clean MCX &
         $n$ controls $+\,1$ target $+\,(n-2)$ ancillas \\
        MCRY($\pi/2$) &
         $n$ controls $+\,1$ target $+\,1$ ancilla \\
        Clean Incrementer &
         $n$ working qubits $+\,(n-2)$ ancillas \\
        Deutsch--Jozsa &
         $n$ control qubits $+\,1$ target $+\,1$ ancilla; oracle MCX; returns true iff the value is $11\dots 11$ \\
        Grover's algorithm &
         $n$ control qubits $+\,1$ target $+\,(n-2)\cdot \mathrm{Iter}$ ancillas; oracle MCX; returns true iff the value is $11\dots 11$ \\
        PiecewiseLinearR &
         $n$ state qubits $+\,1$ target; $k$ piecewise-linear segments, with $k$ sampled from $\{2,\ldots,\min(2^n,100)\}$; distinct breakpoints sampled from $\{0,\ldots,2^n-1\}$, slopes from $\{0,\ldots,100\}$, and offsets from $\{0,\ldots,200\}$ (seed 42); fresh comparator ancillas for each nonzero breakpoint \\
        PolynomialPauliR &
         $n$ state qubits $+\,1$ target; a degree-$(n-1)$ polynomial with all $n$ coefficients equal to 2; a fresh ancilla for each rotation with at least two controls \\
        WeightedAdder &
         $n$ state qubits $+\,s$ sum qubits, where $s=\lfloor\log_2\sum_i w_i\rfloor+1$ for $\sum_i w_i>0$; weights $w_i$ sampled from $\{0,\ldots,10\}$ (seed 42); fresh carry and MCX ancillas for each addition. For $\sum_i w_i=0$, $s=1$ \\
        Dirty MCX &
         $n$ controls $+\,1$ target $+\,(n-2)$ ancillas \\
        Clean Adder &
         $2n$ working qubits $+\,n$ ancillas \\
        Dirty IntegerComparator &
         $n$ controls $+\,1$ target $+\,n$ ancillas; comparing to the $n$-bit string $(010101\dots)_2$ \\
        HighestBitConstAdder &
         $n$ working qubits $+\,(n-1)$ ancillas; adding the constant $2^n-1$ \\
        RiseConditionalCleanMCS &
         $n$ controls $+\,1$ target $+\,2$ ancilla \\
        ConditionalCleanMCS &
         $n$ controls $+\,1$ target $+\,2$ ancilla \\
        RiseConditionalDirtyMCS &
         $n$ controls $+\,1$ target $+\,2$ ancilla \\
        ConditionalDirtyMCS &
         $n$ controls $+\,1$ target $+\,2$ ancilla \\
        Multiplier &
         $n$ qubits per operand; $3n$ working qubits $+\,2n^2$ ancillas; computes $b\gets b+xy\pmod{2^n}$ \\
        Gidney's Incrementer &
         $n$ working qubits $+\,n$ ancillas \\
        Dirty Adder &
         $2n$ working qubits $+\,(n-1)$ ancillas \\
        Dirty Incrementer &
         $n$ working qubits $+\,(n-2)$ ancillas \\

        \bottomrule
    \end{tabular}
    \end{adjustbox}
\end{table}

\newcommand{\SD}[2]{\ensuremath{#1/#2}} 

\begin{table}[tbp]
    \centering
    \small
    \setlength{\tabcolsep}{4pt}
    \renewcommand{\arraystretch}{1.15}
    \caption{Input-circuit size/depth (\textbf{S/D}) at the largest solvable scale for each circuit family reported in~\cref{tab:practical-circuits}.
    The \bAwDep column indicates the scale at which \bAwDep is computed; we choose the largest~$n$ such that the circuit size is at most~200.
    For MCRY$(\pi/2)$, ConditionalCleanMCS, and ConditionalDirtyMCS, the circuit size and depth are constant with respect to~$n$: $n$ only affects the number of controls of $\tMCX{f(n)}$, but not the overall circuit size/depth.
    \ding{55} denotes failure.}
    \label{tab:practical-circuits-size-depth}
    \begin{adjustbox}{max width=\textwidth}
    \begin{tabular}{lcccccc}
        \toprule
        \multirow{2}{*}{\textbf{Circuit}} &
        \multirow{2}{*}{\textbf{\bAwDep: expanded S/D }} &
        \multicolumn{4}{c}{\textbf{\RwUnx}: input size/depth \(\mathbf{S}/\mathbf{D}\)} &
        \multirow{2}{*}{\textbf{\Reqompx}: \(\mathbf{S}/\mathbf{D}\)} \\
        \cmidrule(lr){3-6}
        & & \textsc{Sequential} & \textsc{Reverse} & \textsc{Jointly} & \textsc{Lifetime} \\
        \midrule
        Clean IntegerComparator & $\SD{200}{89}$ & $\SD{10}{10}$ & $\SD{130}{130}$ & $\SD{20}{20}$ & $\SD{20}{20}$ & $\SD{180}{180}$ \\
        Clean MCX & $\SD{200}{200}$ & $\SD{29}{29}$ & $\SD{479}{479}$ & $\SD{239}{239}$ & $\SD{999}{999}$ & $\SD{189}{189}$ \\
        MCRY($\pi/2$) & $\SD{9}{8}$ & $\SD{9}{8}$ & $\SD{9}{8}$ & $\SD{9}{8}$ & $\SD{9}{8}$ & $\SD{9}{8}$ \\
        Clean Incrementer & $\SD{200}{100}$ & $\SD{16}{8}$ & $\SD{278}{139}$ & $\SD{176}{88}$ & $\SD{698}{349}$ & $\SD{338}{169}$ \\
        Deutsch--Jozsa & $\SD{199}{67}$ & \ding{55} & \ding{55} & \ding{55} & $\SD{3001}{1001}$ & $\SD{511}{171}$ \\
        Grover's algorithm~\cite{grover} & $\SD{200}{68}$ & \ding{55} & \ding{55} & \ding{55} & $\SD{10951}{2842}$ & $\SD{4772}{1280}$ \\
        PiecewiseLinearR & $\SD{80}{64}$ & \ding{55} & \ding{55} & \ding{55} & \ding{55} & $\SD{109}{109}$ \\
        PolynomialPauliR & $\SD{187}{77}$ & \ding{55} & \ding{55} & \ding{55} & \ding{55} & $\SD{1023}{1023}$ \\
        WeightedAdder & $\SD{158}{56}$ & $\SD{17}{11}$ & $\SD{17}{11}$ & $\SD{17}{11}$ & $\SD{1581}{262}$ & $\SD{4933}{700}$ \\
        Dirty MCX & $\SD{199}{199}$ & $\SD{137}{137}$ & $\SD{137}{137}$ & $\SD{197}{197}$ & $\SD{1997}{1997}$ & $\SD{1997}{1997}$ \\
        Clean Adder & $\SD{200}{38}$ & $\SD{11}{6}$ & $\SD{159}{80}$ & $\SD{17}{9}$ & $\SD{9}{5}$ & $\SD{319}{160}$ \\
        Dirty IntegerComparator & $\SD{200}{172}$ & $\SD{47}{40}$ & $\SD{47}{40}$ & $\SD{68}{58}$ & $\SD{54}{46}$ & $\SD{2976}{2551}$ \\
        HighestBitConstAdder & $\SD{199}{160}$ & $\SD{39}{32}$ & $\SD{34}{28}$ & $\SD{54}{44}$ & $\SD{39}{32}$ & $\SD{4844}{3876}$ \\
        RiseConditionalCleanMCS & $\SD{199}{100}$ & $\SD{39}{20}$ & $\SD{31}{16}$ & $\SD{23}{12}$ & $\SD{31}{16}$ & \ding{55} \\
        ConditionalCleanMCS & $\SD{19}{8}$ & $\SD{19}{8}$ & $\SD{19}{8}$ & $\SD{19}{8}$ & $\SD{19}{8}$ & \ding{55} \\
        RiseConditionalDirtyMCS & $\SD{200}{103}$ & $\SD{40}{23}$ & $\SD{40}{23}$ & $\SD{40}{23}$ & $\SD{56}{31}$ & \ding{55} \\
        ConditionalDirtyMCS & $\SD{32}{13}$ & $\SD{32}{13}$ & $\SD{32}{13}$ & $\SD{32}{13}$ & $\SD{32}{13}$ & \ding{55} \\
        Multiplier & $\SD{145}{34}$ & $\SD{9}{6}$ & $\SD{9}{6}$ & $\SD{9}{6}$ & $\SD{27}{12}$ & $\SD{779}{76}$ \\
        Gidney's Incrementer & $\SD{195}{147}$ & \ding{55} & \ding{55} & $\SD{355}{267}$ & \ding{55} & \ding{55} \\
        Dirty Adder & $\SD{173}{137}$ & $\SD{71}{56}$ & $\SD{71}{56}$ & $\SD{47}{37}$ & $\SD{173}{137}$ & \ding{55} \\
        Dirty Incrementer & $\SD{197}{184}$ & $\SD{82}{74}$ & $\SD{82}{74}$ & $\SD{82}{74}$ & $\SD{7922}{7834}$ & \ding{55} \\

        \bottomrule
    \end{tabular}
    \end{adjustbox}
\end{table}

\end{document}